\documentclass[11pt]{article}
\usepackage{amsmath}
\usepackage{amssymb}
\usepackage{amsthm}

\usepackage{graphicx} %
\usepackage{fullpage}
\usepackage{hyperref}
\usepackage{xcolor}
\usepackage{multicol}
\usepackage{url}
\usepackage[utf8]{inputenc}
\usepackage{enumitem}
\usepackage[margin=1in]{geometry}
\usepackage{cleveref}
\usepackage{doi}
\usepackage{lineno}
\usepackage{wrapfig}

\usepackage{tikz}
\usetikzlibrary{matrix,calc,arrows,positioning,graphs,mindmap}

\title{The Richness of CSP Non-redundancy\footnote{An extended abstract of our work will appear in the proceedings of FOCS 2026.}}%
\author{Joshua Brakensiek\thanks{University of California, Berkeley. Contact: \href{mailto:josh.brakensiek@berkeley.edu}{josh.brakensiek@berkeley.edu}} \and Venkatesan Guruswami\thanks{Simons Institute for the Theory of Computing and the University of California, Berkeley. Contact: \href{mailto:venkatg@berkeley.edu}{venkatg@berkeley.edu}} \and
  Bart M. P. Jansen\thanks{Eindhoven University of Technology, The Netherlands. Contact: \href{mailto:b.m.p.jansen@tue.nl}{b.m.p.jansen@tue.nl}} \and
  Victor Lagerkvist\thanks{Linköping University, Linköping. Contact: \href{mailto:victor.lagerkvist@liu.se}{victor.lagerkvist@liu.se}} \and
  Magnus Wahlström\thanks{Royal Holloway, University of London. Contact: \href{mailto:magnus.wahlstrom@rhul.ac.uk}{magnus.wahlstrom@rhul.ac.uk}}
}

\date{}

\usepackage{xcolor}
\definecolor{olivegreen}{RGB}{107,142,35}

\newcommand{\surj}{\operatorname{surj}}
\newcommand{\ar}{\operatorname{ar}}

\newcommand{\polylog}{\operatorname{polylog}}

\newcommand{\CSP}{\operatorname{CSP}}
\newcommand{\csp}{\operatorname{CSP}}

\newcommand{\Pol}{\operatorname{Pol}}
\newcommand{\tLIN}{\operatorname{3LIN}}

\newcommand{\sat}{\operatorname{sat}}

\newcommand{\pPol}{\operatorname{pPol}}

\newcommand{\pattern}{\operatorname{Pattern}}
\newcommand{\mpattern}{\operatorname{mPattern}}
\newcommand{\AS}{\operatorname{AS}}
\newcommand{\Bal}{\operatorname{Bal}}
\newcommand{\Cat}{\operatorname{Cat}}
\newcommand{\Grp}{\operatorname{Grp}}
\newcommand{\Cox}{\operatorname{Cox}}
\newcommand{\tx}{\tilde{x}}

\newcommand{\cla}{\Gamma_1}
\newcommand{\clb}{\Gamma_2}

\newcommand{\ex}{\operatorname{ex}}

\newcommand{\cI}{\mathcal I}
\newcommand{\cP}{\mathcal P}

\newcommand{\cF}{\mathcal F}

\newcommand{\cH}{\mathcal H}

\newcommand{\N}{\mathbb{N}}
\newcommand{\F}{\mathbb{F}}

\newcommand{\cC}{\mathcal C}

\newcommand{\OR}{\operatorname{OR}}
\newcommand{\CUT}{\operatorname{CUT}}

\newcommand{\eps}{\varepsilon}

\newcommand{\NRD}{\operatorname{NRD}}

\newcommand{\EQ}{\operatorname{EQ}}

\newcommand{\E}{\mathbb E}
\newcommand{\one}{{\bf 1}}

\newcommand{\BCK}{\operatorname{BCK}}
\newcommand{\CYC}{\operatorname{CYC}}
\newcommand{\CYCs}{\CYC^*}

\newcommand{\rep}{\operatorname{rep}}
\newcommand{\oneinthree}{\operatorname{1-in-3}}

\newcommand{\R}{\mathbb{R}}
\newcommand{\Z}{\mathbb{Z}}

\newcommand{\ORDP}{\operatorname{OR-DP}}
\newcommand{\DP}{\operatorname{DP}}
\newcommand{\SATDP}{\operatorname{SAT-DP}}

\newcommand{\interpretation}[2]{\operatorname{\mathit{I}}_{#1}(#2)}

\newtheorem{theorem}{Theorem}
\numberwithin{theorem}{section}
\newtheorem{lemma}[theorem]{Lemma}
\newtheorem{claim}[theorem]{Claim}
\newtheorem{proposition}[theorem]{Proposition}
\newtheorem{corollary}[theorem]{Corollary}
\newtheorem{observation}[theorem]{Observation}

\newtheorem{conjecture}[theorem]{Conjecture}

\theoremstyle{definition}
\newtheorem{definition}[theorem]{Definition}
\newtheorem{remark}[theorem]{Remark}
\newtheorem{example}[theorem]{Example}
\DeclareMathOperator{\inv}{\operatorname{Inv}}

\DeclareMathOperator{\dom}{\operatorname{dom}}
\newcommand{\qfppp}[1]{\ensuremath{\langle #1 \rangle_{\not \exists}}}
\newcommand{\minor}[1]{\ensuremath{[#1]_{\mathrm{min}}}}
 
\makeatletter
\renewcommand{\paragraph}{%
  \@startsection{paragraph}{4}%
  {\z@}{6pt \@plus 1pt \@minus 1pt}{-5pt}%
  {\normalfont\normalsize\bfseries}%
}
\makeatother

\usepackage{xspace}

\newcommand{\idx}{\mathsf{idx}}

\newcommand{\Oh}{\mathcal{O}}

\newcommand{\ncontainment}{\ensuremath{\mathsf{NP \not\subseteq coNP/poly}}\xspace}
\newcommand{\containment}{\ensuremath{\mathsf{NP  \subseteq coNP/poly}}\xspace}

\newcommand{\true}{\emph{true}\xspace}
\newcommand{\false}{\emph{false}\xspace}
\newcommand{\proj}{\operatorname{proj}}
\newcommand{\PAULI}{\operatorname{PAULI}}

\newenvironment{corollaryrestated}[1]{%
  \par\noindent\textbf{Corollary~\ref{#1}.}\itshape\ }{%
  \par\medskip
}

\newenvironment{theoremrestated}[1]{%
  \par\noindent\textbf{Theorem~\ref{#1}.}\itshape\ }{%
  \par\medskip
}

\usepackage{tocloft}
\begin{document}

\maketitle
\thispagestyle{empty}
\begin{abstract}

An instance of a constraint satisfaction problem (CSP) is called non-redundant if no constraint is implied by the rest. The non-redundancy (NRD) of a relation $R$ is the maximum number of constraints in a non-redundant instance of~$\operatorname{CSP}(R)$, as a function of the number of variables~$n$. Recent progress has shown that non-redundancy is crucially linked to many important questions concerning sparsification, kernelization, query complexity, universal algebra, logic, and extremal combinatorics. Given that non-redundancy is a nexus for these diverse topics, the central goal of this paper is to more deeply understand non-redundancy.

\smallskip
The NRD of an arity-$r$ relation lies between $\Omega(n)$ and $O(n^r)$, and when $r=2$, exhibits a sharp linear-versus-quadratic dichotomy. Until recently, all known examples were consistent with NRD exhibiting only polynomial growth $n^c$ for integer exponents $c$. We show that the general behavior of NRD is substantially richer. We prove that  for every rational $s \ge 1$, there exists a relation with NRD $=\Theta(n^s)$. Using connections to high-girth graphs in extremal combinatorics, we further show that infinitely many exponents arise as the NRD of relations even for arity 3.

\smallskip
Inspired by these concrete results, we develop a general algebraic theory of  non-redundancy based on universal algebraic objects known as partial promise polymorphism patterns. As an application of this algebraic theory, we revisit the notion of \emph{Mal'tsev embeddings}, which is the most general technique known to date for establishing that a relation has linear NRD. We resolve several structural questions about Mal'tsev embeddings, via a powerful new tool we develop called \emph{Catalan polymorphisms}. In particular, we give the first example of a relation that admits no Abelian embedding but does embed into the non-Abelian quantum Pauli group.
\end{abstract}

\pagebreak

\setcounter{tocdepth}{2} %
\enlargethispage{0.5cm}
\tableofcontents
\thispagestyle{empty}
\pagebreak
\setcounter{page}{1}

\section{Introduction}

The seminal works of Karger~\cite{DBLP:conf/soda/Karger93} and Benczur and Karger~\cite{DBLP:conf/stoc/BenczurK96} showed that any undirected graph can be sparsified to a weighted subgraph with only $\tilde{\Oh}(n)$ edges so that all cuts are preserved within a constant multiplicative factor.  An active and influential direction in recent years has extended this cut sparsification result to hypergraphs~\cite{kogan2015,chen2020,khanna2024optimal}, including spectral variants~\cite{spielman2011spectral,batson2012twice,soma2019spectral,KKTY2021,lee2023spectral,jambulapati2023chaining} and Cayley graphs~\cite{khanna2024Code,HsiehLMPZ26,BasuKLM26}.

About a decade ago, Kogan and Krauthgamer~\cite{kogan2015} posed the question of \emph{CSP sparsification}, allowing for constraint types other than the cut constraint. More formally, for a relation $R \subseteq D^r$ over a finite domain $D$, given an instance $\cI$ of $\CSP(R)$ with variable set $V$, a sparsified instance $\cI'$ consists of a (suitably reweighted) subset of constraints with the property that, for every assignment $\sigma : V \to D$, the weight of constraints satisfied by $\sigma$ in $\cI'$ is within $(1 \pm \eps)$ of the number of constraints satisfied by $\sigma$ in $\cI$. Although there are often many possible sparsifiers, we seek to identify one whose support size (i.e., number of constraints with positive weight) is as small as possible. Concretely, when $R = \text{NEQ} = \{(0,1), (1,0)\} \subseteq \{0,1\}^2$, the associated CSP sparsification is exactly \emph{cut sparsification} in undirected graphs. Allowing for arbitrary relations gives us a much richer set of sparsification problems to analyze.

Since the work of Kogan and Krauthgamer, several works have expanded our understanding of CSP sparsification in many cases. For example, the case of binary relations ($R \subseteq D^2$ for any finite domain $D$) is completely understood~\cite{filtser2017Sparsification,butti2020}. A classification of ternary Boolean relations~\cite{khanna2024Characterizations} has also been completed.  A recent work on sparsifying linear codes~\cite{khanna2024Code} also handled the case in which $R$ can be defined by a system of linear \emph{in}equations.

\begin{wrapfigure}{r}{0.35\textwidth}
\begin{center}
\scalebox{0.5}{
\begin{tikzpicture}[scale=0.9]
  \path[mindmap,concept color=blue!50]
    node[concept] {\Huge NRD}
    [clockwise from=0]
    child[concept color=orange] { node[concept] {\Large Extremal Combinatorics} }
    child[concept color=orange] { node[concept] {\Large Query Complexity} }
    child[concept color=orange] {  node[concept] {\Large Universal Algebra}    }
    child[concept color=orange] { node[concept] {\Large Kerneliz\-ation} }
    child[concept color=orange] { node[concept] {\Large Sparsifi\-cation} }
    child[concept color=orange] { node[concept] {\Large Logic} };
\end{tikzpicture}
}
\end{center}
\caption{Non-redundancy and its many connections.}
\end{wrapfigure}

In this backdrop, a recent work by Brakensiek and Guruswami~\cite{brakensiek2024Redundancy}
showed the very general result that \emph{every} CSP can be sparsified down to its \emph{non-redundancy}. Non-redundant instances (to be defined shortly) are obvious obstacles to sparsification, and this result shows that they are the only obstacle, for worst-case CSP instances. This general result motivates the systematic study of  non-redundancy of relations, which we undertake in this paper. In addition to its intimate connection to sparsification, non-redundancy also arises in the study of other computational questions about CSPs. One such area is that of \emph{kernelization}, where one seeks to distill an instance $\cI$ of $\CSP(R)$ to a smaller instance $\cI'$ while preserving satisfiability. In this setting, the study of non-redundancy has greatly aided in the efficient construction of kernels~\cite{lagerkvist2020Sparsification,chen2020BestCase,carbonnel2022Redundancy}.

Another related question is that of the query complexity of CSPs in the partial membership query model~\cite{paulin2008Automatic,bessiere2012non,bessiere2013constraint,bessiere2020Chain}. Each query has as input a partial assignment to the variables of the instance, and the output to the query is whether this partial assignment satisfies the clauses induced by this assignment.  The study of non-redundancy is crucially\footnote{We remark that this line of work appears to be the first to coin the term \emph{non-redundancy}.} tied to this query complexity problem and is in fact equal to the \emph{VC-dimension} of the query complexity problem.

More broadly, the non-redundancy notion has historically given rise to algorithmic problems that have received significant attention. For example, checking whether a given constraint is redundant (often formulated as an \emph{implication problem}~\cite{schnoor2008a} or \emph{constraint entailment} problem~\cite{Bodirsky:book}), as well as the task of finding a maximal set of non-redundant constraints of a given instance (computing a \emph{prime} or \emph{minimal} instance--see also~\cite{MONTANARI197495}).
Furthermore, non-redundancy questions have been considered for propositional logic~\cite{Liberatore:ai2005,Liberatore:ai2008}, Datalog~\cite{datalog}, finite-domain CSPs~\cite{Dechter:Dechter:aaai87,lagerkvist2024b}, and various forms of tempo-spatial qualitative reasoning problems~\cite{Li:etal:ai2015,Lee:etal:ecai2016, lagerkvist2026}. 
Non-redundancy has also been studied extensively in the context of relational databases where a standard approach is to decompose a database into a suitable \emph{normal form} where one can  eliminate redundant entries~\cite{Codd1974,ArenasLibkin2005,DarwenDateFagin2012}.
Beyond these algorithmic applications, the study of non-redundancy lies at the nexus of extremal combinatorics and universal algebra, with its study revealing many intriguing phenomena and connections~\cite{lagerkvist2020Sparsification,chen2020BestCase,bessiere2020Chain,carbonnel2022Redundancy,brakensiek2024Redundancy}.

\smallskip

In this paper, we build a comprehensive theory of non-redundancy of relations, and make numerous contributions advancing our understanding of non-redundancy on multiple fronts, including their possible asymptotic behaviors and interrelationships.

With this context, let us now finally define non-redundancy of relations formally. To introduce it, let us ask the simple question: 
What is the maximum number of possible edges in an $n$-vertex undirected graph that has no cycle? The answer, of course, is $n-1$, achieved when the graph is a tree. Viewing a graph as an instance of a constraint satisfaction problem (CSP) with equality constraints imposed on adjacent vertices, an acyclic graph is one where every equality constraint is \emph{non-redundant} in the sense that it is not implied by the collection of all other constraints. A tree is then a maximum-sized collection of non-redundant equality constraints. 

Generalizing this, for any relation $R \subseteq D^r$ of arity $r$ over a finite domain $D$, we say that an instance of $\CSP(R)$, with constraints $R$ applied to certain $r$-tuples of variables, is non-redundant if each of its constraints is not implied by the others. In other words, for each constraint there is an assignment violating only that constraint, so that dropping any single constraint in a non-redundant instance alters the set of satisfying assignments. The non-redundancy of $R$, which we denote by $\NRD(R, n)$, is the largest number of possible constraints in a non-redundant instance of $\CSP(R)$, as a function of the number $n$ of variables of the CSP instance. For non-trivial relations $\emptyset \neq R \subsetneq D^r$, $\NRD(R, n)$ lies in the range $[\Omega(n),\Oh(n^r)]$.

The aforementioned sparsification result~\cite{brakensiek2024Redundancy} shows that for any instance of $\CSP(R)$ on $n$ variables, there is a sub-instance with at most $\tilde{\Oh}_\eps(\NRD(\overline{R}, n))$ reweighted constraints that preserves the number of original constraints satisfied by all assignments up to $(1 \pm \eps)$ multiplicative factors.\footnote{
Note that an assignment that witnesses that an $\overline{R}$-constraint is not redundant also witnesses the corresponding $R$-constraint cannot be dropped by a sparsifier since an assignment satisfying all but one $\overline{R}$-constraint satisfies exactly one $R$-constraint. This makes $\NRD(\overline{R}, n)$ an obvious lower bound for sparsifying worst-case $n$ variable $\CSP(R)$ instances. For $R = \text{NEQ} = \{(0,1), (1,0)\} \subseteq \{0,1\}^2$ (which corresponds to cut sparsification), the relation $\overline{R} = \text{EQ} = \{(0,0),(1,1)\}$ has $\NRD(\overline{R},n) = n-1$, with trees being extremal non-redundant instances as mentioned earlier.}
This general result still leaves the question of determining $\NRD(R, n)$ for a given relation $R$. As we demonstrate in this paper, the NRD of even simple relations can have a surprisingly rich behavior, and the study of NRD is a challenging pursuit with many exciting mathematical connections.

\subsection{Our Contributions}

We now present the contributions of this paper to the theory of non-redundancy. We divide our contributions into three broad categories. First, we improve our understanding of the possible growth rates of non-redundancy functions. For instance, we show that for any rational number $p/q \ge 1$, there exists a predicate $R_{p,q}$ for which $\NRD(R_{p,q}, n) = \Theta_{p,q}(n^{p/q})$. Second, we shed new light on the nature of predicates $R$ for which $\NRD(R, n)$ grows as a linear function of $n$. Our contributions here include discovering a new universal algebraic object known as \emph{Catalan polymorphisms} which gives new insights into the nature of Mal'tsev embeddings, a general technique for establishing a predicate has linear non-redundancy. Third, to more deeply understand the logical underpinnings of non-redundancy, we develop a universal algebraic theory of a generalization of non-redundancy known as conditional non-redundancy.  As an application, we show that many hypergraph Tur{\'a}n problems in extremal combinatorics (including the notorious Erd{\H o}s box problem) can be recast in the language of conditional non-redundancy.

\subsubsection{Possible Growth Rates for Non-redundancy}

We begin by discussing our results on the possible growth rates of non-redundancy. For any non-trivial relation $R$ of arity $2$ (over any finite domain), it was known that $\NRD(R, n)$ is either $\Theta(n)$ or $\Theta(n^2)$~\cite{bessiere2020Chain,butti2020}. In \cite{brakensiek2024Redundancy}, the authors give the first example of a relation, with arity $3$, whose NRD is asymptotically not a polynomial function, and lies between $\Omega(n^{1.5})$ and $\Oh(n^{1.6})$.
This hints at the richness of non-redundancy and raises the question of whether the growth rate of NRD can have a fractional exponent. As our first result, we give a strong positive answer to this question, showing in fact that \emph{every} rational exponent occurs as the NRD of some relation.
\begin{theorem}[Every fractional exponent as NRD]
\label{thm:intro:every-rational}
For every rational number $p/q \ge 1$, there is a relation $R_{p,q}$ such that $\NRD(R_{p,q}, n)  = \Theta_{p,q}(n^{p/q})$.
\end{theorem}
The arity of the relation $R_{p,q}$ can be taken to be $p$.  In particular there is an arity $3$ relation $R$ with $\NRD(R, n) = \Theta(n^{1.5})$. We discuss this construction in more detail in the technical highlights (Section~\ref{subsec:tech-highlights}). Note that this result may be viewed as a logical analogue of a classical result in extremal combinatorics due to Frankl~\cite{frankl1986All}, who showed that every rational $p/q$ appears as the growth rate of some hypergraph Tur{\'a}n problem. 

The predicates used to prove \Cref{thm:intro:every-rational} also have significant consequences for the concept of \emph{kernelization}, which is a formalization of polynomial-time preprocessing with performance guarantees for NP-hard decision problems~\cite{FominLSZ19}. In the context of a CSP over a finite constraint language~$\Gamma$ consisting of one or more finite relations, a kernelization algorithm with size bound~$f(n)$ is a polynomial-time algorithm that reduces any $n$-variable instance of $\CSP(\Gamma)$ into an equisatisfiable instance on at most~$f(n)$ constraints. A priori, the NRD growth rate of the predicates in~$\Gamma$ seems unrelated to the best-possible kernelization size bound for~$\CSP(\Gamma)$. In one direction, bounds on NRD may be non-constructive and therefore not achievable by a polynomial-time algorithm; in the other direction, a kernelization algorithm is allowed to modify the instance by changing constraints (or even introducing new ones) to arrive at an equisatisfiable instance of small size, while a reduction to a non-redundant instance can only omit constraints. Surprisingly, it turns out that in many cases the best-possible kernelization size for~$\CSP(\Gamma)$ coincides with the NRD value of the most difficult predicate in~$\Gamma$ (up to~$n^{o(1)}$ factors).

To make this precise, define the \emph{kernelization exponent} of $\csp(\Gamma)$ as the infimum of all values~$\alpha$ for which~$\CSP(\Gamma)$ has a kernelization of size~$\Oh(n^{\alpha})$. Chen, Jansen, and Pieterse~\cite{chen2020BestCase} presented a framework based on bounded-degree polynomials to derive kernelization algorithms for various CSPs, which was later generalized by Lagerkvist and Wahlstr\"{o}m~\cite{lagerkvist2020Sparsification} using algebraic tools. Through matching lower bound constructions, conditional on the complexity-theoretic assumption \ncontainment, the exact kernelization exponent is known for many CSPs~\cite{DellM12,chen2020BestCase,lagerkvist2020Sparsification}. All kernelization exponents for~$\CSP(\Gamma)$ known to date are \emph{integers}. Even for kernelization algorithms outside the realm of CSPs, to the best of our knowledge all NP-hard problems for which matching kernelization upper and lower bounds are known (e.g.~\cite{DellM12,dell2014Satisfiability,JansenW24}) have an \emph{integral} kernelization exponent. This raises the question whether~$\CSP(\Gamma)$ has an integral kernelization exponent for each choice of~$\Gamma$. An adaptation of \Cref{thm:intro:every-rational} to the setting of kernelization, described in \Cref{app:ker}, provides a strong negative answer to this question. We show that for each choice of integers~$p \geq q \ge 1$, there is a constraint language~$\Gamma$ whose kernelization exponent is exactly~$\frac{p}{q}$ (assuming \ncontainment).

As our next result we show that the NRD of arity $3$ relations itself has a very rich behavior, and can precisely capture the famous extremal graph theory problem of the largest number of edges in a graph of girth at least $2k$, for any $k \ge 2$. This highlights that the study of non-redundancy may offer powerful new viewpoints on long-standing problems in extremal combinatorics. %

\begin{theorem}[Infinitely many ternary NRD exponents approaching $2$]
\label{thm:intro:all-girths}
For every integer $k\ge 2$, there is a ternary relation $R_{2k}$ such that $\NRD(R_{2k}, n)  = n \cdot  \Theta_k(\ex(n, \{C_3,\dots,C_{2k-1}\}))$ where $\ex(n, \{C_3,\dots,C_{2k-1}\})$ is the largest number of edges in an $n$-vertex graph with no cycle of length less than $2k$. 
\end{theorem}

It is believed that $\ex(n, \{C_3,\dots,C_{2k-1})) = \Theta_k(n^{1+1/k})$ and it is known that it grows as $n^{1+\Theta(1)/k}$~\cite{furedi2013history}. In particular, this implies that $\NRD(R_{2k}, n) = n^{2+\Theta(1)/k}$, showing that infinitely many exponents are possible for ternary relations. 

Recall that the NRD of ternary relations lies in the range $[\Omega(n),\Oh(n^3)]$. It is known that for a ternary relation $R$, either $\NRD(R, n) = \Theta(n^3)$ or $\NRD(R, n)= \Oh(n^{2.75})$~\cite{carbonnel2022Redundancy}, so the NRD exponent is bounded away from $3$. Furthermore, \Cref{thm:intro:all-girths} shows that one can approach quadratic NRD arbitrarily closely. Can the NRD exponent of ternary relations approach $1$ arbitrarily closely? We will soon discuss a candidate family of predicates to address this question (Theorem~\ref{thm:intro:ternary-nrd-near-linear}) but first we describe our results for languages with {\em near-linear} NRD.

\subsubsection{On Linear Non-redundancy: Catalan polymorphisms and Embeddings}
Recall that by the recent general result of \cite{brakensiek2024Redundancy} a linear bound on $\NRD(\overline{R})$ implies a sparsifier for the associated $\CSP(R)$ of $n \polylog(n)$ size. Hence, obtaining near-linear bounds for non-redundancy seems highly desirable, albeit a significant challenge, and has been the primary motivation for earlier works on non-redundancy and sparsification~\cite{jansen2019Optimal,chen2020BestCase,lagerkvist2020Sparsification,kogan2015,filtser2017Sparsification,butti2020,chen2020,khanna2024Code,khanna2024Characterizations,khanna2024optimal}.

To hint at the richness of this question, consider the relation $\oneinthree := \{(1,0,0),(0,1,0),(0,0,1)\}$ $\subset \{0,1\}^3$ whose associated CSP is the classic NP-complete problem $\oneinthree\text{-SAT}$.
Note that $\oneinthree = L \cap \{0,1\}^3$ where $L = \{(a,b,c) \in (\Z/3\Z)^3 \mid a+b+c = 1\}$. Thus, a non-redundant instance of $\CSP(\oneinthree)$ is a fortiori also a non-redundant instance of $\CSP(L)$, with the same Boolean assignments as non-redundancy witnesses.  Being an affine relation,  $\NRD(L, n) = \Oh(n)$ since a maximal non-redundant instance is the same as a basis and  a basis must have linear size. Therefore, we also have that $\NRD(\oneinthree, n) = \Oh(n)$. Such \emph{embeddings} (\Cref{subsec:embeddings}), and more generally Abelian group embeddings (i.e., embeddings into relations consisting of solutions to an Abelian group equation), 
have been {\em the} principal tool for establishing (near) linear bounds in both the NRD~\cite{jansen2019Optimal,chen2020BestCase,lagerkvist2020Sparsification} and the sparsification context~\cite{khanna2024Code,khanna2024Characterizations}. 
Curiously, it is additionally known that one can deduce linear NRD for a relation $R$ from so-called \emph{Mal'tsev embeddings}~\cite{maltsev1961CONSTRUCTIVE,lagerkvist2020Sparsification}, i.e., when $R \subseteq D^r$ is the projection of a relation $S \subseteq E^r$ that is closed under a Mal'tsev term, i.e., a ternary 
function $\phi : E^3 \to E$ satisfying the identities $\phi(x,x,y) = \phi(y,x,x) = y$. In this context $\phi$ is also called a {\em polymorphism}. The map $\phi(x,y,z) = x \cdot y^{-1} \cdot z$ over a group $(G, \cdot)$ is the quintessential example of a Mal'tsev term.

However, a fundamental open question regarding Mal'tsev embeddings was whether they are more powerful than Abelian group embeddings, and all known Mal'tsev embeddings prior to our work coincided with Abelian embeddings. As one of our main results, we show that non-Abelian embeddings are indeed more powerful than Abelian embeddings via a concrete predicate.

\begin{theorem}
    \label{thm:intro:maltsev-vs-abelian}
    There is a relation $\PAULI \subseteq D^6$ with $|D|=3$ that admits a Mal'tsev embedding, in fact a non-Abelian group embedding, but no Abelian group embedding.
\end{theorem}
The non-Abelian group hosting the claimed embedding is the 16-element Pauli group generated by the single qubit unitaries $X,Y,Z$, hence our choice of name $\PAULI$ for the relation.
The proof of the above hinges on a key concept that we introduce toward a better understanding of Mal'tsev terms, namely a higher arity generalization which we call \emph{Catalan polymorphisms}. In the group example above of a Mal'tsev term, the higher order alternating product $\psi_m(x_1,x_2,\dots,x_m) = x_1 \cdot x_2^{-1} \cdot x_3 \cdot x_4^{-1} \cdots \cdot x_{m-1}^{-1} \cdot x_m$ for odd $m \ge 3$ also satisfies a ``cancellation rule" 
$\psi_m(x_1,x_2,\dots,x_m)  =  \psi_{m-2}(x_1,\dots,x_{i-1},x_{i+2},\dots,x_m)$ whenever $x_i=x_{i+1}$. We call these \emph{Catalan identities} due to a close relationship with Catalan numbers.
As our main technical tool, we show (Theorem~\ref{thm:catalan}) that higher order polymorphisms obeying such Catalan identities exist on the basis of \emph{any} Mal'tsev term, not just in the group product case. In an Abelian group, the alternating product has extra non-contiguous cancellations (e.g., when $x_i = x_{i+3}$). Exploiting this, we define a relation over $\{x,y,z\}$ based on Catalan identities involving three variables $x,y,z$ that fails to be closed under any commutative alternating product and thus lacks an Abelian embedding.
However, the map $x \mapsto X$, $y \mapsto Y$ and $z \mapsto Z$ embeds the relation into the non-Abelian Pauli group.

In contrast to the above separation, we show that over Boolean domains, Mal'tsev and Abelian group embeddings coincide, answering an open question of \cite{chen2020BestCase}. This proof also crucially uses the Catalan polymorphisms as a tool. 

\begin{theorem}
    \label{thm:intro:Maltsev-Boolean}
Let $P \subseteq \{0,1\}^r$ be a relation which embeds into a relation $A \subseteq D^r$ that admits a Mal'tsev polymorphism. Then, $P$ also embeds into a finite Abelian group.
\end{theorem}

One of the main remaining open questions about Mal'tsev embeddings for non-Boolean domains is whether the existence of a Mal'tsev embedding over an infinite domain always implies the existence of a Mal'tsev embedding over a finite domain. While this question remains open, we do manage to close the gap between Mal'tsev and group embeddings and prove that a Mal'tsev embedding (over a finite or infinite domain) always implies an embedding into an infinite {\em Coxeter} group. This, again, uses the novel description of Mal'tsev embeddings via Catalan polymorphisms and represents an important step towards fully characterizing Mal'tsev embeddings. In particular, the question of whether an embedding over an infinite group can always be brought down to the finite has a higher chance of being answerable by existing techniques in group theory than the corresponding question for arbitrary Mal'tsev terms.

While Mal'tsev embeddings suffice for linear NRD, it is not clear if they are the only reason governing linear NRD. We define a family of ternary relations, alluded to earlier, which demonstrate that if Mal'tsev embeddings are necessary for linear NRD, they are ``barely so," in the sense that one can obtain NRD at most $n^{1+\eps}$ for any $\eps > 0$ while lacking a Mal'tsev embedding. Note that Theorem~\ref{thm:intro:every-rational} proves a similar statement but requires the arity to grow as $\eps$ gets smaller.

\begin{theorem}
    \label{thm:intro:ternary-nrd-near-linear}
    For every odd $m$, there is a ternary predicate $\CYCs_m$ over domain size $m$ such that $\NRD(\CYCs_m, n) \le \Oh_m(n^{\frac{m}{m-1}}).$ Further, $\CYCs_m$ does not admit a Mal'tsev embedding. 
\end{theorem}

Here, $\CYCs_m := \{(x,y,z) \in (\Z/m\Z)^3 : x+y+z = 0 \wedge y-x \in \{0,1\}\} \setminus \{(0,0,0)\}$. Once again, the Catalan terms that we introduced and related to Mal'tsev terms play a key role in proving the above statement. We
show that $\CYCs_m$ does not admit a Catalan polymorphism of arity $2m-1$ and thus also lacks a Mal'tsev embedding. This can be viewed as some evidence that $\CYCs_m$ might have super-linear NRD. Establishing such a lower bound remains an open question.

\subsubsection{A Theory of Conditional NRD with Connections to Hypergraph Tur{\'a}n}

An important tool to study NRD, including establishing aforementioned Theorems~\ref{thm:intro:every-rational}, \ref{thm:intro:all-girths}, and \ref{thm:intro:ternary-nrd-near-linear}, is the concept of \emph{conditional non-redundancy} defined in \cite{brakensiek2024Redundancy} (although implicit in \cite{bessiere2020Chain}). Here, when studying $\NRD(R, n)$ for $R \subsetneq D^r$, we focus on a relation $S \supsetneq R$ (often with just one more tuple) and require the non-redundant instance of $\CSP(R)$ to have the property that for every constraint, there is an assignment violating only that constraint by landing within $S \setminus R$ (as opposed to anywhere in $D^r \setminus R$). We denote this restricted form of non-redundancy by $\NRD(R \mid S, n)$, where $R \mid S$ is what we call a \emph{promise relation}. It is useful since it can be much easier to upper bound $\NRD(R \mid S, n)$ (which is a priori at most $\NRD(R, n)$), and for judicious choices of $S$, one can also deduce a good upper bound on $\NRD(R, n)$ via a ``triangle inequality'' $\NRD(R, n) \le \NRD(R \mid S, n) + \NRD(S, n)$. En route to Theorem~\ref{thm:intro:all-girths} we give a full \emph{characterization} of all possible behaviors of non-redundancy for promise binary relations over arbitrary finite domains: For every $R \subsetneq S \subseteq D^2$, either $\NRD(R \mid S, n) = O_D(n)$, or for some $k \ge 2$, $\NRD(R \mid S, n) =  \Theta(\ex(n, \{C_3,\dots,C_{2k-1}\}))$, see Theorem~\ref{thm:binary-class}.

The non-redundancy of promise relations appears to be a powerful tool for more fine-grained NRD results (necessary by Theorem~\ref{thm:intro:every-rational}) and it seems apt to study it in the {\em universal algebraic} setting. Here, the basic idea is to forget about concrete relations and focus on algebraic conditions that can be derived from ``higher order'' homomorphisms (\emph{polymorphisms}), of which the previously mentioned Mal'tsev term is a notable special case. 
However, in universal algebra one often assumes {\em total} polymorphisms that, roughly, correspond to classical gadget reductions between problems. For example, 1-in-3-SAT and 3-SAT are polynomial-time interreducible, and in fact have the same polymorphisms, but different NRD (linear versus cubic). A second complication in the promise relation setting (i.e., $\NRD(R \mid S, n)$) is that we need to take the relationship between $R$ and $S$ into account. We prove that {\em partial promise polymorphisms}, applied componentwise on tuples from $R$ and producing a tuple in $\widetilde{S} = D^r \setminus (S \setminus R)$, govern $\NRD(R \mid S, n)$. 
However, describing $\NRD(R \mid S, n)$ by explicitly enumerating all possible partial promise polymorphisms is not realistic. As a comparison, imagine a world where we tried to understand calculus through large tables of function values, rather than using equations or  further abstractions such as only looking at the degree or the derivative of a polynomial. The same applies here: we need abstractions that allow us to concentrate on certain key properties. We develop a general theory 
and prove that the most significant algebraic properties come from (restricted) systems of linear identities/equations which we refer to as {\em polymorphism patterns} and let $\pattern(R,\widetilde{S})$ be the set of patterns satisfied by partial promise polymorphisms of $(R, \widetilde{S})$. See Figure~\ref{fig:3_cube} for an example of a pattern and a  satisfying partial function. We then manage to prove the following general statement.

\begin{figure}[h]
    \centering
    \begin{tikzpicture}[every node/.style={font=\small\boldmath}]

        \definecolor{eq1color}{RGB}{230, 120, 0} %
        \definecolor{eq2color}{RGB}{0, 100, 230} %
        \definecolor{eq3color}{RGB}{40, 180, 40} %

        \node[text=eq1color] (eq1) at (0,1.5) {$U_3(x,x,x,y,y,y,y) = x$};
        \node[text=eq1color] (bool1a) at (5,1.5) {$u_3(0,0,0,1,1,1,1) = 0$};
        \node[text=eq1color] (bool1b) at (9.5,1.5) {$u_3(1,1,1,0,0,0,0) = 1$};

        \node[text=eq2color] (eq2) at (0,1.0) {$U_3(x,y,y,x,x,y,y) = x$};
        \node[text=eq2color] (bool2a) at (5,1.0) {$u_3(0,1,1,0,0,1,1) = 0$};
        \node[text=eq2color] (bool2b) at (9.5,1.0) {$u_3(1,0,0,1,1,0,0) = 1$};

        \node[text=eq3color] (eq3) at (0,0.5) {$U_3(y,x,y,x,y,x,y) = x$};
        \node[text=eq3color] (bool3a) at (5,0.5) {$u_3(1,0,1,0,1,0,1) = 0$};
        \node[text=eq3color] (bool3b) at (9.5,0.5) {$u_3(0,1,0,1,0,1,0) = 1$};

        \node[text=black] (boolConst1) at (5,0) {$u_3(0,0,0,0,0,0,0) = 0$};
        \node[text=black] (boolConst2) at (9.5,0) {$u_3(1,1,1,1,1,1,1) = 1$};

    \end{tikzpicture}
    \caption{The 3-universal/cube pattern $U_3$ and the resulting partial Boolean function (named $u_3$). Each Boolean tuple is colored with the color of the identity from which it can be derived. The two constant Boolean tuples can be produced by any identity and are colored in black.}

    \label{fig:3_cube}
\end{figure}

\begin{theorem}\label{thm:intro:pattern-NRD}
Let $S_1 \subseteq T_1$ and $S_2 \subseteq T_2$ be pairs of relations. If \[\pattern(S_2, \widetilde{T_2}) \subseteq \pattern(S_1, \widetilde{T_1})\] then \[\NRD(S_1 \mid T_1, n) = \Oh(n + \NRD(S_2 \mid T_2, n)).\]
\end{theorem}
The theorem is proven by first establishing a suitable {\em Galois correspondence} between $\pattern(\cdot, \cdot)$ and sets of relations closed under conjunctive formulas with functional, unary terms ({\em functionally guarded promise primitive positive definitions}, fgppp-definitions). This generalizes  \emph{fgpp-definitions} (see Carbonnel~\cite{carbonnel2022Redundancy}) to the promise setting.
To extend this even further we consider a stronger operation that maps \emph{$c$-tuples} of variables in one relation to variables in another, allowing us to, e.g., reduce a problem $\CSP(S_1 \mid T_1)$ on $n$ variables to a problem $\CSP(S_2 \mid T_2)$ on $\Oh(n^2)$ variables. This naturally gives rise to the kind of fractional non-redundancy exponents seen in Theorem~\ref{thm:intro:every-rational}. 
The algebraic counterpart is a novel $c$-ary {\em power} operator $P^c$ on a pattern $P$ where we group together $c$-tuples of variables, and view them as a single variable.
See Figure~\ref{fig:u_hypergraph} for an example when $c = 2$. This pattern is produced by  selecting pairs of variables from the columns of any ordered selection of two equations, and viewing it as a pattern over $V^2$ instead of $V$. In the figure $(x,x), (x,y), (y,x), (y,y)$ have been renamed to $a,b,c,d$. By first relating $c$-ary powers of patterns ($\pattern^c(S,T)$) to $c$-tupled fgppp-definitions we obtain the following link to the non-redundancy of promise relations.

\begin{theorem} \label{thm:intro:pattern-power}
Let $(S_1, T_1)$ and $(S_2, T_2)$ be promise relations over $D_1$, respectively, $D_2$, and let $c \geq 1$.
If \[\pattern^c(S_2, \widetilde{T}_2) \subseteq \pattern(S_1, \widetilde{T}_1)\] then \[\NRD(S_1 \mid T_1, n) = O_{D_1, D_2, r_1,r_2}(n + \NRD(S_2 \mid T_2, n^c)).\]
\end{theorem}

These notions generalize and unify all known algebraic reductions in the CSP literature that involve restricted existential quantification~\cite{brakensiek2024Redundancy, carbonnel2022Redundancy, chen2020BestCase,   lagerkvist2021Coarse}. In particular, the power operation makes it possible to tightly relate $\NRD(S_1 \mid T_1, n)$ and $\NRD(S_2 \mid T_2, n)$ even when they have different exponents. This generalizes and unifies techniques implicit in Theorem~\ref{thm:intro:every-rational} and appears to be a generally useful method for bounding NRD.
For example, if $U_3 \notin \pattern(S, \widetilde{T})$ then $\NRD(S \mid T, n) = \Omega(n^3)$
and if $U_3^2 \notin \pattern(S,\widetilde{T})$ then $\NRD(S \mid T, n) = \Omega(n^{3/2})$.

\begin{figure}[h]
    \centering
    \begin{tikzpicture}[every node/.style={font=\small}]
    \tikzstyle{vertex}=[inner sep=0pt]

\tikzstyle{hyperedge}=[thin,black]

    \definecolor{darkgreen}{RGB}{0,150,0}

        \node (eq1) at (0,2) {$U_3^2(\textbf{\textcolor{darkgreen}{a}},b,b,c,c,d,d) = \textbf{\textcolor{red}{a}}$};
        \node (eq2) at (0,1) {$U_3^2(\textbf{\textcolor{darkgreen}{b}},a,b,c,d,c,d) = \textbf{\textcolor{red}{a}}$};
        \node (eq3) at (0,0) {$U_3^2(\textbf{\textcolor{darkgreen}{b}},c,d,a,b,c,d) = \textbf{\textcolor{red}{a}}$};

        \node[vertex] (a1) at (3,2) {$\bullet$};
        \node[vertex, below=of a1] (a2) {$\bullet$};
        \node[vertex, below=of a2] (a3) {$\bullet$};
        \node[vertex] (l1) at (3,2.5) {$\textbf{a}$};
        
        \node[vertex, right=of a1] (b1) {$\bullet$};
        \node[vertex, below=of b1] (b2) {$\bullet$};
        \node[vertex, below=of b2] (b3) {$\bullet$};
        \node[vertex] (l2) at (4.2,2.5) {$\textbf{b}$};
        
        \node[vertex, right=of b1] (c1) {$\bullet$};
        \node[vertex, below=of c1] (c2) {$\bullet$};
        \node[vertex, below=of c2] (c3) {$\bullet$};
        \node[vertex] (l3) at (5.4,2.5) {$\textbf{c}$};
        
        \node[vertex, right=of c1] (d1) {$\bullet$};
        \node[vertex, below=of d1] (d2) {$\bullet$};
        \node[vertex, below=of d2] (d3) {$\bullet$};
        \node[vertex] (l4) at (6.6,2.5) {$\textbf{d}$};
        
        \draw[thick,color=red] (a1) -- (a2) -- (a3);
        \draw[thick,color=darkgreen] (a1) -- (b2) -- (b3);

        \draw[hyperedge] (b1) -- (a2.east) -- (c3);
        \draw[hyperedge] (b1) -- (b2.east) -- (d3);
        \draw[hyperedge] (c1) -- (c2.west) -- (a3);
        \draw[hyperedge] (c1) -- (d2.west) -- (b3);
        \draw[hyperedge] (d1) -- (c2.east) -- (c3);
        \draw[hyperedge] (d1) -- (d2.east) -- (d3);
    \end{tikzpicture}
    \caption{The polymorphism pattern $U_3^2$ corresponding to the square of $U_3$ from Figure~\ref{fig:3_cube}. For example, the first equation is obtained by combining the orange and blue equation in $U_3$. The corresponding 3-partite hypergraph is obtained by letting each partite set correspond to variables $a,b,c,d$ in the corresponding equation, and adding a hyperedge for each ``column'' of the equations.}
    \label{fig:u_hypergraph}
\end{figure}

Given that polymorphism patterns strongly correspond to non-redundancy, an interesting algebraic question is to fix a set of $r$-ary patterns $Q$ and study $\NRD(S \mid T, n)$ if $r$-ary $(S,T)$ is preserved by $Q$. We relate this to hypergraph Tur{\'a}n problems in the following way: by associating $Q$ to a set of hypergraphs $\cH(Q)$ one can prove that the maximum non-redundancy possible exactly coincides with the size of the largest hypergraph on $n$ vertices excluding $\cH(Q)$ as subhypergraphs. This greatly expands on connections between non-redundancy and hypergraph Tur{\'a}n problems observed by previous works~\cite{bessiere2020Chain, carbonnel2022Redundancy,
brakensiek2024Redundancy}. See Figure~\ref{fig:u_hypergraph} for an example of this translation. To state the theorem formally we let $\inv^{= r}(Q)$ be the set of relation pairs of arity $r$ preserved by $Q$, which we assume is equipped with $r$ disjoint domain types ({\em sorts}) corresponding to the $r$ partite sets in a hypergraph.

\begin{theorem}\label{thm:intro:hypergraph}
  Let $Q$ be a finite set of ($r$-ary multisorted) polymorphism patterns. Then one can construct a set $\cH(Q)$ of $r$-partite hypergraphs with the following property:
  for all $n \in N$,
  \begin{align*}
     \max_{(R,S) \in \inv^{= r}(Q)}\NRD(R \mid S, n) = \ex_r(n, \cH(Q)),
  \end{align*}
  where $\ex_r(n, \cH(Q))$ is the size of the largest $r$-partite $r$-uniform hypergraph on $n$ vertices excluding $\cH(Q)$.
\end{theorem}
For example, if $U_3^2 \in \pattern(R, \widetilde{S})$ then any non-redundant instance $(X,Y)$ of $\CSP(R \mid S)$ is free of the hypergraph in Figure~\ref{fig:u_hypergraph} (as well as other hypergraphs obtained by restricting the pattern $U^2_3$).
One interpretation of this result is that {\em finite} conditions $Q$, which algebraically seem rather weak when interpreted over a specific language in the promise setting, are best studied in the Tur{\'a}n setting, and contain notoriously hard open problems. For example, the cube pattern $U_3$ from Figure~\ref{fig:3_cube} results in the Tur{\'an} problem that is known as \emph{Erd\H{o}s box problem} ~\cite{erdos1964extremal,katz2002Remarks,gordeev2024Combinatorial}, with well-studied generalizations of the box problem corresponding to $U_k$ for arbitrary $k > 3$~\cite{erdos1964extremal,furedi2013history,conlon2021Random}.
 As such, resolving non-redundancy in the algebraically weak setting of finite sets of patterns cannot be done without simultaneously advancing other areas of mathematics. More optimistically, we hope these connections between hypergraph Tur{\'a}n problems and universal algebra could eventually lead to progress on these long-standing problems.

\subsection{Technical Highlights} \label{subsec:tech-highlights}

We now discuss the techniques utilized to prove the main results in this paper. Given the breadth of results we prove about NRD, we highlight a couple of key ideas from each area of contribution.

\subsubsection{NRD Exponents}
\paragraph{All Rational Exponents.} 
Our (first) proof of \Cref{thm:intro:every-rational} considers the predicate~$\ORDP_{p,q}$ on~$p+p^q$ variables over the domain~$\{0,1\} \cup \{0,1\}^q$. Intuitively, a tuple~$(x_1, \ldots, x_p, \tilde{x}_1, \ldots, \tilde{x}_{p^q})$ belongs to~$\ORDP_{p,q}$ when~$(x_1, \ldots, x_p) \in \OR_p$ (the arity-$p$ OR predicate) and the following \emph{direct product} relation holds: for each~$i \in [p^q]$, the value of the \emph{padding} variable~$\tilde{x}_i$ equals the bitstring of length~$q$ obtained by concatenating the values of~$x_{t_1}, \ldots, x_{t_q}$, where~$(t_1, \ldots, t_q) \in [p]^q$ corresponds to the base-$p$ representation of~$i$. This property ensures that whenever two constraints~$y, y'$ over~$\ORDP_{p,q}$ use the \emph{same} variable~$\tilde{x}_i$ at the same position~$i$, then to satisfy these constraints the values at positions~$(t_1, \ldots, t_q)$ of the first constraint must match the values at positions~$(t_1, \ldots, t_q)$ in the second. Intuitively, this allows us to identify the pairs of variables used on these~$q$ positions without changing the set of satisfying assignments. After applying all such identifications, each~$q$-tuple of OR-variables appearing on the first~$p$ positions of a constraint, is governed without loss of generality by a \emph{unique} padding variable~$\tilde{x}_i$.
Conversely, we can model any collection of~$\OR_p$-constraints over~$n$ variables by introducing~$n^q$ fresh variables to be used as padding and plugging these into the later positions of~$\ORDP_{p,q}$ constraints. The latter implies that~$\NRD(\ORDP_{p,q}, n + n^q) \geq \NRD(\OR_p, n)$, which is easily seen to be~$\Omega_p(n^p)$ since any instance of~$\CSP(\OR_p)$ is non-redundant if no two clauses involve the same set of~$p$ variables. This leads to the lower bound~$\NRD(\ORDP_{p,q}, n) \geq \Omega_{p,q}(n^{\frac{p}{q}})$. However, the upper-bound on~$\NRD(\ORDP_{p,q}, n)$ requires more work.

The basic idea for the upper bound is that given a non-redundant instance of $\ORDP_{p,q}$ on variable set $X$ and constraint set $Y$ that has been simplified by the identifications described above, we can derive two set families. First, let $\mathcal{A} \subseteq \binom{X}{p}$ be the $p$-uniform\footnote{We can avoid repeated variables by assuming without loss of generality that our instance is \emph{multipartite}. See Section~\ref{sec:prelim} for more details.} set family containing $\{y_1, \hdots, y_p\}$ for $y \in Y$. Furthermore, let $\mathcal{B} \subseteq \binom{X}{q}$ be the $q$-uniform set family consisting of sets $T$ which are a subset of some $S \in \mathcal{A}$. By the Kruskal-Katona theorem~\cite{kruskal,katona},
we have that $|\mathcal{B}| \ge \Theta_{p,q}(|\mathcal{A}|^{q/p})$. By utilizing the concept of \emph{conditional} non-redundancy we may assume that $|\mathcal{A}|$ equals the number of constraints $|Y|$ (in a suitable conditional setting, two constraints that employ the \emph{same}~$p$ variables in their $\OR_p$ part cannot be violated separately from each other) while $|\mathcal{B}|$ equals the number of padding variables in the instance (since each~$q$-tuple of~$x_j$'s has a unique~$\tilde{x}_i$ associated to them). Therefore, $\NRD(\ORDP_{p,q}, n) = \Oh_{p,q}(n^{p/q}),$ as desired.

Note that the arity of this construction is rather large. We are able to reduce the arity of the construction all the way down to $p$ by \emph{puncturing} the construction.
In particular, we consider a carefully-chosen subset of the coordinates of $\ORDP_{p,q}$ and argue that the asymptotic bound still holds. In particular, we are able to drop the first $p$ coordinates of $\ORDP_{p,q}$ corresponding to $\OR_p$. The main technical difference is instead of Kruskal-Katona we use \emph{Shearer's inequality} to prove the bound. See Section~\ref{subsec:smaller} for more details.

\paragraph{Graph Girth.} Unlike \Cref{thm:intro:every-rational}, the methods behind \Cref{thm:intro:all-girths} are connected to Tur{\'a}n theory. More precisely, for $k\ge 2$, we employ the following bipartite \emph{shoelace} representation (see \Cref{fig:girth}) of the cycle~$C_{2k}$ over two copies of the elements $\{0,\ldots,k-1\}$ as
\[
    C_{2k} := \{(0,0),(0,1),(1,0),(1,2),(2,1), \hdots, (k-2,k-1),(k-1,k-2),(k-1,k-1)\} \subseteq \{0,\hdots, k-1\}^2.
\]
That is, $(x,y) \in C_{2k}$ if $(x,y) \in \{(0,0),(k-1,k-1)\}$ or $|x-y| = 1$. Also let $C^*_{2k} := C_{2k} \setminus \{(0,0)\}$ be the cycle minus an edge. A crucial observation we make is that \emph{bipartite} graphs are non-redundant instances of $\CSP(C^{*}_{2k} \mid C_{2k})$ if and only if they have girth at least $2k$ (see \Cref{fig:girth}).
Since there is a constant-factor reduction from general to bipartite instances (see \Cref{lemma:reduction:to:multipartite}), we have that $\NRD(C^{*}_{2k} \mid C_{2k}, n) =  \Theta(\ex(n, \{C_3,\dots,C_{2k-1}\}))$. In fact, for general binary predicates $R \subsetneq S \subseteq D^2$, we have that $\NRD(R \mid S, n) = \Theta(\ex(n, \{C_3,\dots,C_{2k-1}\}))$, where $2k-1$ is the minimum length of a path between any pair of points $(x_0,y_0) \in S\setminus R$ within the bipartite graph on $2|D|$ vertices whose edges come from $R$.
The special case of no path existing (i.e., infinite girth) corresponds to $\NRD(R \mid S, n) = \Oh(n)$.

\begin{figure}[h]
\begin{center}
\begin{tikzpicture}[every node/.style={minimum size=0.7cm, draw=black, shape=circle},
                    every edge/.style={very thick}, scale=0.75]
\node[fill=green!70!white] (x_0) at (-0.5,0) {$x_0$};
\node[fill=red!70!white] (y_0) at (1,0) {$y_0$};

\node[fill=yellow!70!black] (a) at (-1.5, 0.8) {1};
\node[fill=yellow!70!black] (b) at (-1.5, -0.8) {1};

\node[fill=blue!70!white] (c) at (2, 0.8) {1};
\node[fill=blue!70!white] (d) at (2, -0.8) {1};

\node[fill=purple!70!white] (e) at (-2.5, 1.6) {2};
\node[fill=purple!70!white] (f) at (-2.5, 0.0) {2};
\node[fill=purple!70!white] (g) at (-2.5, -1.6) {2};

\node[fill=orange!70!white] (h) at (3, 1.6) {2};
\node[fill=orange!70!white] (i) at (3, 0.0) {2};
\node[fill=orange!70!white] (j) at (3, -1.6) {2};

\node[fill=orange!70!white] (k) at (-3.5, 2.4) {2};
\node[fill=orange!70!white] (l) at (-3.5, 0.8) {2};
\node[fill=orange!70!white] (m) at (-3.5, -0.8) {2};

\node[fill=purple!70!white] (n) at (4, 2.4) {2};
\node[fill=purple!70!white] (o) at (4, 0.8) {2};
\node[fill=purple!70!white] (p) at (4, -0.8) {2};

\node[fill=purple!70!white] (q) at (-4.5, 0) {2};

\draw[very thick,red,dashed] (x_0) -- (y_0);
\draw[very thick] (b) -- (x_0) -- (a);
\draw[very thick] (d) -- (y_0) -- (c);

\draw[very thick] (e) -- (a) -- (f);
\draw[very thick] (g) -- (b);

\draw[very thick] (i) -- (d) -- (j);
\draw[very thick] (h) -- (c);

\draw[very thick] (k) -- (n);
\draw[very thick] (g) -- (j);

\draw[very thick] (h) -- (o) -- (i);
\draw[very thick] (p) -- (j);

\draw[very thick] (g) -- (m) -- (f);
\draw[very thick] (k) -- (e) -- (l);

\draw[very thick] (q) -- (m);
\draw[very thick] (n) -- (h);

\node[fill=green!70!white] (l0) at (7,-1.1) {$0$};
\node[fill=red!70!white] (r0) at (8.5,-1.1) {$0$};
\node[fill=blue!70!white] (l1) at (7,0.4) {$1$};
\node[fill=yellow!70!black] (r1) at (8.5,0.4) {$1$};
\node[fill=purple!70!white] (l2) at (7,1.9) {$2$};
\node[fill=orange!70!white] (r2) at (8.5,1.9) {$2$};

\draw[draw=red,dashed,very thick] (l0) -- (r0);
\draw[very thick] (l0) -- (r1);
\draw[very thick] (l1) -- (r0);
\draw[very thick] (l1) -- (r2);
\draw[very thick] (l2) -- (r1);
\draw[very thick] (l2) -- (r2);

\end{tikzpicture}
\end{center}
\caption{This figure illustrates how a bipartite graph $G$ with girth $6$ (left) is a non-redundant instance of $\CSP(C^*_6 \mid C_6)$ (right). In particular, the colors illustrate a homomorphism from $G$ to $C_6$ such that only the target edge $(x_0, y_0)$ is mapped to $(0,0)$. This coloring is generalized to bipartite graphs of arbitrary girth in \Cref{lem:cycle-nrd}.} \label{fig:girth}
\end{figure}

We can lift bounds on the NRD of binary promise predicates into the NRD of (non-promise) ternary predicates by considering the following family of ternary predicates. For any $k \ge 2$, define
\begin{align*}
R_{2k} &:= (C_{2k} \times \{0,1\}) \setminus \{(0,0,0)\}, & S_{2k} &:= C_{2k} \times \{0,1\}.
\end{align*}
Using techniques similar to those for analyzing $\NRD(C^*_{2k} \mid C_{2k}, n)$, we can prove that $\NRD(R_{2k} \mid S_{2k}, n) = \Theta(n \cdot \ex(n, \{C_3,\dots,C_{2k-1}\}))$. The key observation though is that $S_{2k}$ can be decomposed into a binary relation ($C_{2k}$) and a unary relation ($\{0,1\}$). As a result, we have that $\NRD(S_{2k}, n) = \Theta(n^2)$. Therefore, by the triangle inequality, we have that $\NRD(R_{2k}, n) \le \NRD(R_{2k} \mid S_{2k}, n) + \NRD(S_{2k}, n) = \Theta(n \cdot \ex(n, \{C_3,\dots,C_{2k-1}\})).$ Note this technique does not work for analyzing $\NRD(C^*_{2k} \mid C_{2k}, n)$ because for $k \ge 3$, we have that $\NRD(C_{2k}, n) = \Theta(n^2) > \NRD(C^*_{2k} \mid C_{2k}),$ so the triangle inequality would give an upper bound of $n^2$ instead of $n^{1+\Theta(1)/k}$. 

\subsubsection{Linear NRD and Catalan Polymorphisms}

We now dive deeper into the structure of the \emph{Catalan polymorphisms} which drive our main results on linear non-redundancy and Mal'tsev embeddings (Theorems~\ref{thm:intro:maltsev-vs-abelian}, \ref{thm:intro:Maltsev-Boolean}, and \ref{thm:intro:ternary-nrd-near-linear}). In particular for any domain $D$, we say that a sequence of operators $\{\psi_m : D^m \to D \mid m \in \N \text{ odd}\}$ are \emph{Catalan polymorphisms} (\Cref{def:Catalan-identities}) if they have the following properties: (1) for all $x \in D$, $\psi_1(x) = x$, and (2) for all odd $m \in \N$ at least $3$ and all $i \in [m-1]$, we have that for all $x_1, \hdots, x_m \in D$,
\begin{align}
x_i = x_{i+1} \implies \psi_m(x_1, \hdots, x_{m}) = \psi_{m-2}(x_1, \hdots, x_{i-1}, x_{i+2}, \hdots, x_{m}).\label{eq:catalan-intro}
\end{align}

Our critical technical result (Theorem~\ref{thm:catalan}) is that given \emph{any} operator $\varphi : D^3 \to D$ satisfying the Mal'tsev identities that $\varphi(x,x,y) = \varphi(y,x,x) = y$ for all $x,y \in D$, we can generate \emph{all} the Catalan polymorphisms by (1) composing $\varphi$ with itself and (2) identifying input coordinates as equal. As base cases, we have that  $\psi_1(x) = \varphi(x,x,x)$ and $\psi_3(x,y,z) = \varphi(x,y,z)$. For the first nontrivial case of $\psi_5$, we have that $\psi_5(x_1, x_2, x_3, x_4, x_5)$ is
\begin{align}
    \varphi(\varphi(x_1, x_2, \varphi(x_3, x_4, x_5)),
    \varphi(\varphi(x_1,x_2,x_3), x_3, \varphi(x_3,x_4,x_5)), \varphi(\varphi(x_1, x_2, x_3), x_4, x_5)).\label{eq-intro-cat-5}
\end{align}
One can verify that the $m=5$ case of (\ref{eq:catalan-intro}) holds for $\psi_5$ equal to (\ref{eq-intro-cat-5}) and $\psi_3$ equal to $\varphi$ (see \Cref{ex:cat5}). This construction is generalized to all odd $m$ by using a recursive formula reminiscent of the recursive formula for the Catalan numbers. Importantly, \Cref{thm:catalan} may be seen as a ``group theoretic'' analogue of the classical complexity theory result of Valiant (see e.g., \cite{gupta1996Using} and references therein) that a majority gate on three bits can simulate an arbitrary majority gate. That is, we prove that a ``Mal'tsev gate'' can simulate an arbitrarily large ``Catalan gate,'' although the depth of the reduction is exponential in $m$. Note that Valiant's reduction is of much importance in the current study of promise CSPs (e.g., \cite{austrin2017sat,brakensiek2021Promise,barto2021Algebraica}), perhaps suggesting Catalan polymorphisms could also be used to advance the theory of promise CSPs. The Catalan polymorphisms also have natural partial analogues which turn out to coincide with the so-called {\em universal partial Mal'tsev polymorphisms}~\cite{lagerkvist2020Sparsification} arising from any Mal'tsev embedding. However, the latter are poorly understood even for small arities, and have no clear connection to group theory.

The true power of the Catalan polymorphisms is revealed in its applications. In particular, the cancellation laws of the Catalan polymorphisms show that any Mal'tsev embedding ``acts like'' an embedding into an infinite group. More precisely, we prove that every predicate with a Mal'tsev embedding also embeds into an infinite group known as the Coxeter group. This immediately resolves Theorem~\ref{thm:intro:Maltsev-Boolean}, as \cite{chen2020BestCase} already showed that the result is true when a Boolean predicate has an infinite group embedding.\footnote{Independently, \cite{khanna2024Characterizations} proved a similar result to \cite{chen2020BestCase} when the predicate has an \emph{Abelian} group embedding.}

As previously discussed, the identities underlying Catalan polymorphisms were crucial in inspiring the construction of $\PAULI$ in \Cref{thm:intro:maltsev-vs-abelian}. In particular, for any operator $f \colon D^5 \to D$, consider the following six identities involving arbitrary $x,y,z \in D$.
\begin{align*}
        f(x,x,y,y,z) &= z,&
        f(x,x,z,y,y) &= z,&
        f(z,x,x,y,y) &= z,\\
        f(x,y,y,x,z) &= z,&
        f(z,x,y,y,x) &= z,&
        f(x,y,z,x,y) &= z.
\end{align*}
The first five identities are satisfied by any Catalan polymorphism $\psi_5$, but the last one can only be guaranteed when the operator arises from an Abelian group. As such, if we define $\PAULI \in \{x,y,z\}^6$ to consist of five tuples corresponding to each coordinate of the identities above (e.g., the first coordinate of each tuple implies $(x,x,z,x,z,x) \in \PAULI$), then we can deduce that $\PAULI$ does not embed into an Abelian group as that would imply that $(z,z,z,z,z,z) \in \PAULI$. However, as previously mentioned, $\PAULI$ embeds into the Pauli group by mapping each of $x,y,z$ to its corresponding generator. 
Thus, the insights from Catalan polymorphisms allow us to separate Abelian embeddings from non-Abelian embeddings and conclude that no analogue of Theorem~\ref{thm:intro:Maltsev-Boolean} exists even when the domain size is $3$.

Finally, we note that Catalan polymorphisms greatly streamline proofs showing that a predicate lacks a Mal'tsev embedding. In particular, previous proofs that predicates lack Mal'tsev embeddings (in e.g., \cite{chen2020BestCase,bessiere2020Chain,lagerkvist2020Sparsification}) relied on constructing complicated Mal'tsev gadgets similar to that of (\ref{eq-intro-cat-5}), but the gadgets were found in an ad-hoc manner for the particular problem. With the Catalan polymorphisms being ``universal'' in the sense that they alone imply an infinite group embedding, to rule out an explicit predicate $P$ having a Mal'tsev embedding, we just need to check if there is an application of the Catalan polymorphisms to $P$ which results in a tuple not in $P$. This is precisely how we prove the lack of Mal'tsev embeddings in Theorem~\ref{thm:intro:ternary-nrd-near-linear}, while the upper bound follows from techniques developed to prove Theorem~\ref{thm:intro:all-girths}.

\subsubsection{Theory of Conditional NRD}

Let us highlight some of the most important techniques in our algebraic approach for studying conditional non-redundancy. First, for a pair of relations $S \subseteq T \subseteq D^r$ we say that an $n$-ary partial function $f \colon X \to D$ (for some $X \subseteq D^n$ called the {\em domain} of $f$) is a {\em partial promise polymorphism} of $(S,T)$ if 
\[(f(s^{(1)}_1, \ldots, s^{(n)}_1), \ldots, f(s^{(1)}_r, \ldots, s^{(n)}_r)) \in T\]
for each sequence $(s^{(1)}_1, \ldots, s^{(1)}_r), \ldots, (s^{(n)}_1, \ldots,  s^{(n)}_r) \in S$ where each involved application is defined.  If we let $\pPol(S,T)$ be the set of all such partial promise polymorphisms we obtain a well-behaved algebraic object which is closed under two natural algebraic notions:  identifying arguments ({\em variable identification minors}) and weakening a function by making it less defined (taking {\em subfunctions}). We call these objects (strong) {\em partial minions} since the corresponding algebraic objects in the total function setting are known as minions~\cite{barto2021Algebraica}. Here, it is interesting to observe that both {\em promise polymorphisms} and {\em partial polymorphisms} are well-studied algebraic objects: the former underpins the algebraic approach for studying {\em promise CSPs}~\cite{austrin2017sat,brakensiek2021Promise,barto2021Algebraica}, and the latter has seen important applications in fine-grained complexity of CSPs~\cite{lagerkvist2021Coarse} and kernelization/non-redundancy~\cite{bessiere2020Chain, carbonnel2022Redundancy, chen2020BestCase, lagerkvist2020Sparsification}. Curiously, the combination is virtually unexplored and only briefly appears in a technical lemma by Pippenger~\cite{pippenger2002}. 

While it is then relatively straightforward to establish that $\pPol(S, \widetilde{T})$ (recall that $\widetilde{T} = D^r \setminus (T \setminus S)$) governs the non-redundancy $\NRD(S \mid T, n)$, obtaining concrete bounds based on completely classifying all partial promise polymorphisms seems daunting. Here, inspired by Lagerkvist \& Wahlstr\"om~\cite{lagerkvist2021Coarse} and Carbonnel~\cite{carbonnel2022Redundancy} we instead concentrate on the aforementioned {\em identities} $\pattern(S, \widetilde{T})$ satisfied by the partial functions in $\pPol(S, \widetilde{T})$. This abstraction allows us to concentrate on the functions in $\pPol(S, \widetilde{T})$ that actually affect $\NRD(S \mid T, n)$ to a significantly large degree and avoids exhaustive characterizations of partial promise polymorphisms.

In the algebraic approach for studying CSPs one often seeks a way to relate functional objects (such as patterns, or polymorphisms) with relational objects that can be used to obtain reductions~\cite{barto2017Polymorphisms}. Often, one wants a {\em Galois connection} between these sets which implies that the objects are dually related, in the sense that small functional objects correspond to large relational objects, and vice versa.
Hence, we set forth to develop a Galois connection between patterns and sets of relations closed under certain restricted first-order formulas. Here, on the relational side we consider conjunctive formulas where atoms are allowed to use functional expressions of bounded arity $c \geq 1$ ($c$-fgppp-definitions).
The reason for allowing such expressions is that they relate NRD within a power of $c$,  which, in the light of Theorem~\ref{thm:intro:every-rational}, seems necessary to capture even a fragment of the rich non-redundancy landscape. On the operational side we then introduce a $c$-ary {\em power} operation on a pattern $P$ ($P^c$). This operation is by design very similar to the power operation in universal algebra, well-known from Birkhoff's HSP theorem~\cite{algebrabook}. But applying it to patterns is strikingly different, and seemingly much more powerful, than to directly apply it to partial functions. We study powers of the cube pattern (from Figure~\ref{fig:3_cube}) in extra detail since the {\em absence} of this pattern gives a general and useful method for proving lower bounds on NRD.

In order to relate invariance under a finite set of patterns to hypergraph Tur{\'a}n problems we first extend the algebraic machinery to {\em multisorted} patterns in order to apply it to $r$-partite CSP instances. The key idea is then to interpret the patterns over the $r$ different sorts in an $r$-partite CSP instance, since the application of such a pattern to constraints in a CSP instance can be used to witness redundancy. For each finite set of $r$-sorted patterns $Q$ we then construct a corresponding set of hypergraphs $\cH(Q)$ (recall Figure~\ref{fig:u_hypergraph}) whose Tur{\'a}n number closely matches the maximum possible non-redundancy of all $r$-ary relation pairs $S \subseteq T$ invariant under $Q$. Choosing the correct set $\cH(Q)$ is not completely trivial and requires subtle algebraic properties of $Q$. We provide several concrete examples of patterns and what this implies for the corresponding Tur{\'a}n problems and note that the assumption that $Q$ is finite is, algebraically speaking, a rather weak condition. This is in contrast to the CSP dichotomy theorem where tractability boiled down to the existence of a single condition, but more in line with tractability results for PCSPs which generally require an infinite family of polymorphisms~\cite{brakensiek2021Promise}. Fundamentally, the technical reason is that there is no general way to {\em compose} functions in this setting while preserving a given relation pair. 
Hence, algebraically oriented methods for analyzing conditional NRD generally require more than a finite set of conditions.

\subsection*{Organization}

In Section~\ref{sec:prelim}, we list some definitions and known results for non-redundancy and related problems. In Section~\ref{sec:frac}, we prove \Cref{thm:intro:every-rational} which shows that non-redundancy can capture any rational exponent at least one. In Section~\ref{sec:nrd-bin}, we classify the non-redundancy of every binary promise language and as corollaries prove \Cref{thm:intro:all-girths} and the upper bound for \Cref{thm:intro:ternary-nrd-near-linear}. In Section~\ref{sec:theory}, we extend the techniques in Section~\ref{sec:frac} and Section~\ref{sec:nrd-bin} and formulate the essential components of a comprehensive algebraic theory of non-redundancy, based around promise partial polymorphisms definable by patterns.
In Section~\ref{sec:linear-nrd}, we build on Section~\ref{sec:theory} to give a number of new insights on the structure of Mal'tsev embeddings, including the theory of Catalan polymorphisms. Here, Theorems~\ref{thm:intro:maltsev-vs-abelian}, \ref{thm:intro:Maltsev-Boolean}, and \ref{thm:intro:ternary-nrd-near-linear} are proved. In Section~\ref{sec:powers-and-fgpp} we complete the algebraic theory with a new power operator and describe the resulting gadget reductions that preserve conditional NRD, thus 
proving \Cref{thm:intro:pattern-NRD} and \Cref{thm:intro:pattern-power}, among other results. 
We also apply the algebraic theory to better understand the connections between non-redundancy and hypergraph Tur{\'a}n theory, and prove e.g.\ \Cref{thm:intro:hypergraph}, as well for consequences for relations invariant under a well-known pattern in universal algebra. In Section~\ref{sec:concl}, we give some concluding remarks and discuss directions for future research. In Appendix~\ref{app:frac}, we give some proofs omitted from Section~\ref{sec:frac}. 

\subsection*{Acknowlegments}
JB and VG are supported in part by a Simons Investigator award and NSF grant CCF-2211972. JB is also supported by NSF grant DMS-2503280. BJ is supported by the Dutch Research Council (NWO) through Gravitation-grant NETWORKS-024.002.003. VL is partially supported by the Swedish research council (VR) under grant VR-2025-04487.

\subsection*{AI Disclosure}

No AI tools were used in the discovery of the results or the exposition of the research.

\section{Preliminaries}\label{sec:prelim}

We begin the main part of the paper by introducing important notation. Throughout, for an $r$-ary tuple $t$ and $i \in [r]$ we write either $t_i$ or $t[i]$ for the $i$th component of $t$, depending on the context. We defer most of the algebraic notions to Section~\ref{sec:theory} where we also restate some basic notions in an algebraic context.

\subsection{Non-redundancy}

We now formally define the main notation we need to discuss non-redundancy. In particular, we adapt the notation used in previous works \cite{bessiere2020Chain,carbonnel2022Redundancy,brakensiek2024Redundancy}.  We use the terms \emph{predicate} and/or \emph{relation} to describe a set $R \subseteq D^r$, where $D$, the \emph{domain}, is finite unless otherwise specified, and $r \ge 1$ is a positive integer. We call the elements $t \in R$ \emph{tuples} and say that $R$ is \emph{non-trivial} if it contains at least one tuple and does not contain all tuples in $D^r$. We consider a \emph{structure} to be a pair\footnote{Usually in the context of CSPs, a structure can have many more symbols, corresponding to a multitude of relations. However, in the context of non-redundancy and related questions one can always without loss of generality consider a single relation, since the non-redundancy of a language $\Gamma$ is within a constant factor of the largest non-redundancy of a single relation $R \in \Gamma$; see for example \cite{carbonnel2022Redundancy,khanna2024Code,brakensiek2024Redundancy}.} of the form $(D, R)$. A \emph{homomorphism} between structures $(D,R\subseteq D^r)$ and $(E,S \subseteq E^r)$ is a map $\sigma : D \to E$ such that for every $t \in R$, we have that $\sigma(t) := (\sigma(t_1), \hdots, \sigma(t_r)) \in S$. If $D$ and $E$ are clear from context, then we say that $\sigma$ is a homomorphism from $R$ to $S$.

Like in \cite{brakensiek2024Redundancy}, we view instances of a CSP as structures themselves. In particular, given a relation $R \subseteq D^r$, we can view any structure $\Psi := (X, Y \subseteq X^r)$ as an \emph{instance} of $\CSP(R)$. We usually call $X$ the \emph{variables} and $Y$ the \emph{clauses}, but we may also refer to $X$ as \emph{vertices} and $Y$ as \emph{(hyper)edges}. We also usually reserve the symbol $n := |X|.$ We let $\sat(R, Y)$ denote the set of homomorphisms from $(X,Y)$ to $(D, R)$ (note that $D$ and $X$ can be inferred from context). More traditionally, $\sat(R, Y)$ can be thought of as the set of assignments $\sigma : X \to D$ such that each clause $y \in Y$ has that $\sigma(y) \in R$. 

We say that a clause $y \in Y$ is \emph{non-redundant} if $\sat(R, Y) \subsetneq \sat(R, Y \setminus \{y\}).$ In other words, there exists a homomorphism $\sigma$ from $Y \setminus \{y\} \to R$ which does not extend to $Y$. We call such a map $\sigma$ a \emph{witness} for $y$. Semantically, this means that the clause $y$ cannot be logically inferred from $Y \setminus \{y\}$. We say that $(X,Y)$ is a non-redundant instance of $\CSP(R)$ if every $y \in Y$ is non-redundant. This leads to our definition of non-redundancy of a predicate.

\begin{definition}[Non-redundancy of a predicate~\cite{bessiere2020Chain}]\label{def:NRD-CSP}
Given a relation $R \subseteq D^r$,  the maximum number of clauses of a non-redundant instance of $\CSP(R)$ on $n$ variables is denoted by $\NRD(R, n)$.
\end{definition}

Often, it is more convenient to work with instances of $\CSP(R)$ which are \emph{multipartite}. More formally, for $R \subseteq D^r$, we say that an instance $(X,Y)$ of $\CSP(R)$ is $r$-partite if there is a partition of the variable set as~$X := X_1 \dot \cup \ldots \dot \cup X_r$, such that for each clause~$(y_1, \ldots, y_r) \in Y$ and each~$i \in [r]$ we have~$y_i \in X_i$. To avoid having to keep track of the exact arity, we say an instance of~$\CSP(R)$ is \emph{multipartite} if it is $r$-partite where~$r$ is the arity of~$R$. As we show below, we can obtain bounds on general non-redundancy by analyzing the non-redundancy of multipartite instances. Therefore it will be convenient to have notation for this quantity.

\begin{definition}[Non-redundancy of multipartite instances]\label{def:NRD-CSP:multipartite}
Given a relation $R \subseteq D^r$ and $n \in \N$, we define $\NRD^*(R, n)$ to be the maximum number of clauses of any \emph{multipartite} non-redundant instance of $\CSP(R)$ on $n$ variables. 
\end{definition}

Clearly $\NRD(R, n) \ge \NRD^*(R,n)$. As observed by \cite{brakensiek2024Redundancy}, the reverse inequality is approximately true.

\begin{lemma}[Implicit in \cite{brakensiek2024Redundancy}] \label{lemma:reduction:to:multipartite}
Let~$R$ be a relation of arity~$r$ over domain~$D$ and let~$\Psi = (X,Y)$ be a non-redundant instance of~$\CSP(R)$. Then there is a non-redundant $r$-partite instance of~$\CSP(R)$ on~$r|X|$ variables and~$|Y|$ clauses. Consequently,~$\NRD(R,n) \leq \NRD^*(R, r\cdot n)$.
\end{lemma}
\begin{proof}
Let~$(X,Y)$ be a non-redundant instance of~$\CSP(R)$. Define $X' := X \times [r]$ and map each $(y_1, \hdots, y_r) \in Y$ to $((y_1,1), \hdots, (y_r,r)) \in Y' \subseteq (X')^r$. It remains to argue that the instance~$(X',Y')$ is non-redundant. For any clause~$y' \in Y'$ corresponding to a clause~$y \in Y$, by the non-redundancy of~$(X,Y)$ there is an assignment~$\sigma_y \colon X \to D$ that satisfies all clauses~$Y \setminus \{y\}$ but not~$y$. It is easy to verify that~$\sigma'_y \colon X' \to D$ defined via~$\sigma'_y((x,i)) = \sigma_y(x)$ satisfies all clauses of~$Y' \setminus \{y'\}$ but not~$y'$. Hence~$(X',Y')$ is a non-redundant instance on~$r|X|$ variables and~$|Y|$ clauses. It follows that the non-redundancy of general instances on~$n$ variables is upper-bounded by the non-redundancy of multipartite instances on~$rn$ variables, so that~$\NRD(R,n) \leq \NRD^*(R,r\cdot n)$.
\end{proof}

We also now observe that if $c > 1$ is a constant, then $\NRD(R, n)$ and $\NRD(R, cn)$ are also related by a constant.

\begin{lemma}[Implicit in e.g., \cite{carbonnel2022Redundancy}]\label{lem:NRD-scale}
For $R \subseteq D^r$ and for any $c > 1$, we have that
\[
    \NRD(R, cn) = O_{c}(\NRD(R, n)).
\]
\end{lemma}
\begin{proof}
Let $(X,Y)$ be a non-redundant instance of $\CSP(R)$ with $|X| = cn$ and $|Y| = \NRD(R, cn)$. Let $X' \subseteq X$ be a random subset of $X$ of size $n$, and let $Y' := Y \cap X'^r$.
If $n$ is sufficiently large, the probability that any $y \in Y$ is also in $Y'$ is approximately $(1-o(1))/c^r$. Thus, in expectation we have that $\NRD(R, n) \ge (1-o(1)) \NRD(R, cn) / c^r$. This proves the asymptotic bound.
\end{proof}

We remark that the analogous bound holds for multipartite instances.

\subsection{Conditional Non-redundancy}

An important concept introduced in \cite{brakensiek2024Redundancy} is the notion of \emph{conditional} non-redundancy, which refines the notion of non-redundancy (Definition~\ref{def:NRD-CSP}) to impose more structure on the witness to a clause being non-redundant. More formally, given relations $P \subsetneq Q \subseteq D^r$, we say that $P \mid Q$ is a \emph{promise relation}. For a structure $(X, Y \subseteq X^r)$, we say that a clause $y \in Y$ is \emph{non-redundant} for $\CSP(P \mid Q)$ if there is a witness $\sigma \in \sat(P, Y \setminus \{y\}) \setminus \sat(P, Y)$ which also satisfies $\sigma \in \sat(Q, Y)$. In other words, although our witness $\sigma$ cannot satisfy $y$ (i.e., $\sigma(y) \notin P$), we want to make sure that $\sigma(y) \in Q$. As such, we call $Q$ the \emph{scaffolding} of $P$.

Analogous to the non-redundancy of non-promise relations, we say that an instance $(X,Y)$ of $\CSP(P \mid Q)$ is non-redundant if every $y \in Y$ is non-redundant. Analogous to Definition~\ref{def:NRD-CSP}, we let $\NRD(P \mid Q, n)$ denote the size of the largest non-redundant instance of $\CSP(P \mid Q)$ on $n$ variables. The multipartite version $\NRD^*(P \mid Q, n)$ can be defined analogously, and we remark that Lemma~\ref{lemma:reduction:to:multipartite} and Lemma~\ref{lem:NRD-scale} can be extended to the conditional NRD setting with the same proof methods.

The key utility in studying the non-redundancy of promise relations is that it satisfies a triangle inequality, which we quote as follows.

\begin{lemma}[Triangle inequality~\cite{brakensiek2024Redundancy}]\label{lem:NRD-chain}
For any predicates $R \subsetneq S \subsetneq T \subseteq D^r$ and $n \in \N$, we have that
\[
    \NRD(R \mid T, n) \le \NRD(R \mid S, n) + \NRD(S \mid T, n).
\]
In particular, if $T = D^r$, we have that
\[
    \NRD(R, n) \le \NRD(R \mid S, n) + \NRD(S, n).
\]
\end{lemma}
We note that the analogous bounds for multipartite instances, obtained by replacing~$\NRD$ by~$\NRD^*$, also hold by the exact same proof. 

\subsection{Embeddings}\label{subsec:embeddings}

Given two structures $(D,R)$ and $(E,S)$, we say that $R$ \emph{embeds} into $S$ if there exists an injective homomorphism $\sigma : D \to E$ from $R$ to $S$ which satisfies the additional property that $\sigma(R) = \sigma(D)^r \cap S$.
As noticed by prior works (e.g., \cite{lagerkvist2020Sparsification,bessiere2020Chain}), if an embedding exists, then any non-redundant instance of $\CSP(R)$ is a non-redundant instance of $\CSP(S)$, so $\NRD(R, n) \le \NRD(S, n)$. As such, we can use embeddings to prove non-redundancy upper bounds.

Often, one seeks to find embeddings into structures $(E,S)$ admitting a special algebraic operation known as a \emph{Mal'tsev polymorphism} (or Mal'tsev term). That is, there is a map $\varphi : E^3 \to E$ such that for any $x,y \in E$, we have that $\varphi(x,y,y) = \varphi(y,y,x) = x$ and further for any $t_1, t_2, t_3 \in S$ we have that $\varphi (t_1, t_2, t_3) \in S$ (where $\varphi$ on tuples is applied componentwise). A canonical example is when $(E, \cdot)$ is a group and $S \subseteq E^r$ is a coset. Then, we can define $\varphi(x,y,z) = x \cdot y^{-1} \cdot z$. As we shall see much later (Section~\ref{sec:linear-nrd}), this canonical example is in a strong sense the \emph{only} example.

Mal'tsev polymorphisms and embeddings imply linear NRD upper bounds, as formalized below. 
Such a bound follows from a characterization made by Bulatov and Dalmau~\cite{bulatov2006Simple} of structures with a Mal'tsev polymorphism.

\begin{theorem}[\cite{lagerkvist2020Sparsification,bessiere2020Chain}]\label{thm:Malt-embedding}
If $R \subseteq D^r$ is a predicate with a Mal'tsev polymorphism, then $\NRD(R, n) \le |D|n.$ Likewise, if $R$ has an embedding into a structure $(E, S)$ with a Mal'tsev polymorphism then $\NRD(R, n) \le |E|n$.
\end{theorem}

We also note that embeddings are themselves a special case of a much broader class of reductions known as \emph{fgpp}-definitions which were first studied by Carbonnel~\cite{carbonnel2022Redundancy} and extended by us in Section~\ref{sec:theory} in the context of conditional non-redundancy.

\section{Every Possible Rational Exponent is Attainable}\label{sec:frac}

A key question about non-redundancy is the set of possible exponents which can appear. More precisely, for a finite domain $D$ and an arity $r \in \N$, we define the \emph{exponent} $\alpha(R)$ of a non-trivial relation $R \subseteq D^r$ to be the least $\alpha \ge 1$ such that $\NRD(R, n) = \Oh(n^{\alpha+o(1)})$. The main result of this section is that any rational number at least $1$ can be the exponent of a finite predicate.

\begin{theorem}\label{thm:frac-nrd}
For every rational number $p/q \ge 1$, there exists a nontrivial relation $\ORDP_{p,q}$ for which $\alpha(\ORDP_{p,q}) = p/q$. In fact, we have that $\NRD(\ORDP_{p,q}, n) = \Theta_{p,q}(n^{p/q})$.
\end{theorem}

\begin{remark}
This result is similar in spirit to the hypergraph Tur{\'a}n result of Frankl \cite{frankl1986All}. 
\end{remark}

\subsection{Defining the Relations}

We start by defining the predicates with the prescribed exponents and their domains.

\begin{definition}
For~$p, q \in \mathbb{N}$ with~$p \geq q$, we define~$D_{p,q} := \{0,1\} \cup \{ (b) \mid b \in \{0,1\}^q \}$.
\end{definition}

Hence the domain~$D_{p,q}$ consists of the Boolean values~$0$ and~$1$, together with one element~$(b)$ for each bitstring of length~$q$. Note that when~$q = 1$, the elements~$(0)$ and~$(1)$ corresponding to bitstrings of length~$1$ are distinct from the elements~$0$ and~$1$ denoting \false and \true. We now define two relations over $D_{p,q}$: $\ORDP_{p,q}$ and $\DP_{p,q}$. Here, $\ORDP_{p,q}$ is the relation referred to in Theorem~\ref{thm:frac-nrd} and $\DP_{p,q}$ is a scaffolding for $\ORDP_{p,q}$ which will assist in upper-bounding $\NRD(\ORDP_{p,q}, n)$ via Lemma~\ref{lem:NRD-chain}.

\begin{definition} \label{def:Rpq}
For~$p, q \in \mathbb{N}$ with~$p \geq q$, we define relations $\ORDP_{p,q} \subsetneq \DP_{p,q} \subseteq D_{p,q}^{p + p^q}$ as follows. Let~$\idx_{p,q} \colon [p^q] \to [p]^q$ be a bijection between the integers~$\{1, \ldots, p^q\}$ and $q$-tuples over~$\{1, \ldots, p\}$. The only tuples contained in $\ORDP_{p,q}$ and $\DP_{p,q}$ are of the form~$T = (x_1, \ldots, x_p, \tx_1, \ldots, \tx_{p^q})$ where~$x_i \in \{0,1\}$ for each~$i \in [p]$ and~$\tx_i \in \{(b) \mid b \in \{0,1\}^q\}$ for each~$i \in [p^q]$.
A tuple~$T$ of this form is contained in~$\DP_{p,q}$ if and only if the following condition is met:
\begin{enumerate}
    \item For each~$i \in [p^q]$, let~$S^{(i)} = (s^{(i)}_1, \ldots, s^{(i)}_q) = \idx_{p,q}(i) \in [p]^q$ be the $q$-tuple indexed by~$i$. The value of~$\tx_i$ in tuple~$T$ should be equal to the bitstring~$(x_{s^{(i)}_1}, \ldots, x_{s^{(i)}_q})$. This constraint acts like a ``direct product," hence the name $\DP_{p,q}$.\label{constraint:dp}
\end{enumerate}
We say that a tuple $T \in \DP_{p,q}$ is also in $\ORDP_{p,q}$ if and only if the following additional condition is met:
\begin{enumerate}[resume]
    \item At least one of the values~$x_i$ for~$i \in [p]$ is equal to~$1$. (So the~$x_i$'s are not all~$0$; this part of the constraint acts as~$\mathsf{OR}_p$.)\label{constraint:or}
\end{enumerate}
\end{definition}
In other words, a tuple~$T$ is contained in~$\ORDP_{p,q}$ if its first~$p$ values are Boolean and at least one of them is~$1$, while the remaining~$p^q$ variables each take a bitstring as value, and the bitstring of variable~$\tx_i$ is exactly the concatenation of the values of the $x$-variables in the positions encoded by~$\idx_{p,q}(i)$. Note that for each tuple~$T \in \ORDP_{p,q}$, the values of the first~$p$ entries of the tuple uniquely determine the value of the remaining~$p^q$ entries. We leverage this ``redundancy'' introduced by the $\tx_i$ variables in order to precisely connect the non-redundancy of the predicate $\OR_p$ with the non-redundancy of our constructed predicate $\ORDP_{p,q}$. In particular, by utilizing the Kruskal-Katona theorem, we will be able to tightly quantify the impact these $\tx_i$ variables have on the redundancy.

\begin{example}
  For~$p=3$ and~$q=2$, the relation~$\ORDP_{3,2}$ has arity $3+3^2=12$. If we project it down to a more manageable arity by only keeping the ``essential'' padding variables for tuples $(1,2)$, $(1,3)$ and $(2,3)$ we get the 6-ary relation
\begin{align*}
     \ORDP_{3,2}' = \{ \quad &
     (1, 0, 0, (10), (10), (00)), \\
     & (0, 1, 0, (01), (00), (10)), \\
     & (0, 0, 1, (00), (01), (01)), \\ 
     & (1, 1, 0, (11), (10), (10)), \\ 
     & (1, 0, 1, (10), (11), (01)), \\ 
     & (0, 1, 1, (01), (01), (11)), \\
     & (1, 1, 1, (11), (11), (11)) \quad \}
 \end{align*}
 when using the bijection~$\idx_{p,q}$ given by listing the size-2 subsets of~$[3]$ in lexicographic order. The relation $\DP_{3,2}$ has the additional tuple $(0, 0, 0, (00), \ldots, (00))$.
\end{example}

We proceed by analyzing the non-redundancy exponent of $\ORDP_{p,q}$. 

\subsection{Lower Bound on \texorpdfstring{$\NRD(\ORDP_{p,q}, n)$.}{NRD(ORDPpq,n)}}

We start with proving a lower bound for the non-redundancy of $\ORDP_{p,q}$, as the result is a bit easier to obtain and illustrates the intent of the construction.

\begin{lemma}\label{lem:lowerbound-nrd}
For each fixed~$p,q \in \mathbb{N}$ with~$p \geq q$, we have that 
\[\NRD(\ORDP_{p,q}, n+n^q)\ge \NRD(\OR_p, n) \ge \binom{n}{p}.\]
Therefore, $\NRD(\ORDP_{p,q}, n) \ge \Omega_{p,q}(n^{p/q}).$
\end{lemma}

\begin{proof}
For arbitrary~$n$, let $(X,Y)$ be the instance of $\CSP(\OR_p)$ with $X = \{x_1, \ldots, x_n\}$ and $Y = \{(x_{i_1}, \ldots, x_{i_p}) \mid i_1 < \ldots < i_p\}$. Then~$|Y| = \binom{n}{p}$ and the instance~$(X,Y)$ is non-redundant, since for any choice of~$i_1 < \ldots < i_p$ the assignment that maps~$x_{i_1}, \ldots, x_{i_p}$ to~$0$ and all other variables to~$1$, satisfies all clauses except~$(x_{i_1}, \ldots, x_{i_p})$. Consider $X' := X \cup X^{q}$. We define a set $Y' \subseteq (X')^{p + p^q}$. For each $y \in Y$, we construct a corresponding $y' \in Y'$ as follows.
\begin{itemize}
    \item if~$i \in [p]$, then $y'_i = y_i$.
    \item if~$i > p$, then~$y'_i = (y_{t_1}, \hdots, y_{t_q}) \in X^q$, where $t_1, \hdots, t_q\in [p]$ are chosen such that
    \[
        (t_1, \hdots, t_q) = \idx_{p,q}(i-p),
    \]
    where we recall that $\idx_{p,q}$ is a bijection between $[p^q]$ and $[p]^q$.
\end{itemize}

It remains to argue that $(X', Y')$ is a non-redundant instance of $\CSP(\ORDP_{p,q})$. Consider each $y \in Y$ and its corresponding $y' \in Y'$. Since $(X,Y)$ is a non-redundant instance of $\CSP(\OR_p)$, there exists an assignment $\sigma_y : X \to \{0,1\}$ which satisfies every clause $z \in Y$ except $y$ itself. We now lift this to a map $\sigma'_{y'} : X' \to D_{p,q}$ as follows:
\begin{itemize}
    \item For each $x \in X' \cap X$, we set $\sigma'_{y'}(x) := \sigma_y(x) \in D_{p,q} \cap \{0,1\}$.
    \item For each $\tx := (x_1, \hdots, x_q) \in X' \cap X^q$, we set $\sigma'_{y'}(\tx) := (\sigma_y(x_1), \hdots, \sigma_y(x_q)) \in D_{p,q} \cap \{(b) \mid b \in \{0,1\}^q\}$.
\end{itemize}
For each $z' \in Y'$ with corresponding $z \in Y$, observe that $\sigma'_{y'}$ satisfies Property~\ref{constraint:dp} of \Cref{def:Rpq} for $z'$ since for every $i \in [p^q]$ with $(t_1, \hdots, t_q) = \idx_{p,q}(i)$, we have that
\[
    \sigma'_{y'}(z'_{i+p}) = (\sigma_y(z'_{t_1}), \hdots, \sigma_y(z'_{t_q})) = (\sigma'_{y'}(z'_{t_1}), \hdots, \sigma'_{y'}(z'_{t_q})).
\]
However, $\sigma'_{y'}$ satisfies Property~\ref{constraint:or} of \Cref{def:Rpq} for $z'$ if and only if $\sigma_y$ satisfies $z$. In other words, $\sigma'_{y'}$ satisfies $z'$ if and only if $z' \neq y'$. Thus, since the choice of $y'$ was arbitrary, we have that $(X', Y')$ is a non-redundant instance of $\CSP(\ORDP_{p,q})$ of size $|Y'| = |Y| = \binom{n}{p}$, as desired. The asymptotic bound on $\NRD(\ORDP_{p,q}, n)$ then immediately follows.
\end{proof}

\begin{remark}
This lower bound technique can be made much more general using the theory of \emph{$c$-fgpp reductions}. See Section~\ref{sec:theory} for more details.
\end{remark}

\subsection{Upper Bound on \texorpdfstring{$\NRD(\ORDP_{p,q}, n)$}{NRD(ORDPpq,n)}}

In this section we prove an upper bound on $\NRD(\ORDP_{p,q}, n)$ using the Kruskal-Katona theorem. We first introduce some relevant terminology.

A \emph{set system} over a universe~$U$ is a collection~$\mathcal{F}$ of subsets of~$U$. A set system is $p$-uniform for~$p \in \mathbb{N}$ if all sets in~$\mathcal{F}$ have cardinality~$p$. The \emph{$q$-shadow} of a $p$-uniform set family~$\mathcal{F}$ over~$U$ is the $q$-uniform set family~$\{X' \in \binom{U}{q} \mid \exists X \in \mathcal{F} \colon X' \subseteq X\}$. So the sets in the shadow of~$\mathcal{F}$ are precisely those sets that can be obtained from a member of~$\mathcal{F}$ by removing $p-q$ elements. Kruskal and Katona proved lower bounds on the size of the $q$-shadow of a uniform set family.

\begin{theorem}[Kruskal-Katona theorem, simplified form~\cite{kruskal,katona}] \label{kruskal-katona}
Let~$p,q \in \mathbb{N}$ with~$p \geq q$. Let~$U$ be a finite set, let~$\mathcal{A}$ be a $p$-uniform set family over~$U$, and let~$\mathcal{B}$ be the $q$-shadow of~$\mathcal{A}$. If~$|\mathcal A| = \binom{x}{p}$, then~$|\mathcal B| \geq \binom{x}{q}$.
\end{theorem}

Using this ingredient, we prove an upper bound on $\NRD(\ORDP_{p,q}, n)$.

\begin{theorem} \label{thm:upperbound-nrd}
For each fixed~$p,q \in \mathbb{N}$ with~$p \geq q$, there exists a constant $C_{p,q} > 0$ such that $\NRD(\ORDP_{p,q}, n) \leq C_{p,q} \cdot n^{\frac{p}{q}}$.
\end{theorem}

We achieve this upper bound by applying the triangle inequality of \Cref{lem:NRD-chain} with $\DP_{p,q}$ as scaffolding. By \Cref{lemma:reduction:to:multipartite} and \Cref{lem:NRD-scale}, it suffices to consider multipartite instances. The bulk of the technical work is then devoted to bounding $\NRD^*(\ORDP_{p,q} \mid \DP_{p,q}, n)$.

\begin{lemma} \label{lem:or-dp:upper}
    For each fixed~$p,q \in \mathbb{N}$ with~$p \geq q$ it holds that $\NRD^*(\ORDP_{p,q} \mid \DP_{p,q}, n) \leq \Oh_{p,q}(n^{\frac{p}{q}})$.
\end{lemma}

\begin{proof}
Let~$c_{p,q} \geq 1$ be any constant such that~$\binom{m}{p} \leq c_{p,q} \cdot \binom{m}{q}^{\frac{p}{q}}$ for all $m \in \N$ with $m \geq p$, which exists since~$\binom{m}{q} \geq 1$ and~$\binom{m}{r} \in \Theta(m^r)$ for all~$r \in \N$. To prove the lemma, we will use induction on the number of variables~$n$ to show that $\NRD^*(\ORDP_{p,q} \mid \DP_{p,q}, n) \leq 2p \cdot c_{p,q} \cdot n^{\frac{p}{q}}$. So consider a multipartite instance~$(X,Y)$ of~$\CSP(\ORDP_{p,q})$ on~$n$ variables that is a non-redundant instance of $\CSP(\ORDP_{p,q} \mid \DP_{p,q}).$ Choose~$X_1 \dot \cup X_2 \dot \cup \ldots \dot \cup X_{p+p^q} = X$ such that~$Y \subseteq X_1 \times X_2 \times \ldots \times X_{p+p^q}$. We will refer to the variables~$\bigcup _{i > p} X_i$ as \emph{padding variables}.

The base case deals with instances on~$n \leq p+p^q$ variables. Since each variable can only appear on a single coordinate in a multipartite instance, such instances have at most one clause. Hence~$|Y| \leq 1 \leq 2p \cdot c_{p,q} \cdot n^{\frac{p}{q}}$.

For the induction step, assume the instance has~$n > p+p^q$ variables and assume that the bound holds for multipartite non-redundant instances with fewer variables. We say that our instance $(X,Y)$ is \emph{reduced} if it has the following property:
\begin{itemize}
    \item \textbf{Reduced:} for any $i \in [p^q]$ with its associated tuple~$(t_1, \hdots, t_q) = \idx_{p,q}(i)$, for any $y,y' \in Y$, the following implication holds: if $y_{i+p} = y'_{i+p}$, then $y_{t_j} = y'_{t_j}$ for all~$j \in [q]$.
\end{itemize}
Next, we show how to derive the desired bound on~$|Y|$ by induction if~$(X,Y)$ is not reduced. Afterwards, we will derive the bound for reduced instances. 

\subparagraph{Non-reduced instances.} Suppose there exists an index~$i \in [p^q]$ and clauses~$y,y' \in Y$ such that~$y_{i+p} = y'_{i+p}$ while $y_{t_j} \neq y'_{t_j}$ for some~$j \in [q]$. So clauses~$y,y'$ use the same padding variable~$y_{i+p} = y'_{i+p}$ that governs a tuple of positions including~$t_j$, while the variables~$y_{t_j}$ and~$y'_{t_j}$ appearing at that position differ. Consider the instance~$(\widetilde{X}, \widetilde{Y})$ of~$\CSP(\ORDP_{p,q})$ derived from~$(X,Y)$ as follows. Let~$\widetilde{X} := X \setminus \{y'_{t_j}\}$ be obtained by removing variable~$y'_{t_j}$ and let~$\widetilde{Y}$ contain a copy~$\widetilde{y}$ of each~$y \in Y$, in which the occurrence of variable~$y'_{t_j}$ is replaced by~$y_{t_j}$ (if it occurs at all).

It is clear that~$|\widetilde{X}| = n-1$. Next we argue that~$|\widetilde{Y}| = |Y|$. Assume for a contradiction that~$|\widetilde{Y}| < |Y|$, which implies that~$Y$ contains two distinct clauses~$y^1, y^2$ such that~$y^2$ is obtained from~$y^1$ by replacing variable~$y'_{t_j}$ (which occurs at position~$t_j$ due to the multipartite property) by variable~$y_{t_j}$. Note in particular that this implies~$y^1_{p+i} = y^2_{p+i}$, so that the two clauses use the same variable at the padding position for tuple~$\idx_{p,q}(i)$. Since~$(X,Y)$ is a non-redundant instance of $\CSP(\ORDP_{p,q} \mid \DP_{p,q})$, there is an assignment~$\sigma$ that $\DP_{p,q}$-satisfies all clauses of~$Y$, while it~$\ORDP_{p,q}$-satisfies~$Y \setminus \{y^1\}$ but not~$y^1$. We now use the structure of the scaffolding predicate~$\DP_{p,q}$ to reach a contradiction. Since~$\sigma$ $\DP_{p,q}$-satisfies~$y^1$, by Property~\ref{constraint:dp} of \Cref{def:Rpq} the $j$-th bit of~$\sigma(y^1_{p+i})$ equals the value of~$\sigma(y^1_{t_j})$. Analogously, the $j$-th bit of~$\sigma(y^2_{p+i})$ equals the value of~$\sigma(y^2_{t_j})$. Since~$y^1_{p+i} = y^2_{p+i}$ we find~$\sigma(y^1_{t_j}) = \sigma(y^2_{t_j})$. As~$\sigma(y^1) \in \DP_{p,q} \setminus \ORDP_{p,q}$, by Property~\ref{constraint:or} of \Cref{def:Rpq} the first~$p$ variables of~$y^1$ evaluate to~$0$ under~$\sigma$. Since~$y^2$ can be obtained from~$y^1$ by replacing~$y_{t_j}$ by~$y'_{t_j}$, which has the same value under~$\sigma$, we find that the first~$p$ variables of~$y^2$ also evaluate to~$0$ under~$\sigma$. But then~$\sigma(y^2) \notin \ORDP_{p,q}$, a contradiction.

From the previous argument we conclude that~$|\widetilde{Y}| = |Y|$. To bound~$|Y|$ in terms of~$|X|$, it therefore suffices to bound~$|\widetilde{Y}|$ in terms of~$|\widetilde{X}|$ via induction. To that end, we argue that~$(\widetilde{X}, \widetilde{Y})$ is again a non-redundant instance of $\CSP(\ORDP_{p,q} \mid \DP_{p,q})$. So consider an arbitrary clause~$\tilde{y}^* \in \widetilde{Y}$ that was obtained from some~$y^* \in Y$ by replacing~$y'_{t_j}$ (if it occurs) by~$y_{t_j}$. Since the original instance~$(X,Y)$ was non-redundant, there is an assignment~$\sigma^*$ that satisfies~$Y \setminus \{y^*\}$ but not~$y^*$. Moreover, since~$(X,Y)$ is a non-redundant instance of $\CSP(\ORDP_{p,q} \mid \DP_{p,q})$, we know~$\sigma^*$ $\DP_{p,q}$-satisfies all clauses in~$Y$. In particular, this means that Property~\ref{constraint:dp} of \Cref{def:Rpq} holds for clauses~$y$ and~$y'$ under~$\sigma^*$. Recall that our choice of~$y,y'$ as a violation of being reduced ensured that~$y_{t_j} \neq y'_{t_j}$ while~$y_{p+i} = y'_{p+i}$, with~$t_j \in \idx_{p,q}(i)$. Property~\ref{constraint:dp} therefore ensures that~$\sigma^*(y_{t_j}) = \sigma^*(y'_{t_j})$.  Consequently,  replacing an occurrence of~$y'_{t_j}$ in a clause of~$Y$ by~$y_{t_j}$ to obtain a clause of~$\widetilde{Y}$ does not affect to which tuple~$\sigma^*$ maps the clause. Since~$\tilde{\sigma}^*$ is just~$\sigma^*$ restricted to the remaining variables, it follows that~$\tilde{\sigma}^*$ is a witness for~$\tilde{y}^*$. This establishes that~$(\widetilde{X}, \widetilde{Y})$ is a non-redundant instance of $\CSP(\ORDP_{p,q} \mid \DP_{p,q})$. Since it has~$n-1$ variables, we may apply induction to conclude that~$|Y| = |\widetilde{Y}| \leq 2p \cdot c_{p,q} \cdot |\widetilde{X}|^{\frac{p}{q}} \leq 2p \cdot c_{p,q} \cdot n^{\frac{p}{q}}$.

\subparagraph{Reduced instances.} Finally, we show how to derive the desired bound on~$|Y|$ when~$(X,Y)$ is reduced. Since~$(X,Y)$ is a non-redundant instance of $\CSP(\ORDP_{p,q} \mid \DP_{p,q})$, for each clause~$y \in Y$ its projection onto the first~$p$ coordinates~$(y_1, \ldots, y_p)$ is unique: if two tuples~$y,y' \in Y$ have the same variables at the first~$p$ positions, it is impossible to~$\DP_{p,q}$-satisfy both~$y,y'$ while~$\ORDP_{p,q}$-satisfying exactly one of them. The only way to falsify~$\ORDP_{p,q}$ while satisfying~$\DP_{p,q}$ is to have the first~$p$ variables map to~$0$. That would, however, also $\ORDP_{p,q}$-violate any other clause with the same first~$p$ variables.

Since~$(X,Y)$ is multipartite, each variable~$x \in X_j$ can only appear at one position~$j$ in clauses of~$Y$. Therefore, the set system~$\mathcal{A} := \{ A_y = \{y_1, \ldots, y_p\} \mid y \in Y\}$ is~$p$-uniform and contains a unique set for each~$y \in Y$. We may therefore bound~$|Y|$ via~$|\mathcal{A}|$. To do so, we consider the $q$-shadow~$\mathcal{B}$ of~$\mathcal{A}$. Let~$X' := X_{p+1} \dot \cup X_{p+2} \dot \cup \ldots \dot \cup X_{p+p^q}$ be the padding variables of~$X$. We argue that~$|\mathcal{B}| \leq |X'|$ by showing that there is an injection~$f \colon \mathcal{B} \to X'$ from~$\mathcal{B}$ to~$X'$. For each set~$B \in \mathcal{B}$, we define its value~$f(B)$ as follows. Since~$\mathcal{B}$ is the $q$-shadow of~$\mathcal{A}$, there exists~$A_y = \{y_1, \ldots, y_p\} \in \mathcal{A}$ with~$B \subseteq A_y$. In fact, there might be multiple such~$A_y$; fix one arbitrarily. Let~$t_1, \ldots, t_q$ be the~$q$ indices for which~$y_{t_j} \in B$, which correspond to the elements of~$A_y$ that generated its $q$-shadow~$B$. Let~$i := \idx_{p,q}^{-1}(t_1, \ldots, t_q) \in [p^q]$ be the index of the $q$-tuple. We set~$f(B) := y_{p+i}$, i.e., we map the $q$-set~$B$ in the shadow to the padding variable in the set~$X'$ that appears at the position of clause~$y$ controlling the tuple of variables that form the shadow. The fact that~$f$ is an injection now follows directly from the fact that~$(X,Y)$ is multipartite and reduced.

Since~$f$ is an injection from~$\mathcal{B}$ to~$X'$ we infer~$|\mathcal{B}| \leq |X'|$: the number of distinct padding variables in the instance is at least as large as the $q$-shadow~$\mathcal{B}$ of~$\mathcal{A}$. Let~$m \geq p$ be the largest integer such that~$|\mathcal{A}| \geq \binom{m}{p}$. This implies that~$|\mathcal{A}| < \binom{m+1}{p} \leq 2p \binom{m}{p}$. By the Kruskal-Katona theorem (\Cref{kruskal-katona}) we have~$|X'| \geq |\mathcal{B}| \geq \binom{m}{q}$. Hence we can now derive:
\begin{align*}
    |Y| = |\mathcal{A}| \leq 2p \cdot \binom{m}{p} \leq 2p \cdot c_{p,q} \cdot \binom{m}{q}^{\frac{p}{q}} \leq 2p \cdot c_{p,q} \cdot |X'|^{\frac{p}{q}} \leq 2p \cdot c_{p,q} \cdot n^{\frac{p}{q}},
\end{align*}
which concludes the proof.
\end{proof}

To facilitate a bound on~$\NRD^*(\ORDP_{p,q}, n)$ via conditional non-redundancy, we now derive a bound on~$\NRD^*(\DP_{p,q}, n)$ by essentially showing it is an affine constraint.

\begin{lemma} \label{lem:nrd:or-dp-star:linear}
    For each fixed~$p,q \in \mathbb{N}$ with~$p \geq q$ it holds that $\NRD^*(\DP_{p,q}, n) \leq \Oh_{p,q}(n)$.
\end{lemma}
\begin{proof}
We prove this with an application of \Cref{thm:Malt-embedding}. More concretely, the projection constraints defining $\DP_{p,q}$ can be viewed as projection constraints inside of the Abelian group $G := ((\Z/2\Z)^{q+2}, +)$. We define a coset $H \subseteq G^{p + p^q}$ as follows. We define $H$ to be the set of $(h^1, \hdots, h^{p+p^q}) \in G^{p+p^q}$ satisfying the following conditions:
\begin{itemize}
\item[(i)] If $i \in [p]$, then $h^i = (0,0,\hdots, 0)$ or $h^i = (1,0,0,\hdots, 0)$.
\item[(ii)] If $i \in [p^q]$, then $h^{i+p} = (0,1,x_1, \hdots, x_q)$ for some $x_1, \hdots, x_q \in \Z/2\Z$.
\item[(iii)] For any $i \in [p^q]$ and $(t_1, \hdots, t_q) = \idx_{p,q}(i)$, we have for every $j \in [q]$ that $h^{i+p}_{j+2} = h^{t_j}_1.$
\end{itemize}
It is straightforward to verify that $H$ is a coset. Consider the map $\sigma : D_{p,q} \to (\Z/2\Z)^{q+2}$ defined as follows:
\begin{itemize}
\item For $b \in D_{p,q} \cap \{0,1\}$, we set $\sigma(b) = (b, 0, 0, \hdots, 0)$.
\item For $(b_1, \hdots, b_q) \in D_{p,q}$, we set $\sigma(b_1, \hdots, b_q) = (0,1,b_1, \hdots, b_q)$.
\end{itemize}
By inspection, properties (i)--(iii) ensure that $\sigma(\DP_{p,q}) = H \cap \sigma(D_{p,q})^{p+p^q}.$ In particular, for any $T = (x_1, \hdots, x_p, \tx_1, \hdots, \tx_{p^q}) \in H \cap \sigma(D_{p,q})^{p+p^q}$, we have that (i) ensures that $x_1, \hdots, x_p$ correspond to $\sigma(0)$ or $\sigma(1)$, (ii) ensures that $\tx_1, \hdots, \tx_{p^q}$ each correspond to some $\sigma(b_1, \hdots, b_q)$, and (iii) ensures Property~\ref{constraint:dp} is satisfied in the definition of $\DP_{p,q}$ (\Cref{def:Rpq}). Therefore, $\sigma$ is an embedding of $\DP_{p,q}$ into $H$. 

Thus, since $H$ is an Abelian coset (and thus has a Mal'tsev polymorphism), by \Cref{thm:Malt-embedding} we have that $\NRD^*(\DP_{p,q}, n) \leq \NRD^*(H, n)\leq \Oh_{p,q}(n)$, as desired.
\end{proof}

By using the established upper-bounds together with the triangle inequality for conditional non-redundancy and the reduction to multipartite instances, the desired bound on~$\NRD(\ORDP_{p,q}, n)$ now follows easily.

\begin{lemma} \label{lem:upperbound:rpq}
    For each fixed~$p,q \in \mathbb{N}$ with~$p \geq q$ it holds that $\NRD(\ORDP_{p,q}, n) \leq \Oh_{p,q}(n^{\frac{p}{q}})$.
\end{lemma}
\begin{proof}
    The lemma follows from the following derivation:
    \begin{align*}
        \NRD&(\ORDP_{p,q},n)\\
        &\leq \NRD^*(\ORDP_{p,q}, (p+p^q) n) & \text{By \Cref{lemma:reduction:to:multipartite}.} \\
        &\leq \NRD^*(\ORDP_{p,q} \mid \DP_{p,q}, (p+p^q) n) + \NRD^*(\DP_{p,q}, (p+p^q) n) & \text{By \Cref{lem:NRD-chain}.} \\
        &\leq \Oh_{p,q}(((p+p^q)n)^{\frac{p}{q}}) + \Oh_{p,q}((p+p^q) n) & \text{\Cref{lem:or-dp:upper} and~\ref{lem:nrd:or-dp-star:linear}.} \\
        &\leq \Oh_{p,q}(n^{\frac{p}{q}}). & \qedhere
    \end{align*}
\end{proof}

The proof of \Cref{thm:frac-nrd} now follows directly from the lower bound in \Cref{lem:lowerbound-nrd} with the matching upper bound in \Cref{lem:upperbound:rpq}.

\subsection{Connections to CSP Kernelization}\label{subsec:kernel}

The methods used to prove \Cref{thm:frac-nrd} can be used to prove an even stronger fact: for every $p \ge q$, there exists a constraint language $\SATDP_{p,q}$ such that instances of $\CSP(\SATDP_{p,q})$ on $n$ variables can be efficiently kernelized to have at most $\Oh(n^{p/q})$ clauses. The constraint language $\SATDP_{p,q}$ consists of two relations: the aforementioned $\ORDP_{p,q}$ as well as $\CUT := \{(0,1),(1,0)\} \subseteq \{0,1\}^2 \subseteq D_{p,q}^2$ which is simply the negation relation on the Boolean values~$0$ and~$1$.

The reason we need to add this extra relation is that $\CSP(\ORDP_{p,q})$ has a near-trivial polynomial time algorithm: if the Boolean variables and the padding variables (as defined in the proof of \Cref{thm:upperbound-nrd}) overlap, we know the instance is unsatisfiable. Otherwise, we can assign $1$ to each Boolean variable and $(1,\hdots, 1)$ to each padding variable. Note that instances of $\CSP(\OR_{p})$ have the analogous issue.

By introducing $\CUT$, one can then reduce CNF-$p$-SAT to $\CSP(\SATDP_{p,q})$, allowing the application of the kernelization lower bound for CNF-$p$-SAT for $p\ge 3$ by Dell and van Melkebeek~\cite{dell2014Satisfiability}. See Appendix~\ref{app:ker} for the full details on porting \Cref{thm:frac-nrd} to CSP kernelization.

\subsection{An Alternative Construction with Smaller Arity}\label{subsec:smaller}

We conclude this section by noting that \Cref{thm:frac-nrd} can be proved for predicates of (much) smaller arity by \emph{puncturing} $\ORDP_{p,q}$ (i.e., dropping coordinates). For any $p, q \in \N$ with $p \ge q$, let $E_{p,q} := \{0,1\}^q$. Let $\cF \subseteq \binom{[p]}{q}$ be a family of sets. Define the operation $\proj_{\cF} : \{0,1\}^p \to E_{p,q}^{\cF}$ as follows
\[
    \proj_{\cF}(x) := (x|_{S} : S \in \cF).
\]
Given this operation, we can define a pair of predicates analogous to $\ORDP_{p,q}$ and $\DP_{p,q}$ from \Cref{def:Rpq}.
We define the following pair of predicates.
\begin{align*}
    \ORDP_{\cF} &:= \left\{\proj_{\cF}(x) : x \in \OR_p \right\} \subset E_{p,q}^{\cF}\\
    \DP_{\cF} &:= \left\{\proj_{\cF}(x) : x \in \{0,1\}^p \right\} \subset E_{p,q}^{\cF}.
\end{align*}
Unlike $\ORDP_{p,q}$ and $\DP_{p,q}$, $\ORDP_{\cF}$ and $\DP_{\cF}$ no longer have Boolean variables keeping track of the underlying Boolean assignment. This allows for the predicates to have even smaller arity (with a tradeoff of a somewhat trickier proof).

We derive non-redundancy bounds for a specific class of set families $\cF$. We say that $\cF$ is \emph{$q/p$-regular} if there exists a probability distribution over $\cF$ such that for every $i \in [p]$, the probability that a sample $S$ from $\cF$ satisfies $i \in S$ is exactly $q/p$. Our bound is as follows.

\begin{theorem}\label{thm:frac-nrd-alt}
For any $p \ge q \ge 1$ and for any $q/p$-regular $\cF \subseteq \binom{[p]}{q}$, $\NRD(\OR_{\cF}, n) = \Theta_q(n^{p/q})$.
\end{theorem}
The most notable difference between the proof of \Cref{thm:frac-nrd} and \Cref{thm:frac-nrd-alt} is that the former uses Kruskal-Katona in the upper bound while the latter uses Shearer's inequality. Details are available in Appendix~\ref{app:frac-alt}.%

\section{Graph Girth and Binary Conditional NRD Classification}\label{sec:nrd-bin}

One of the first classification theorems in the context of non-redundancy is the classification of \emph{binary} relations. In particular, given any relation $R \subseteq D^2$, we either have that $\NRD(R, n) = O_D(n)$ or $\NRD(R, n) = \Omega_D(n^2)$. This was proved in the context of non-redundancy by \cite{bessiere2020Chain} and an analogous result was proved in the context of sparsification by \cite{butti2020}.

In this section, we seek to generalize these results to the setting of \emph{conditional} non-redundancy. In particular, given a domain $D$ and $R \subsetneq S \subseteq D^2$, we seek to classify the possible asymptotics of $\NRD(R \mid S, n)$. As we shall see, unlike the dichotomy theorem for non-promise binary relations, there are infinitely many possible exponents possible for $\NRD(R \mid S, n)$, corresponding to the possible girths a bipartite graph may have.
This classification is given in Section~\ref{subsec:binary-class}.

As an application, we construct two families of (non-promise) ternary predicates. The first family of predicates constructed in Section~\ref{subsec:inf-ternary} proves that infinitely many NRD exponents are possible for ternary predicates--in contrast to the constructions in Section~\ref{sec:frac} which only produce finitely many examples for each arity. The second family of predicates, constructed in Section~\ref{subsec:approach-linear}, shows that ternary predicates which lack of Mal'tsev embedding can have non-redundancies getting arbitrarily close to linear.

\subsection{Classification of Binary Conditional NRD}\label{subsec:binary-class}

Recall that $\NRD(R \mid S, n)$ captures the maximum possible size of a non-redundant instance of $\CSP(R)$ where the non-satisfying assignments must come from $S \setminus R$ (i.e., all witnessing solutions are satisfying assignments to $\CSP(S)$). First, we reduce to the case that $|S \setminus R| = 1$.

\begin{proposition}\label{prop:scaffold-one}
For any $R \subsetneq S \subseteq D^r$, we have that
\[
    \max_{t \in S \setminus R} \NRD(R \mid R \cup \{t\}, n) \le \NRD(R \mid S, n)  \le |D|^r\max_{t \in S \setminus R} \NRD(R \mid R \cup \{t\}, n).
\]
\end{proposition}
\begin{proof}
  The lower bound is trivial. For the upper bound, consider any non-redundant instance $(X,Y \subseteq X^r)$ of $\CSP(R \mid S)$ and note that we can partition the $y \in Y$ based on which $t_y \in S \setminus R$ is assigned by the witness corresponding to $y$. If we partition $Y$ into $|S \setminus R|$ parts according to the value of $t_y$, then
  $(X, Y)$ can be viewed as the disjoint union of non-redundant instances of $\CSP(R \mid R \cup \{t\})$ for $t \in S \setminus R$.
\end{proof}

We now work toward describing the binary characterization. Given an integer $k \ge 2$, we let $C_{2k} \subseteq \{0,1,\hdots, k-1\}^2$ denote the $2k$-cycle. More explicitly, we use a bipartite representation where the two copies of the domain $\{0,\ldots,k-1\}$ represent different vertices, and define
\[
    C_{2k} := \{(0,0),(0,1),(1,0),(1,2),(2,1), \hdots, (k-2,k-1),(k-1,k-2),(k-1,k-1)\}.
\]
We also let $C^*_{2k} := C_{2k} \setminus \{(0,0)\}$. 
\begin{remark}\label{rem:tC2k}
Equivalently, the $2k$-cycle can be represented by the following predicate:
\[
    \widetilde{C}_{2k} := \{(0,0),(0,1),(1,1),(1,2),(2,2), \hdots, (k-1,k-1),(k-1,0)\}
\]
with $\widetilde{C}^*_{2k} = \widetilde{C}_{2k} \setminus \{(0,0)\}.$ This latter representation will be of use in Section~\ref{subsec:approach-linear}. Since the bipartite graphs $C_{2k}$ and $\widetilde{C}_{2k}$ are isomorphic, one can show that $\NRD(C^*_{2k} \mid C_{2k}, n) = \Theta(\NRD(\widetilde{C}^*_{2k} \mid \widetilde{C}_{2k}, n))$. Such equivalences are formalized in the theory of fgppp-definitions in Section~\ref{sec:theory}.
\end{remark}

We now classify the (bipartite) non-redundant instances of $\CSP(C^*_{2k} \mid C_{2k})$. See Figure~\ref{fig:girth} for an illustration of the proof technique.

\begin{lemma}\label{lem:cycle-nrd}
For any $k \ge 2$, a bipartite graph $G = (A,B,E\subseteq A \times B)$ is a non-redundant instance of $\CSP(C^*_{2k} \mid C_{2k})$ iff $G$ has girth at least $2k$.
\end{lemma}

\begin{remark}
Note that since $C_{4} = \{0,1\}^2$, a non-redundant instance of $\CSP(C^*_{4} \mid C_4)$ is precisely a non-redundant instance of $\CSP(\OR_2)$. Hence, that particular case was already handled by previous works (e.g., \cite{filtser2017Sparsification,bessiere2020Chain}).
\end{remark}

\begin{proof}
First assume that $G$ is a non-redundant instance but has a cycle of length $\ell < 2k$. Since $G$ is bipartite, $\ell$ is even. Thus, there exists an injective homomorphism $f : C_{\ell} \to G$. Since $G$ is non-redundant, for each $e \in C_{\ell}$, there exists a map $g_e : G \to C_{2k}$ such that $(g_e\circ f)^{-1}((0,0)) = \{(0,0)\}$. However, since $\ell < 2k$, every homomorphism from $C_{\ell}$ to $C_{2k}$ is constant, so we have a contradiction. Thus, every non-redundant instance of $\CSP(C^*_{2k} \mid C_{2k})$ has girth at least $2k$.

Now, assume that $G$ has girth at least $2k$. Pick any edge $e := (x_0, y_0)$ of $G$, we seek to construct a homomorphism $f_e : G \to C_{2k}$ such that $f_e^{-1}(0,0) = \{e\}$. For every $v \in A \cup B$ (the vertex set of $G$), let $d_e(v)$ be the distance from $v$ to $\{x_0, y_0\}$ in $G$. In particular, $d_e(x_0) = d_e(y_0) = 0$. Any vertex in a different connected component than $\{x_0,y_0\}$ is assigned a distance of $\infty$. Our homomorphism $f_e$ is then
\[
    f_e(v) := \min (d_e(v), k-1).
\]
To see why this is valid, consider any edge $(x,y)$ of $G$. If $d_e(x),d_e(y) \ge k-1$, we are fine as $(f_e(x), f_e(y)) = (k-1,k-1) \in C_{2k}$. Otherwise, note that by the triangle inequality $|d_e(x) - d_e(y)| \le 1$. If $|d_e(x) - d_e(y)| = 1$, then we have that $(f_e(x), f_e(y)) \in C_{2k}$. Otherwise, assume for sake of contradiction that $d_e(x) = d_e(y) \le k-2$. In that case, there must be a trail from $x$ to $\{x_0,y_0\}$ to $y$ back to $x$ of length $d_e(x) + d_e(y) + 2 \le 2k-2$. Since $d_e(x) = d_e(y)$, the shortest paths from $x$ and $y$ to $\{x_0,y_0\}$ cannot use the edge $(x,y)$. Thus, this trail must include a simple cycle of length at most $2k-2$, a contradiction.   
\end{proof}

As an immediate corollary of Lemma~\ref{lem:cycle-nrd} and \Cref{lemma:reduction:to:multipartite}, we have that the non-redundancy of $\CSP(C^*_{2k} \mid C_{2k})$ is asymptotically equal to the maximum size of bipartite graphs of girth at least $2k$, which up to a constant factor the maximum size of any graph of girth at least $2k$.

\begin{corollary}\label{cor:C2k}
For all $k\ge 2$,
\[
 \NRD(C^*_{2k} \mid C_{2k}, n) = \Theta_D(\ex(n, \{C_3, C_4, C_5, \hdots, C_{2k-1}\})).
\]
In particular, by the estimates of \cite{furedi2013history}, we have that there exists a universal constant $c > 0$ such that.
\[
\NRD(C^*_{2k} \mid C_{2k}, n) \in [\Omega_D(n^{1+\frac{c}{k-1}}), O_D(n^{1+\frac{1}{k-1}})].
\]
\end{corollary}

\begin{remark}
It is a long-standing conjecture in combinatorics that the constant $c$ can be taken to be $1$ for all $k \ge 2$, but this is only known to be true for $k \in \{2,3,4,6\}$. See Section 4 of \cite{furedi2013history} for a detailed discussion.
\end{remark}

We now complete the classification by showing that \Cref{lem:cycle-nrd} captures (essentially) all cases.

\begin{theorem}\label{thm:binary-class}
For every $R \subsetneq S \subseteq D^2$, we either have that $\NRD(R \mid S, n) = O_D(n)$, or there exists $k\ge 2$ such that 
\[
\NRD(R \mid S, n) = \Theta_D(\NRD(C^*_{2k} \mid C_{2k}, n)).
\]
\end{theorem}

\begin{proof}
By Proposition~\ref{prop:scaffold-one}, we may assume without loss of generality that $S \setminus R = \{t_0\}$. Let $(x_0, y_0) := t_0$ and let $\ell$ be the length of the shortest path from $x_0$ to $y_0$ in $R$ (when viewed as the adjacency matrix of a bipartite graph). Since $t_0 \not\in R$, we either have that $\ell = \infty$ (i.e., a path does not exist), or $\ell$ is an odd integer at least $3$.

If $\ell = \infty$, consider a bipartite non-redundant instance $G := (A, B, E \subseteq A \times B)$ of $\CSP(R \mid S)$. For each edge $e := (x,y) \in E$, there must exist a homomorphism $f_e : A \cup B \to D$ from $G$ to $S$ for which $f_e^{-1}((x_0,y_0)) = \{(x, y)\}$. Since $x_0$ and $y_0$ are not connected in $R$, we have that if we delete $e$ from $G$, then $x$ and $y$ are in separate connected components. Therefore, $G$ must be acyclic. By \Cref{lemma:reduction:to:multipartite}, we have that $\NRD(R \mid S, n) = \Oh(n)$.

Otherwise, assume that $\ell$ is an odd integer at least $3$. Thus, there exists an integer $k\ge 2$ such that $\ell = 2k-1$. We claim that $\NRD(R \mid S, n) = \Theta_D(\NRD(C^*_{2k} \mid C_{2k}, n)).$ By definition of $k$, there exists an injective homomorphism $h : C_{2k} \to S$ such that $h(0,0) = t_0$. Thus, every non-redundant instance of $\NRD(C^*_{2k} \mid C_{2k}, n)$ is a non-redundant instance of $\NRD(R \mid S, n)$.

Conversely, let $G := (A, B, E \subseteq A \times B)$ be a bipartite non-redundant instance of $\NRD(R \mid S, n)$. For any $e := (x,y) \in G$, consider a homomorphism $f_e : G \to S$ such that $f_e^{-1}(t_0) = \{e\}$. Since the shortest path from $x_0$ to $y_0$ in $R$ is at most $\ell = 2k-1$, any path from $x$ to $y$ in $G \setminus \{e\}$ must also have length at least $2k-1$. That is, every cycle of $G$ containing $e$ has length at least $2k$. Since $e$ is arbitrary, $G$ has girth at least $2k$.

Therefore, by Lemma~\ref{lem:cycle-nrd}, every \emph{bipartite} non-redundant instance of $\CSP(R \mid S)$ is a non-redundant instance of $\CSP(C^*_{2k} \mid C_{2k})$. Therefore, by \Cref{lemma:reduction:to:multipartite}, we have that $\NRD(R \mid S, n) = \Theta(\NRD(C^*_{2k} \mid C_{2k}, n)).$
\end{proof}

\begin{remark}
The correspondence between $(R,S)$ and $(C^*_{2k}, C_{2k})$ in the proof of \Cref{thm:binary-class} can be viewed as an \emph{fgppp-equivalence}. See Section~\ref{sec:theory} for more details.
\end{remark}

\subsection{Infinitely Many NRD Exponents for Ternary Predicates}\label{subsec:inf-ternary}

As an immediate application of Lemma~\ref{lem:cycle-nrd}, we can prove that infinitely many exponents are possible for (non-promise) ternary predicates. For $k \ge 2$, let $S_{2k} := C_{2k} \times \{0,1\} \subseteq \{0,1,\hdots, k-1\}^3$ and $R_{2k} := S_{2k} \setminus \{(0,0,0)\}$. We now prove the following

\begin{lemma}\label{lem:cond-R-S}
For all $k \ge 2$, $\NRD(R_{2k} \mid S_{2k}, n) = \Theta(n \cdot \NRD(C^*_{2k} \mid C_{2k}, n))$.
\end{lemma}
\begin{proof}
We divide the proof into the lower bound and the upper bound on $\NRD(R_{2k} \mid S_{2k}, n)$.
\paragraph{Lower Bound.} First, we show that $\NRD(R_{2k} \mid S_{2k}, 2n) \ge \Omega(n \cdot \NRD(C^*_{2k} \mid C_{2k}, n))$ -- note that $\NRD(R_{2k} \mid S_{2k}, 2n) = \Theta(\NRD(R_{2k} \mid S_{2k}, n))$ by Lemma~\ref{lem:NRD-scale}.  In particular, by \Cref{lemma:reduction:to:multipartite} there exists a bipartite non-redundant instance  $G = (A_1, A_2, E)$ of $\CSP(C^*_{2k} \mid C_{2k})$ with $\Omega(\NRD(C^*_{2k} \mid C_{2k}, n))$ edges. Let $A_3$ be a set of size $n$ disjoint from $A_1$ and $A_2$. Let $H = (A_1, A_2, A_3, E \times A_3)$ be a tripartite hypergraph with $\Omega(n \cdot \NRD(C^*_{2k} \mid C_{2k}, n))$ hyperedges. We claim that $H$ is a non-redundant instance of $R_{2k} \mid S_{2k}$.

Fix a hyperedge $(x_1, x_2, x_3) \in E \times A_3$. Since $G$ is a non-redundant instance of $\CSP(C^{*}_{2k} \mid C_{2k})$, we have that there exists a map $f_{x_1,x_2} : A_1 \cup A_2 \to \{0,  1, \hdots, k-1\}$ such that for all $e \in E$, $f_{x_1,x_2}(e) \in C_{2k}$ and $f_{x_1,x_2}(e) = (0,0)$ iff $e = (x_1, x_2)$. Consider then the map $g_{x_1,x_2,x_3} : A_1 \cup A_2 \cup A_3 \to \{0, 1, \hdots, k-1\}$ defined by
\[
    g_{x_1,x_2,x_3}(y) = \begin{cases}
    f_{x_1,x_2}(y) & y \in A_1 \cup A_2\\
    0 & y = x_3\\
    1 & y \in A_3 \setminus \{x_3\}
    \end{cases}.
\]
In particular, for any $e := (y_1, y_2, y_3) \in E \times A_3$, we have that $(g_{x_1,x_2,x_3}(y_1), g_{x_1,x_2,x_3}(y_2)) = (f_{x_1,x_2}(y_1), f_{x_1,x_2}(y_2)) \in C_{2k}$ and $g_{x_1,x_2,x_3}(y_3) \in \{0,1\}$ so $g_{x_1,x_2,x_3}(e) \in S_{2k}$. Furthermore, $g_{x_1,x_2,x_3}(e) = (0,0,0)$ if and only if $(y_1,y_2) = (x_1,x_2)$ and $y_3 = x_3$. Thus, $H$ is indeed a non-redundant instance of $\CSP(R_{2k} \mid S_{2k})$.

\paragraph{Upper Bound.} Next, we show that $\NRD(R_{2k} \mid S_{2k}, n) \le n \cdot \NRD(C^*_{2k} \mid C_{2k}, n)$. By \Cref{lemma:reduction:to:multipartite}, it suffices consider a tripartite non-redundant instance $H := (A_1, A_2, A_3, F \subseteq A_1 \times A_2 \times A_3)$ of $\CSP(R_{2k} \mid S_{2k})$. Fix $x_3 \in A_3$ and let $E_{x_3} := \{(x_1,x_2) : (x_1, x_2, x_3) \in F\}$. We seek to prove that $G_{x_3} := (A_1, A_2, E_{x_3})$ is a non-redundant instance of $\CSP(C^*_{2k} \mid C_{2k})$. This suffices to prove the upper bound as then
\[
|F| \le \sum_{x_3 \in A_3} |E_{x_3}| \le \sum_{x_3 \in A_3} \NRD(C^*_{2k} \mid C_{2k}, n) = n \cdot \NRD(C^*_{2k} \mid C_{2k}, n).
\]
In particular, consider $(x_1, x_2) \in E_{x_3}$. Since $H$ is a tripartite non-redundant instance of $\CSP(R_{2k} \mid S_{2k})$, there exists a homomorphism $f_e : H \to S_{2k}$ which sends only $(x_1,x_2,x_3)$ to $(0,0,0)$. In particular,  $f_e(x_3) = 0$. Therefore, for any other $(x,y) \in E_{x_3}$, we must have $(f_e(x), f_e(y)) \in C^*_{2k}$. Thus, $G_{x_3}$ is a non-redundant instance of $\CSP(R_{2k} \mid S_{2k})$.
\end{proof}

Furthermore, it is not hard to bound the non-redundancy of $S_{2k}$.

\begin{proposition}\label{prop:S-n^2}
For all $k\ge 2$, $\NRD(S_{2k}, n) = \Oh(n^2)$.
\end{proposition}
\begin{proof}
Let $B = \{0,1\}$ be the unary relation restricting a predicate to be $0$ or $1$. By definition of $S_{2k}$, we have the logical equivalence
\[
    S_{2k}(x_1, x_2, x_3) \equiv C_{2k}(x_1,x_2) \wedge B(x_3).
\]
That is, $\{C_{2k}, B\}$ fgpp-defines $S_{2k}$. By a bound of Carbonnel (cf.~\cite[Prop.~2]{carbonnel2022Redundancy}), we have that $\NRD(S_{2k}, n) \le \NRD(C_{2k}, n) + \NRD(B, n) = \Oh(n^2),$ as $C_{2k}$ is a binary constraint and $B$ is a unary constraint.
\end{proof}

Combining \Cref{lem:cond-R-S} and \Cref{prop:S-n^2} using the triangle inequality (\Cref{lem:NRD-chain}) and invoking \Cref{cor:C2k}, we can deduce the main result of this section.

\begin{theorem}\label{thm:inf-ternary-exp}
For all $k \ge 2$, $\NRD(R_{2k}, n) = \Theta(n \cdot \NRD(C^*_{2k} \mid C_{2k}, n)) \in [n^{2+\frac{\Omega(1)}{k-1}}, n^{2+\frac{1}{k-1}}]$.
\end{theorem}

\begin{remark}
Note that $\NRD(R_6, n) = \Theta(n^{2.5})$. In particular, $2.5$ is currently the largest known non-redundancy exponent less than $3$ for a ternary predicate. By a result of Carbonnel~\cite{carbonnel2022Redundancy}, such an exponent must be less than $2.75$.
\end{remark}

\subsection{Ternary Predicates Approaching Linear NRD}\label{subsec:approach-linear}

\begin{figure}
    \centering
    \begin{tikzpicture}[every node/.style={font=\small}]
    \tikzstyle{vertex}=[inner sep=0pt]
    \tikzstyle{hyperedge}=[thin,blue]
\node[vertex] (a1) at (7,2) {$\bullet$};
\node[vertex, right=of a1] (b1) {$\bullet$};
\node[vertex, right=of b1] (c1) {$\bullet$};
\node[vertex, right=of c1] (d1) {$\bullet$};
\node[vertex, right=of d1] (e1) {$\bullet$};

\node[vertex, below=of a1] (a2) {$\bullet$};
\node[vertex, right=of a2] (b2) {$\bullet$};
\node[vertex, right=of b2] (c2) {$\bullet$};
\node[vertex, right=of c2] (d2) {$\bullet$};
\node[vertex, right=of d2] (e2) {$\bullet$};

\node[vertex, below=of a2] (a3) {$\bullet$};
\node[vertex, right=of a3] (b3) {$\bullet$};
\node[vertex, right=of b3] (c3) {$\bullet$};
\node[vertex, right=of c3] (d3) {$\bullet$};
\node[vertex, right=of d3] (e3) {$\bullet$};

\draw[thick,color=red] (a1.west) -- (a2.west) -- (a3.west);

\draw[hyperedge] (a1.east) -- (b2.west) -- (e3.west);

\draw[hyperedge] (b1.west) -- (b2.east) -- (d3.west);

\draw[hyperedge] (b1.east) -- (c2.west) -- (c3.west);

\draw[hyperedge] (c1.west) -- (c2.east) -- (b3.east);

\draw[hyperedge] (c1.east) -- (d2.west) -- (a3.east);

\draw[hyperedge] (d1.west) -- (d2.east) -- (e3.east);

\draw[hyperedge] (d1.east) -- (e2.west) -- (d3.east);

\draw[hyperedge] (e1.east) -- (e2.east) -- (c3.east);

\draw[hyperedge] (e1.west) -- (a2.east) -- (b3.west);
    \end{tikzpicture}
    \caption{A hypergraph representation of $\CYC_5$, with the red hyperedge being the one to delete to get $\CYCs_5$. Note that if we delete any row from this figure, we get a graph isomorphic to $C_{10}$.}
    \label{fig:CYC-m}
\end{figure}

To conclude this section, we give another ternary adaptation of $\CSP(C^*_{2k} \mid C_{2k})$. However, this time the non-redundancies are approaching linear. For odd $m \in \N$ define
\[
    \CYC_m := \{(x,y,z) \in (\Z/m\Z)^3 : x+y+z = 0 \wedge y-x \in \{0,1\}\}.
\]
Further define $\CYCs_{m} := \CYC_m \setminus \{(0,0,0)\}$.  See Figure~\ref{fig:CYC-m} for an illustration. We seek to prove the following
\begin{theorem}\label{thm:CYC-upper}
For all odd $m \ge 3$, we have that $\NRD(\CYCs_m, n) \le \Oh_m(n^{\frac{m}{m-1}}).$ 
\end{theorem}

\begin{remark}
This proof never uses the fact that $m$ is odd. We state it this way because the corresponding ``lower bound'' of \Cref{thm:CYC-Mal'tsev} only holds for odd $m$. In fact, for even $m \ge 2$, one can show that $\NRD(\CYCs_m) = \Oh_m(n)$ (e.g., $\CYCs_2$ is $2$-in-$3$-SAT).
\end{remark}

\begin{proof}
Similar to the proof of \Cref{thm:inf-ternary-exp}, we split the proof of our upper bound into two parts. First, we show that $\NRD(\CYCs_m\mid \CYC_m, n) = \Oh_m(n^{\frac{m}{m-1}})$. Second, we show that $\NRD(\CYC_m, n) = \Oh_m(n)$. These imply the upper bound on $\NRD(\CYCs_m, n)$ by \Cref{lem:NRD-chain}.

\paragraph{Step 1.} By \Cref{lemma:reduction:to:multipartite}, it suffices to show that a tripartite non-redundant instance $H := $ $(A_1, A_2,$ $A_3,$ $F \subseteq A_1 \times A_2 \times A_3)$ of $\CSP(\CYCs_m \mid \CYC_m)$ has at most $\Oh_m(n^{\frac{m}{m-1}})$ hyperedges.

First, notice that $H$ is a \emph{linear} hypergraph; that is, no two hyperedges overlap in exactly two vertices. Assume for sake of contradiction that $e = (x_1, x_2, x_3)$ and $f = (y_1, y_2, y_3)$ have two vertices in common. Since $H$ is non-redundant, there must be a map that sends $e$ to $(0,0,0)$ but $f$ to some element of $\CYCs_m$. However, $\CYCs_m$ lacks any elements with two or more zeros, a contradiction.

Since $H$ is a linear hypergraph, we may define $E := \{(x_1, x_2) : \exists x_3 \in A_3, (x_1, x_2, x_3) \in F\}$ and note that $|E| = |F|$. Thus, by \Cref{cor:C2k} and \Cref{rem:tC2k}, it suffices to prove that $G := (A_1, A_2, E)$ is a non-redundant instance of $\CSP(\widetilde{C}_{2m}^* \mid \widetilde{C}_{2m})$.

For any $e := (x_1, x_2, x_3) \in F$ let $f_e : A_1 \cup A_2 \cup A_3 \to \{0,1,\hdots, m-1\}$ be the homomorphism which maps $H$ to $\CYC_m$ but $e$ is the only hyperedge mapped to $(0,0,0)$. Note that $\widetilde{C}^*_{2m} = \{(z_1,z_2) : \exists z_3, (z_1,z_2,z_3) \in \CYC_m\}$. Therefore, $f_e$ restricted to $A_1 \cup A_2$ is a homomorphism from $G$ to $\widetilde{C}_{2m}$ which only maps $(x_1,x_2)$ to $(0,0)$. Thus, $G$ is a non-redundant instance of $\CSP(\widetilde{C}_{2m}^* \mid \widetilde{C}_{2m})$, as desired.

\paragraph{Step 2.} To finish, we show that $\NRD(\CYC_m, n) = \Oh_m(n)$. We do this by constructing an explicit Abelian embedding of $\CYC_m$. Let $D := \Z/m\Z$ and consider the group $G = (\Z/2\Z)^D$ which we identify with the set of subsets of $D$ equipped with the symmetric difference operator $\oplus$. Furthermore, we identify each element $i \in D$ with its singleton $\{i\} \in G$. 

Let $Q$ be the subgroup of $G^3$ generated by the elements of $\CYC_m$. By definition, $\CYC_m \subseteq Q \cap D^3$. We claim that $\CYC_m = Q \cap D^3$. Consider any triple $(a_1,a_2,a_3) \in D^3$ such that there exists a subset $S \subseteq \CYC_m$ such that 
\[
    (\{a_1\}, \{a_2\}, \{a_3\}) = \bigoplus_{(b_1,b_2,b_3) \in S} (\{b_1\}, \{b_2\}, \{b_3\}).
\]
A key observation is the following. For any $T \subseteq \widetilde{C}_{2m}$, consider the analogous expression $(B_1, B_2) := \bigoplus_{(b_1, b_2) \in T} (\{b_1\}, \{b_2\})$. If we interpret $T$ as a collection of paths in $\widetilde{C}_{2m}$, then $B_1$ and $B_2$ are precisely the sets of endpoints of the paths (restricted to the two halves of the bipartite graph). Thus, since $\CYC_m$ restricted to its first two coordinates is precisely $\widetilde{C}_{2m}$, the projection of $S$ to its first two coordinates must correspond precisely to a path between $a_1$ (on the left) and $a_2$ (on the right) in $\widetilde{C}_{2m}$. Since each tuple of $\CYC_m$ is uniquely determined by its first two coordinates, we have narrowed down $S$ to the following two possibilities:
\begin{align*}
S_0 &:= \{(a_1, a_1, -2a_1), (a_1-1, a_1, -2a_1+1), (a_1-1, a_1-1, -2a_1+2), \hdots, (a_2, a_2, -2a_2)\}\\
S_1 &:= \{(a_1, a_1+1, -2a_1-1), (a_1+1, a_1+1, -2a_1-2), \hdots, (a_2-1, a_2, -2a_2+1)\}
\end{align*}
One can verify that $\bigoplus_{(b_1,b_2,b_3) \in S_i} (\{b_1\}, \{b_2\}, \{b_3\})$ gives the same value for $i \in \{0,1\}$. Note that $S_0 \sqcup S_1$ is a partition of $\CYC_m$. Thus, at least one of $S_0$ or $S_1$ has cardinality at most $m$. Notice that the third coordinate of $S_0$ increases by $1$ for each successive tuple, and the third coordinate of $S_1$ decreases by $1$ for each successive tuple. Thus, for the set $S_i$ of size at most $m$, the expression $\bigoplus_{(b_1,b_2,b_3) \in S_i} (\{b_1\}, \{b_2\}, \{b_3\}).$ will sum $|S_i|$ distinct values in the third coordinate. The only way this can equal $\{a_3\}$ is for $|S_i| = 1$ and $S_i = \{(a_1, a_2, a_3)\}$. Thus, $(a_1,a_2,a_3) \in \CYC_m$, as desired. 

Thus, we have proven that $\CYC_m$ has an Abelian group embedding, so $\NRD(\CYC_m, n) \le \Oh_m(n)$.
\end{proof}

\begin{remark}
The argument used in Step 1 can be interpreted as a special case of Theorem~\ref{thm:fgppp-nrd} with the fgppp-definition of $(\CYCs_{m}\mid\CYC_m)(x_1,x_2,x_3) \equiv (\widetilde{C}^*_{2m}\mid \widetilde{C}_{2m})(x_1,x_2)$.
\end{remark}

\begin{remark}
One can check that $\CYCs_{3} = \{111,222,012,120,201\}$ which is equal to the $\BCK$ predicate first considered by \cite{bessiere2020Chain}. The best upper bound on $\NRD(\BCK, n)$ in the literature is $\Oh(n^{1.5}\log n)$ by \cite{brakensiek2024Redundancy}. Theorem~\ref{thm:CYC-upper} mildly improves this upper bound to $\Oh(n^{1.5})$.
\end{remark}

Sadly, we do not know of a nontrivial lower bound for the non-redundancy $\CYCs_m$. However, we give evidence that a nontrivial lower bound should exist by proving that $\CYCs_m$ lacks a Mal'tsev embedding.

\begin{theorem}\label{thm:CYC-Mal'tsev}
For all odd $m \ge 3$, $\CYCs_m$ lacks a Mal'tsev embedding. 
\end{theorem}

The proof of Theorem~\ref{thm:CYC-Mal'tsev} is rather technical as it first requires developing new properties about the operators in an algebra with a Mal'tsev term. As such, we defer the proof of Theorem~\ref{thm:CYC-Mal'tsev} to Section~\ref{subsec:excl}.

\section{Towards An Algebraic Approach: Partial Promise Polymorphisms and Patterns} 
\label{sec:theory}

In this section we take the first steps for developing a novel algebraic theory for studying conditional non-redundancy. It turns out that the fundamental algebraic object are {\em partial polymorphisms} between relation pairs $(S,T)$. This polymorphism notion merges and generalizes the following two notions.

\begin{enumerate}
    \item 
    Polymorphisms between relation pairs $(A,B)$ that have been fundamental in studying the approximation variant of CSP known as {\em promise CSP}~\cite{brakensiek2021Promise,barto2021Algebraica}.
    \item 
    Partial polymorphisms of relations that have been used to study fine-grained complexity questions~\cite{lagerkvist2021Coarse}, kernelization~\cite{lagerkvist2020Sparsification}, and non-redundancy~\cite{bessiere2020Chain,carbonnel2022Redundancy}.
\end{enumerate}

We begin in Section~\ref{subsec:galois} by proving that our partial polymorphism notion on the functional side yields sets of partial functions closed under the formation of {\em minors}. On the relational side we instead get {\em quantifier-free} primitive positive definitions, and we obtain a full Galois correspondence between these notions. 
Via this connection it is straightforward to show that partial promise polymorphisms  govern conditional non-redundancy. 
Moving up one abstraction level we in Section~\ref{sec:patterns} concentrate on {\em identities}, or {\em patterns}, satisfied by partial promise polymorphisms.

    \subsection{A Galois Connection} \label{subsec:galois}
    We begin by developing a Galois connection between sets of partial functions closed under the formation of {\em minors} and sets of relations closed under a certain type of conjunctive formulas. Throughout, we assume that all domains $D_1$ and $D_2$ are finite, although many concepts generalize to e.g.\ {\em $\omega$-categorical} structures.
    A {\em relation pair} over finite domains $D_1$ and $D_2$ is a tuple $(S, T)$ where $S \subseteq D_1^r$ and $T \subseteq D_2^r$  for some $r \geq 0$ (note that $D_1^0 = D_2^0 = \emptyset$). If, in addition,  there exists a homomorphism $h \colon D_1 \to D_2$, i.e., $h(t) \in T$ for each $t \in S$ then $(S,T)$ is said to be a {\em promise relation}, and is sometimes written $(S,T, h)$. A {\em promise language} is a set of relation pairs where there exist $h \colon D_1 \to D_2$ such that $(S,T, h)$ is a promise relation for each relation pair $(S,T)$, and a {\em language} is simply a set of relation pairs. Note that all relation pairs in a language might not have the same arity. In the forthcoming applications we often assume that $D_1 \subseteq D_2$ in which case $h$ can always be defined as the identity function $h(x) = x$ over $D_1$. If the domain for a promise relation $(S,T)$ is not specified we always assume that $S \subseteq T$. We often write $(\cla, \clb)$ for a language and thus write $(S,T) \in (\cla, \clb)$ to denote a specific relation pair.

Let $D_1$ and $D_2$ be two finite sets. An $n$-ary {\em partial function} from $D_1$ to $D_2$ is a map $f \colon X \to D_2$ for $X \subseteq D_1^n$, and we write $\dom(f) = X$ for the set of values where $f$ is defined. We frequently express $x \notin \dom(f)$ by $f(x) = \bot$. A {\em subfunction} of $f$ is a partial function $g$ where $\dom(g) \subseteq \dom(f)$ and such that $g(\mathbf{x}) = f(\mathbf{x})$ for every $\mathbf{x} \in \dom(g)$.

We write $\pi^n_i$ for the $i$-ary {\em projection}, in approximation contexts often called {\em dictator} satisfying $\pi^n_i(\mathbf{x}) = \mathbf{x}_i$. The domain is usually taken from the context and not explicitly denoted. A {\em partial projection} (or {\em partial dictator}) is a subfunction of a projection. When $D_1 = D_2 = D$ then it is natural to combine functions in the following sense: 
for an $n$-ary partial function $f$ and $m$-ary partial functions $g_1, \ldots, g_n$ over the same universe $D$ we define composition of $f, g_1, \ldots, g_n$ as \[(f \circ (g_1, \ldots, g_n))(x_1, \ldots, x_m) = f(g_1(x_1, \ldots, x_m), \ldots, g_n(x_1, \ldots, x_m))\] for all sequences of arguments $x_1, \ldots, x_m \in A$ such that (1) each $(x_1, \ldots, x_m) \in \dom(g_i)$ and (2) $g_1(x_1, \ldots, x_m), \ldots, g_n(x_1, \ldots, x_m) \in \dom(f)$. Composition of total functions is then just a special case when all involved functions $f, g_1, \ldots, g_n$ are total, i.e., $\dom(f) = D^n$ and each $\dom(g_i) = D^m$. For a set of (total or partial) functions $F$ over $A$ we write $[F]$ for the smallest set of (total or partial) functions over $A$ which contains $F$, is closed under composition, and which contains all (total or partial) projections. Such sets are sometimes called (strong partial) {\em clones} and naturally correspond to the term algebra induced by a given algebra.

When working with partial functions from $D_1$ to $D_2$ it is generally not possible to compose functions, and the relevant algebraic closure condition is instead formation of {\em partial minors}.

\begin{definition}
A $n$-ary (variable-identification) {\em minor} $f_{/h}$ of an $m$-ary partial function $f \colon X \to D_2$, where $h \colon [m] \to [n]$, is the unique partial function $f_{/h} : D_1^n \to D_2 \cup \{\perp\}$ satisfying $f_{/h}(x_1, \ldots, x_n) = f(x_{h(1)}, \ldots, x_{h(m)})$ for all $x_1, \hdots, x_n \in D_1$.    
\end{definition}

We now define a {\em partial minor}, often just called a {\em minor}, $g$ of $f$ as a subfunction of $g'$ for a variable identification minor $g'$ of $f$ (observe that $g'$ is technically a subfunction of itself, so any variable identification minor is also a minor). 

\begin{definition}
A set of partial functions $F$ from $D_1$ to $D_2$ is said to be a {\em strong partial minion} if it is closed under minors, i.e., if $f \in F$ and $g$ is a minor of $f$ then $g \in F$. We write $\minor{F}$ for the smallest strong partial minion that contains $F$.
\end{definition}

Note that the subfunction condition in the definition of a partial minor clearly sets apart strong partial minions from minions of total functions and prevents the former to be defined under the unifying umbrella of {\em abstract minions}~\cite{brakensiek2020Powera}.
However, we will see that many notions carry over to the partial setting and that strong partial minions can be described via relations.
Thus, we say that an $n$-ary partial function $f \colon \dom(f) \to D_2$ is a {\em partial polymorphism} of an $r$-ary relation pair $(S,T)$ if, for each sequence $x^{(1)}, \ldots, x^{(n)} \in S$, either

\begin{enumerate}
    \item 
    there exists $1 \leq i \leq r$ such that $(x^{(1)}_i, \ldots, x^{(n)}_i) \notin \dom(f)$, or
    \item 
    $(f(x^{(1)}_1, \ldots, x^{(n)}_1), \ldots, f(x^{(1)}_r, \ldots, x^{(n)}_r)) \in T$.
\end{enumerate}

Similarly, $f$ is a partial polymorphism of a language $(\cla, \clb)$ if it is a partial polymorphism of every $(S,T) \in (\cla, \clb)$, in which case we also say that $(\cla, \clb)$ is {\em invariant} under $f$.

\begin{definition}
For a language $(\cla, \clb)$ we let $\pPol(\cla, \clb)$ be the set of all partial polymorphisms from $\cla$ to $\clb$.
\end{definition}

If $(\cla, \clb) = \{(S,T)\}$ is singleton then we write $\pPol(S,T)$ rather than $\pPol(\{(S,T)\})$, and if $\cla = \clb$, simply $\pPol(\cla)$. Let us also remark that sets of the form $\pPol(\cla)$ are closed under functional composition as well as subfunctions.

\begin{definition}
    For a set of partial functions $F$ from $D_1$ to $D_2$ we let \[\inv(F) = \{(S,T) \mid r \geq 0, S \subseteq D_1^r, T \subseteq D_2^r, F \subseteq \pPol(S,T)\}.\]
\end{definition}

We observe that $\inv(F)$ is a promise language if $\minor{F}$ contains at least one unary, total function.

Let us now relate the $\inv(\cdot)$ and $\pPol(\cdot)$ operators together as follows. To simplify the definition we temporarily work in the setting of relational structures in a prescribed signature $\tau$ containing relation symbols and their associated arities.
A {\em quantifier-free primitive positive formula} (qfpp-formula) over a relational signature $\tau$ is a conjunctive formula $\varphi(x_1, \ldots, x_r)$ with free variables $x_1, \ldots, x_r$ where each atom is of the form $R_i(\mathbf{x}^i)$ (for some $(R_i, r_i) \in \tau$ and $r_i$ tuple of variables $\mathbf{x}^i$), an equality of the form $x = y$, for $x,y \in \{x_1, \ldots, x_r\}$, or $\mathsf{f}(x_i)$ for some $x_i \in \{x_1, \ldots, x_r\}$, where $\mathsf{f}$ is a special notation for the empty unary relation. Each combination of a qfpp-formula $\varphi(x_1, \ldots, x_n)$ and a $\tau$-structure $\cla$ induces a relation $\{(a_1, \ldots, a_r) \mid \varphi^{\cla}(a_1, \ldots, a_r)\}$, where equality $=$ is interpreted in the obvious way, and where $\sf{f}$ is interpreted as $\emptyset$. The sole reason for introducing $\sf{f}$ is that the relation pair $(\emptyset, \emptyset)$ is invariant under every partial polymorphism and thus should always be definable. In the literature one sometimes assumes the existence of a symbol $\sf{t}$ always interpreted as the full relation, but note that we can simulate a constraint $\sf{t}$ via an equality constraint $x = x$. 

We lift this notion to relation pairs as follows.
Let $(\cla, \clb)$ be a language over $D_1, D_2$ and $(\cla', \clb')$ a language over $D_1, D_2$. We say that $(\cla', \clb')$ is 

\begin{enumerate}
    \item {\em quantifier-free primitive positive definable} (qfpp-definable) in $(\cla, \clb)$ if, for each relational symbol $R$ (of arity $r$) of $(\cla', \clb')$, there is a qfpp-formula $\varphi_R$  such that
    \[
    R^{\cla'} = \{(a_1, \ldots, a_{r}) \in D_1^{r} \mid \varphi^{\cla}_R(a_1, \ldots, a_{r})\},
    \]  
    and, similarly, 
    \[
        R^{\clb'} = \{(a_1, \ldots, a_{r}) \in D_2^{r} \mid \varphi^{\clb}_R(a_1, \ldots, a_{r})\},  
    \]

    \item a {\em strict relaxation} of $(\cla, \clb)$ if $\cla, \cla', \clb, \clb'$ are similar relational structures and, for every relational symbol $R$,
    \[
    R^{\cla'} \subseteq R^{\cla} \quad \text{and} \quad R^{\clb} \subseteq R^{\clb'}.
    \]
\end{enumerate}

More generally $(\cla, \clb)$ is {\em quantifier-free primitive positive promise definable} (qfppp-definable) in $(\cla', \clb')$ if the former can be obtained by a sequence of qfpp-definitions and strict relaxations of the latter.
We write $\qfppp{\cla, \clb}$ for the smallest set of relation pairs containing $(\cla, \clb)$ and which is closed under qfppp-definitions, and similarly to the $\pPol(\cdot, \cdot)$  operator write $\qfppp{S,T}$ for a singleton language $(\cla, \clb) = \{(S,T)\}$. 

\begin{lemma} \label{lemma:is_a_minion}
    Let $(\cla, \clb)$ be a language. Then 
    $\pPol(\cla, \clb)$ is a strong partial minion ($\pPol(\cla, \clb) = \minor{\pPol(\cla, \clb)}$).
\end{lemma}

\begin{proof}
The inclusion $\pPol(\cla, \clb) \subseteq \minor{\pPol(\cla, \clb)}$ is trivial since every partial function is a minor of itself. For the other direction, let $g$ be an $n$-ary minor of an $m$-ary $f \in \pPol(\cla, \clb)$, i.e., $g$ is a subfunction of a variable-identification minor $f_{/h}$ of $f$ for some $h \colon [m] \to [n]$. For each $(S,T) \in (\cla, \clb)$ it follows that each application $f_{/h}(t^{(1)}, \ldots, t^{(m)})$ for $t^{(1)}, \ldots, t^{(m)} \in S$ is equivalent to $f$ applied as $f(t^{h(1)}, \ldots, t^{h(m)})$, which is either undefined or included in $T$. Hence, $f_{/h}$ preserves $(S,T)$. Similarly, it is easy to see that any subfunction $g$ of $f_{/h}$ also preserves any relation invariant under $f_{/h}$, and we conclude that $\pPol(\cla, \clb) = \minor{\pPol(\cla, \clb)}$.
\end{proof}

Note that if $(\cla, \clb)$ over $D_1, D_2$ is a promise language, witnessed by a function $h \colon D_1 \to D_2$, then $h \in \pPol(\cla, \clb)$, but if $(\cla, \clb)$ is not a promise language then we are not guaranteed the existence of a total, unary function in $\pPol(\cla, \clb)$.

\begin{lemma} \label{lemma:is_qfpp_closed}
Let $D_1$ and $D_2$ be two finite sets. Then $\qfppp{\inv(F)} = \inv(F)$ for any set $F$ of partial functions from $D_1$ to $D_2$. 
\end{lemma}

\begin{proof}
    The inclusion $\inv(F) \subseteq \qfppp{\inv(F)}$ is trivial since any relation can qfpp-define itself. For the other direction, consider relation symbols $R$ of arity $r$ and $R'$ of arity $r'$, interpreted by $R^{\cla}, R'^{\cla}$ in $\cla$ over $D_1$ and $R^{\clb}, R'^{\clb}$ over $D_2$, and the following qfpp-definitions $\varphi$.
    \begin{enumerate}
    \item 
     $R(x_1, \ldots, x_r) \land \mathsf{f}(x_i)$  for some $x \in [r]$. In this case the statement is trivial.
    \item
      $R(x_1, \ldots, x_r) \land (x_i = x_j)$ for $i,j \in [r]$. Then $\varphi^{\cla} = \{(d_1, \ldots, d_r) \mid (d_1, \ldots, d_r) \in S, d_i = d_j\}$  and $\varphi^{\clb} = \{(d_1, \ldots, d_r) \mid (d_1, \ldots, d_r) \in R^{\clb}, d_i = d_j\}$. Consider an $n$-ary $f \in F$ and an application $f(t^{(1)}, \ldots, t^{(n)}) = t$ for $t^{(1)}, \ldots, t^{(n)} \in \varphi^{\cla}$ where every application $f(t^{(1)}_k, \ldots, t^{(n)}_k)$, $1 \leq k \leq r$, is defined. Observe that $t_i = t_j$, i.e., the $i$th element of $t$ equals the $j$th element of $t$, which implies that $t \in \varphi^{\clb}$ (since $t \in R^{\clb}$).
    \item
      $R(x_{h(1)}, \ldots, x_{h(r)})$ for $h \colon [r] \to [s]$ for some $s \geq 1$. We sketch the most important arguments. Note that one for each $(d_1, \ldots, d_s) \in \varphi^{\cla}$ (respectively, $\varphi^{\clb}$) can define a corresponding tuple $(d_{h(1)}, \ldots, d_{h(r)}) \in \varphi^{\cla}$ (respectively, $\varphi^{\clb}$). Thus, if $f(t^{(1)}, \ldots, t^{(n)}) = t$  for some $n$-ary $f \in F$ and $t^{(1)}, \ldots, t^{(n)} \in \varphi^{\cla}$ then one can define tuples which when applied to $f$ yields a tuple in $\varphi^{\clb}$.
    \item
      $R(x_{h(1)}, \ldots, x_{h(r)}) \land R'(x_{h'(1)}, \ldots, x_{h'(r')})$ for $h \colon [r] \to X$ and $h' \colon [r'] \to Y$. First, assume that $X \cap Y = \emptyset$, i.e., the two constraints have no overlapping variables. Then $\varphi^{\cla}$ (respectively, $\varphi^{\clb}$) is simply the Cartesian product of two relations that are invariant under $F$ from the preceding case, and the claim easily follows. Similarly, if $X = Y$ then the resulting relation is the intersection of two invariant constraints. The general case follows via similar arguments. 
    \end{enumerate}
This generalizes to arbitrary qfpp-definitions via a simple inductive argument.
Last, consider a strict relaxation $(S',T')$ of some $(S,T) \in \inv(F)$, i.e., $S' \subseteq S$ and $T' \supseteq T$. Assume there exists an $n$-ary $f \in F$ and tuples $t^{(1)}, \ldots, t^{(n)} \in S' \subseteq S$ such that $f(t^{(1)}, \ldots, t^{(n)}) = t \notin T'$. It follows that $f$ does not preserve $(S,T)$, either, since $t^{(1)}, \ldots, t^{(n)} \in S \supseteq S'$ and since $t \notin T \subseteq T'$.
\end{proof}

With similar arguments we obtain the following two characterizations.

\begin{lemma} \label{lemma:ppol_preserves_qfpp}
Let $(\cla, \clb)$ be a language. Then $\pPol(\qfppp{\cla, \clb}) = \pPol(\cla, \clb)$.
\end{lemma}

\begin{lemma} \label{lemma:minors_preserve_relations}
Let $F$ be a set of partial functions over $A,B$. Then $\inv(F) = \inv(\minor{F})$.
\end{lemma}

To complete our description we additionally need the following two lemmas.

\begin{lemma} \label{lemma:pippenger1} 
  Let $F$ be a strong partial minion over domains $D_1$, $D_2$. For any partial function $g \colon X \to D_2$ (for some $n \geq 1$ and $X \subseteq D_1^n$) where $g \notin F$ there 
  \begin{enumerate}
  \item 
        exists a relation pair $(S,T)$ such that $F \subseteq \pPol(S,T)$ and $g \notin \pPol(S,T)$, and
      \item 
      if $F$ contains a unary, total function $h$ then $(S,T)$ is a promise relation.
  \end{enumerate}
\end{lemma}

\begin{proof}
  First, consider the relation $S = \{t^{(1)}, \ldots, t^{(n)}\}$ whose arguments (in, say, lexicographical order) enumerates the domain $X$ of $g$, i.e., for the $i$th tuple $t \in X$ we have that $(t^{(1)}_i, \ldots, t^{(n)}_i) = t$. Second, consider the relation $T$ formed by, for every $f \in F$ of arity $m$, and all tuples $t^{(i_1)}, \ldots, t^{(i_m)} \in S$, adding $f(t^{(i_1)}, \ldots, t^{(i_m)})$ whenever each application is defined. 

  By definition, $g$ cannot preserve $(S,T)$ since it would then be equivalent to a partial function in $F$. Similarly, by construction, $(S,T)$ is invariant under every function in $F$. Note that if $F$ contains a unary, total function $h$, then by construction, $h$ is a homomorphism from $S$ to $T$, and $(S,T)$ is thus a promise relation.
\end{proof}  

\begin{lemma} \label{lemma:pippenger2}
    Let $(\cla, \clb)$ be a  language. Let $(S,T)$ be a  relation pair over $A,B$ such that $(S,T) \notin \qfppp{\cla, \clb}$. Then there exists $f \in \pPol(\cla, \clb)$ such that $f \notin \pPol(S, T)$.
\end{lemma}

\begin{proof}
This follows from Proposition 2.5 in Pippenger~\cite{pippenger2002} since the construction does not use existential quantification.
\end{proof}

We may now conclude that sets of the form $\inv(\pPol(\cla, \clb))$ are precisely the sets containing $(\cla, \clb)$ closed under qfppp-definitions, and, similarly, that all strong partial minions $\minor{F}$ can be described as the set of polymorphisms of the relations invariant under $F$ ($\pPol(\inv(F))$).

\begin{theorem}
    Let $(\cla, \clb)$ be a language over $D_1, D_2$. Then $\qfppp{\cla, \clb} = \inv(\pPol(\cla, \clb))$.
\end{theorem}

\begin{proof}
The first direction follows from Lemma~\ref{lemma:ppol_preserves_qfpp}. Second, let $(S,T) \in \inv(\pPol(\cla, \clb))$. If $(S,T) \notin \qfppp{\cla, \clb}$ then Lemma~\ref{lemma:pippenger2} implies that there is $f \in \pPol(\cla, \clb)$ such that $f \notin \pPol(S,T)$, contradicting that $(S,T)$ is invariant under every partial function in $\pPol(\cla, \clb)$.
\end{proof}

\begin{theorem}
Let $F$ be a set of partial functions from $D_1$ to $D_2$. Then $\minor{F} = \pPol(\inv(F))$.
\end{theorem}

\begin{proof}
First, let $f \in \minor{F}$. Then $f$ preserves every relation in $\inv(F)$ (by Lemma~\ref{lemma:minors_preserve_relations}). Second, let $f \in \pPol(\inv(F))$. If $f \notin \minor{F}$ then Lemma~\ref{lemma:pippenger1} implies that there is $(S,T)$ such that $\minor{F} \subseteq \pPol(S,T)$ but $f \notin \pPol(S,T)$, i.e., $(S,T) \in \inv(F)$, contradicting the assumption that $f \in \pPol(\inv(F))$.
\end{proof}

By combining everything together thus far we obtain a Galois connection between $\inv(\cdot)$ and $\pPol(\cdot, \cdot)$ which implies the following duality.

\begin{theorem} \label{thm:galois}
    Let  $(\cla, \clb)$ and  $(\cla', \clb')$ be languages over $D_1, D_2$. Then \[\pPol(\cla, \clb) \subseteq \pPol(\cla', \clb') \]
if and only if 
\[ (\cla', \clb') \subseteq \qfppp{\cla, \clb}.\]
\end{theorem}

\begin{proof}
      First, assume that $\pPol(\cla, \clb) \subseteq \pPol(\cla', \clb')$. Assume there exists $(S,T) \in (\cla', \clb')$ such that $(S,T) \notin \qfppp{\cla, \clb}$. By Lemma~\ref{lemma:pippenger2} there exists $f \in \pPol(\cla, \clb)$ such that $f \notin \pPol(S,T) \supseteq \pPol(\cla', \clb')$, contradicting the assumption that $\pPol(\cla, \clb) \subseteq \pPol(\cla', \clb')$.

        For the other direction, assume that $(\cla', \clb') \subseteq \qfppp{\cla, \clb}$. Assume there exists $f \in \pPol(\cla, \clb)$ such that $f \notin \pPol(\cla', \clb')$, and let $(S,T) \in (\cla', \clb')$  be a promise relation witnessing this, i.e., $f \notin \pPol(S,T))$. 
  By assumption, $(S,T) \in \qfppp{\cla, \clb}$, and 
  Lemma~\ref{lemma:ppol_preserves_qfpp} then implies that $\pPol(\cla, \clb) = \pPol(\qfppp{\cla, \clb}) \subseteq \pPol(S,T)$, contradicting the assumption that $f \in \pPol(\cla, \clb)$.
\end{proof}

Recall that for a promise relation $(S, T)$ over $D$ we write $\widetilde{T}$ for  $D^r \setminus (T \setminus S)$. We can now prove that the non-redundancy of a promise relation is determined by the partial polymorphisms from $S$ to $\widetilde{T}$, i.e., we are interested in partial polymorphisms $f$ such that if $f(t_1, \ldots, t_n) = t$ is defined then $t \in S$ {\em if} $t \in T$. 

\begin{proposition}\label{prop:gadget-nrd}
Let $(S_1, T_1)$ and $(S_2, T_2)$ be promise relations over a domain $D$. If  $\pPol(S_2, \widetilde{T_2}) \subseteq \pPol(S_1, \widetilde{T_1})$ then $\NRD(S_1 \mid T_1, n) = \Oh(\NRD(S_2 \mid T_2, n))$.
\end{proposition}

\begin{proof}
We provide a proof sketch since this proposition is superseded by Theorem~\ref{thm:fgppp-nrd} in Section~\ref{sec:gadget_reductions}, and ignore handling relaxations and equality constraints. First, via Theorem~\ref{thm:galois} the condition that $\pPol(S_1, \widetilde{T_1}) \subseteq \pPol(S_2, \widetilde{T_2})$ implies that $(S_2, T_2)$ qfpp-defines $(S_1, T_1)$. 
Consider a non-redundant instance $(X,Y)$ of $\CSP(S_1 \mid T_1)$ on $n$ variables. We replace every constraint of $S_1$ with the corresponding constraints of $S_2$, as prescribed by the witnessing qfpp-definition. For each constraint $y \in Y$, there then exists an assignment $\sigma_y : X \to D$ for which $\sigma_y(y') \in T_1$ for all $y' \in Y \setminus \{y\}$ but $\sigma_y(y) \in T_1 \setminus S_1$. 

In particular, $\sigma_y$ will $S_2$-satisfy every constraint which replaces each $y' \in Y \setminus \{y\}$. However, for the clauses replacing $y$, we have that
\[
    0 = (D^r \setminus (T_1 \setminus S_1))(\sigma_y(y)) = \bigwedge_{i=1}^t (D^r \setminus (T_2 \setminus S_2))(\sigma_y(y|_{S_i})).
\]
Thus, complementing both sides, we then have that
\[
    1 = (T_2 \setminus S_2)(\sigma_y(y)) = \bigvee_{i=1}^t (T_1 \setminus S_1)(\sigma_y(y|_{S_i}))
\]
and at least one of the constraints replacing $y$ will have $\sigma_y$ assign it to some value in $T_1 \setminus S_1$. Pick one such clause arbitrarily for each $y \in Y$. The $\sigma_y$'s witness that the resulting instance is non-redundant for $\CSP(S_1 \mid T_1)$.
\end{proof}

From now on, we primarily work in the signature-free setting and when specifying that an $r$-ary relation $R$ over $D$ is definable via a combination of a qfpp-formula $\varphi_R$ and a structure $\Gamma$ over $D$ as $R = \{(a_1, \ldots, a_r) \in D^r \mid \varphi^{\Gamma}_{R}(a_1, \ldots, a_r)\}$ we express this as
\[R(x_1, \ldots, x_r) \equiv \varphi^{\Gamma}_R(x_1, \ldots, x_r)\]
and do not make a sharp distinction between structures $\Gamma$ and sets of relations $\Gamma$.

\subsection{Polymorphism Patterns}
\label{sec:patterns}

Fully describing $\pPol(\cla, \clb)$ even for Boolean promise languages $(\cla, \clb)$ appears to be out of reach with current algebraic methods. Fortunately, we can abstract away specific operations and concentrate on {\em identities} satisfied by the operations instead, similar to
the universal algebraic approach for studying classical complexity of (P)CSPs where one typically looks only at {\em linear} identities $f(x_1, \ldots, x_n) \approx g(y_1, \ldots, y_m)$~\cite{barto2017Polymorphisms}. However, the identities relevant in our setting have an even more restricted form where the right-hand side consists of a single variable, and can concisely be defined as follows in the signature-free setting.  

\begin{definition}
    An $n$-ary {\em polymorphism pattern} over a set of variables $V$ for $n \geq 1$ is a subset of $\{(t,x) \mid t = (x_1, \ldots, x_n) \in V^n, x \in V\}$.
\end{definition}

\begin{example} \label{ex:maltsev}
A {\em Mal'tsev} term $\phi$ can be defined by $P = \{((x,x,y), y), ((y,x,x), y)\}$ over $V = \{x,y\}$ and a {\em majority} term can be defined by $P' = \{((x,x,y), x), ((x,y,x), x), ((y,x,x), x)\}$.
\end{example}

\begin{definition}
    Let $D$ be a set and let $P$ be an $n$-ary polymorphism pattern. 
    We define $\interpretation{D}{P}$ to be the set of $n$-ary partial functions $f$ such that:
    \begin{enumerate}
            \item 
            if $((x_1, \ldots, x_n), x) \in P$ where $x \in \{x_1, \ldots, x_n\}$, then for every $g \colon \{x_1, \ldots, x_n\} \to D$ we have  \[f(g(x_1), \ldots, g(x_n)) = g(x),\]
            \item 
        if $((x_1, \ldots, x_n), y) \in P$ where $y \notin \{x_1, \ldots, x_n\}$ then $(g(x_1), \ldots, g(x_n)) \in \dom(f)$ for every $g \colon \{x_1, \ldots, x_n\} \to D$.
    \end{enumerate}
\end{definition}

Functions in $\interpretation{D}{P}$ are sometimes called {\em pattern polymorphisms}.

\begin{example} \label{ex:interpretation}
    If we revisit Example~\ref{ex:maltsev} then for the Mal'tsev term $P = \{((x,x,y), y), (y,x,x), y)\}$ we have $\interpretation{\{0,1\}}{P} = \{m\}$ for 
    \begin{enumerate}
        \item 
        $m(1,1,0) = m(0,1,1) = m(0,0,0) = 0$, and
        \item 
        $m(0,0,1) = m(1,0,0) = m(1,1,1) = 1$, 
    \end{enumerate}
    and which is undefined for $(0,1,0)$ and $(1,0,1)$. Note that this function is {\em idempotent} and {\em self-dual} in the sense that negating the input variables corresponds to negating the output.

    In contrast, for $P' = \{((x,x,y), x), ((x,y,x), x), ((y,x,x), x)\}$ we get $\interpretation{\{0,1\}}{P'} = \{m'\}$ where $m'$ is the totally defined majority function on the Boolean domain. However, for any $D$ with at least three elements  $\interpretation{D}{P'}$ will contain a properly partial function.
\end{example}

If $\interpretation{D}{P} = \emptyset$ then $P$ is {\em inconsistent} over $D$ and if $\interpretation{D}{P}$ contains all $n$-ary functions over $D$ then $P$ is said to be {\em trivial} over $A$. A pattern which is not inconsistent is simply said to be {\em consistent}.
We observe that for a consistent pattern $P$ we have $|\interpretation{D}{P}| = 1$ {\em unless} $P$ contains an identity $((x_1, \ldots, x_n), y)$ for $y \notin \{x_1, \ldots, x_n\}$, in which case this identity induces $|D|$ different partial functions $f_1, \ldots, f_{|A|}$ where $f_i(g(x_1), \ldots, g(x_n)) = i$ for every $g \colon \{x_1, \ldots, x_n\} \to D$. In other words the specific value for this identity is arbitrary, and can be uniquely defined as a {\em multifunction} $f$ such that $f(g(x_1), \ldots, g(x_n)) = D$. However, to avoid overloading the technical machinery more than needed, we stick with the viewpoint that $\interpretation{A}{P}$ is a set of functions.
We remark that patterns $P$ such that $|\interpretation{D}{P}|>1$ are typically too restrictive to be interesting in the general setting, but they are relevant in the multisorted setting discussed below. 

We proceed by relating $\interpretation{D}{P}$ to the property of {\em commuting} with all maps which, intuitively, shows that the domain of any function in $\interpretation{D}{P}$ is highly symmetric.

\begin{definition}
    Let $f$ be an $n$-ary partial function over $D$ and let $g \colon D \to D$ be a unary map. We say that $f$ {\em commutes} with $g$ if $f(g(x_1), \ldots, g(x_n)) = g(f(x_1, \ldots, x_n))$ for all $(x_1, \ldots, x_n) \in \dom(f)$. 
\end{definition}

\begin{example}
    Recall that the partial Boolean Mal'tsev operation $m$ was idempotent and self-dual. One alternative way to define this notion is that $m$ commutes with the two Boolean constant maps and with Boolean negation $g(x) = 1 - x$.
\end{example}

We lift this notion to pairs of functions (possibly defined over different domains) as follows.

\begin{definition}
    Let $f$ and $f'$ be $n$-ary partial functions over $D_1$ and $D_2$, respectively. We say that $f,f'$ {\em commute} with $g \colon D_1 \to D_2$ if 
    \[f'(g(x_1), \ldots, g(x_n)) = g(f(x_1, \ldots, x_n))\] and is defined for all $(x_1, \ldots, x_n) \in \dom(f)$.
\end{definition}

\begin{proposition} \label{prop:commutes}
    Let $P$ be an $n$-ary polymorphism pattern and let $D_1$ and $D_2$ be finite sets. For every $p \in \interpretation{D_1}{P}$ there exists $p' \in \interpretation{D_2}{P}$ such that $p,p'$ commutes with every unary $g \colon D_1 \to D_2$.
\end{proposition}

\begin{proof}
Let $p \in \interpretation{D_1}{P}$. We observe that if $((x_1, \ldots, x_n), x) \in P$ where $x \in \{x_1, \ldots, x_n\}$ then, for every $p' \in \interpretation{D_2}{P}$, we must have that $p'(g(b_1), \ldots, g(b_n)) = g(p(b_1, \ldots, b_n))$ for every $g \colon D_1 \to D_2$ and $b_1, \ldots, b_n$ such that there is $h \colon \{x_1, \ldots, x_n\} \to D_1$, where $b_i = h(x_i)$. 

On the other hand, assume that $((x_1, \ldots, x_n), y) \in P$ for $y \notin \{x_1, \ldots, x_n\}$. Let $b_1, \ldots, b_n \in D_1$ such that there is $h \colon \{x_1, \ldots, x_n\} \to D_1$ where $h(x_i) = b_i$ for every $1 \leq i \leq n$, i.e., $b_1, \ldots, b_n$ matches the pattern, as witnessed by the function $h$. There then exists $p' \in \interpretation{D_2}{P}$ such that $g(p(b_1, \ldots, b_n)) = p'(g(b_1), \ldots, g(b_n))$ for every $g \colon D_1 \to D_2$.
\end{proof}

We will often work in the following setting: consider finite, disjoint sets $D_1, \ldots, D_k$ for some $k \geq 1$, and a relation $R \subseteq (D_1 \cup \ldots \cup D_k)^r$ with an associated signature $s \colon [k] \to [r]$ which for each argument returns the corresponding sort, i.e., if $h(i) = j$ then $t_i \in D_j$ for each $t \in R$.
Note that we for each $D$ and $r \geq 1$ can construct $r$ disjoint copies $D_1, \ldots, D_r$ of $R$ and thus simply speak about $r$-ary multisorted relations over $D$. To ease the notation we thus sometimes simply say that $R$ is an $r$-ary multisorted relation over $D$, rather than over disjoint copies $D_1, \ldots, D_r$.  It is then convenient not only to let the pattern act independently on each domain, but to allow different patterns for each domain type, which is in line with multisorted algebras in the CSP literature~\cite{bulatov2003}.

\begin{definition}
    A \emph{multisorted pattern} is a tuple $(P_1,\ldots,P_r)$ where $P_i$ is an $n$-ary polymorphism pattern for each $i \in [r]$. We say that  $P$ preserves an $m$-ary multisorted promise relation $(S,T)$ over $D_1 \cup \ldots \cup D_r$ with signature $s \colon [m] \to [r]$ if
    \[(p_{s(1)}(x^{(1)}_1, \ldots, x^{(n)}_1), \ldots, p_{s(m)}(x^{(1)}_m, \ldots, x^{(n)}_m)) \in T\] for all $p_{s(1)} \in \interpretation{D_{s(1)}}{P_{s(1)}}, \ldots, p_{s(m)} \in \interpretation{D_{s(m)}}{P_{s(m)}}$ and $x^{(1)}, \ldots, x^{(n)} \in S$ where each $(x^{(1)}_i, \ldots, x^{(n)}_i) \in \dom(p_{s(i)})$.
\end{definition}

This naturally allows us to speak about the set of patterns preserving a given promise relation.

\begin{definition}
    For a (multisorted) promise relation $(S,T)$ over $D$ ($D_1 \cup \ldots \cup D_r$) we let \[\pattern(S,T) = \{P \mid P \text{ is an } n \text{-ary  pattern}, n \geq 1, (S,T) \in \inv(\{P\})\}\] and 
    \[\mpattern(S,T) = \{(P_1, \ldots, P_r) \mid (P_1, \ldots, P_r) \text{ is a multisorted pattern}, (S,T) \in \inv(\{P\})\}.\]
\end{definition}

We are now interested in properties of $\pattern(\cdot)$ (and $\mpattern(\cdot)$) when viewed as algebraic objects themselves. With the aim of showing that these sets form minor like objects we define the following closure operator for (multisorted) patterns.

\begin{definition}
    Let $P = (P_1, \ldots, P_r)$ be a multisorted pattern where each $P_i$ has arity $n$. A {\em variable identification minor} of $P$ is a multisorted pattern $P_{/h} = (P_{1/h}, \ldots, P_{r/h})$ where $P_{i/h} = \{(x_{1}, \ldots, x_{m}), x)) \mid (x_{h(1)}, \ldots, x_{h(n)}), x) \in P_i\} $ for some $h \colon [n] \to [m]$ and $m \geq 1$. 
\end{definition}

Similarly to minors of partial functions we then say that $P'$ is a {\em minor} of a multisorted pattern $P$ if it is a subset of a variable identification minor of $P$.

\begin{definition}
    For a set of multisorted patterns $Q$ we let $[Q]_{\min} = \{P' \mid P' \text{ is a minor of } P \in Q\}$.
\end{definition}

Simultaneously, it then also makes sense to describe the relations invariant under a set of multisorted patterns.

\begin{definition}
    For a set of (multisorted) polymorphism patterns $Q$ and $D$ we write $\inv_D(Q)$ for the set of all (multisorted) $(S,T)$ over $D$ preserved by each $P \in Q$, and $\inv(Q) = \bigcup_{i \geq 1, |D| = i} \inv_D(Q)$ for the set of all (multisorted) relations, over some finite domain, preserved by $Q$.
\end{definition}

We can relate these operators together as follows.

\begin{proposition} \label{prop:minors_preserve}
    For any set $Q$ of (multisorted) polymorphism patterns $\inv(Q) = \inv([Q]_{\min})$.
\end{proposition}

\begin{proof}
    Since $Q \subseteq [Q]_{\min}$ it trivially follows that $\inv([Q]_{\min}) \subseteq \inv(Q)$. For the other direction, pick $(S,T) \in \inv(Q)$ over some $D = D_1 \times \ldots \times D_r$. Let $P' \in [Q]_{\min}$ be a multisorted polymorphism pattern. First, assume that $P'$ is a variable identification minor of some $P \in Q$. Then $P'$ must preserve $(S,T)$, too, since each function in $\interpretation{D}{P'}$ is a variable identification minor of a function in $\interpretation{D}{P}$, and we may simply apply Proposition~\ref{lemma:minors_preserve_relations}. Similarly, if $P'' \subseteq P'$, then $P''$ preserves $(S,T)$ since every partial function in $\interpretation{D}{P''}$ is a subfunction of a partial function in $\interpretation{D}{P}$.
\end{proof}

Similarly, we obtain the following proposition via an abstraction of Lemma~\ref{lemma:is_a_minion} on the pattern level.

\begin{proposition}
$\mpattern(S,T) = [\mpattern(S,T)]_{\min}$ for any multisorted promise relation $(S,T)$.
\end{proposition}

Thus, patterns can be used to define partial promise polymorphisms and form reasonable algebraic objects in their own right. It is natural to then ask what happens with the general algebraic picture if we \emph{only} consider partial promise polymorphisms definable by patterns. A limited case of this setting was considered by Carbonnel~\cite{carbonnel2022Redundancy} who introduced {\em functionally guarded pp-definitions} (fgpp-definitions) as a relational counter part to patterns. Here, the basic building blocks over some $r$-ary relation $R$ over a domain $D$ are  atoms of the form $R(g_1(x_1), \ldots, g_r(x_r))$ where $x_1, \ldots, x_r$ are variables and $g_1, \ldots, g_r \colon D \to D$ are unary functions (maps). Each combination of a map and a variable can then be seen as a generalization of a Boolean literal to arbitrary maps over arbitrary finite domains.  This notion generalizes to maps between arbitrary finite domains and we say that $(S_1, T_1)$ over $D_1$ has a \emph{functionally guarded promise pp-definition} (fgppp-definition) over $(S_2, T_2)$ over $D_2$ if 
there are maps $g_1, \ldots, g_m \colon D_1 \to D_2$ and qfpp-formulas $\varphi^{S_2}$ and $\varphi^{T_2}$ (differing only in the sense that the former uses $S_2$ and the latter $T_2$) such that

    \[S(x_1, \ldots, x_r) \equiv \varphi^{\cla}(\ell_1, \ldots, \ell_k)\]
    and  
        \[T(x_1, \ldots, x_r) \equiv \varphi^{\clb}(\ell_1, \ldots, \ell_k)\] 
where each $\ell_i$ is a combination of a map $g_i \in \{g_1, \ldots, g_m\}$ and a variable from $\{x_1, \ldots, x_r\}$. In the multisorted setting we associate each sort $1 \leq i \leq r$ with a set of maps and only allow maps in position $i$ to use maps of that type. Concretely, we obtain the following links between (multisorted) patterns and fgppp-definitions.

\begin{theoremrestated}{thm:fgppp-pattern}
Let $(S_1, T_1)$ be a promise relation over $D_1$ and $(S_2, T_2)$ a promise relation over $D_2$. Then $\pattern(S_2, T_2) \subseteq \pattern(S_1, T_1)$ if and only if $(S_2, T_2)$ fgppp-defines $(S_1, T_1)$.
\end{theoremrestated}

\begin{corollaryrestated}{cor:pattern_galois}
    Let $(S_1, T_1)$ be a multisorted promise relation over disjoint $D_1, \ldots, D_r$ and $(S_2, T_2)$ a multisorted promise relation over disjoint $E_1, \ldots, E_s$. Then we have $\mpattern(S_2, T_2) \subseteq \mpattern(S_1, T_1)$ if and  only if $(S_2, T_2)$ fgppp-defines $(S_1, T_1)$.
\end{corollaryrestated}

We defer the reader to Section~\ref{sec:fgpp} for the proofs since fgppp-definitions turn out to be a special case of a more powerful concept where non-unary functions are allowed.

\section{Catalan Polymorphisms and Linear NRD}\label{sec:linear-nrd}

In this section, we use polymorphism patterns to study the properties of predicates which have \emph{linear} non-redundancy. In other words, $\NRD(R, n) = O_{R}(n)$. In the literature, the primary known reason for a predicate $P \subseteq D^r$ having linear redundancy is the existence of a \emph{Mal'tsev embedding} \cite{lagerkvist2020Sparsification,bessiere2020Chain}. To recall Section~\ref{subsec:embeddings}, we say that $P$ has a Mal'tsev embedding if there exists a (finite) domain $E \supseteq D$ and a predicate $Q \subseteq E^r$ such that
\begin{itemize}
\item $Q \cap D^r = P$, and
\item There exists a \emph{Mal'tsev term} $\varphi \in \Pol(Q)$. That is, $\varphi$ is ternary and satisfies the identity $\varphi(x,x,y) = \varphi(y,x,x) = y$.
\end{itemize}

For any predicate $R$ admitting a Mal'tsev embedding over a finite domain $E$ one can prove that  $\NRD(R,n) = \Oh_E(n)$ by applying the simple algorithm for Mal'tsev constraints~\cite{bulatov2006Simple} in order to compute a sufficiently compact basis~\cite{lagerkvist2020Sparsification}.

\begin{remark}
As pointed out in \cite{bessiere2020Chain}, a slightly larger class of predicates has linear redundancy: the class of predicates which can be fgpp-defined from predicates with Mal'tsev embeddings. However, they prove that any such predicate has what is called an \emph{infinite} Mal'tsev embedding (i.e., the domain $E$ is infinite).
\end{remark}

A special class of Mal'tsev embeddings is the \emph{group embeddings}. In particular, imagine that $(E, \cdot, {}^{-1})$ is a group. Then, $\varphi(x,y,z) = x\cdot y^{-1}\cdot z$ is a Mal'tsev term. Under this stipulation, the property of $\varphi \in \Pol(Q)$ is equivalent to $Q$ being a coset of the group. We call such embeddings \emph{group} embeddings. If the group happens to be an Abelian group, we call it an \emph{Abelian} embedding. For these special classes of Mal'tsev embeddings we can also make the distinction between embeddings over a finite or infinite domain. In particular, it is known that Abelian embeddings over an infinite domain (e.g. $(\mathbb Z, +)$) always imply an Abelian embedding over a finite domain--see for instance Section~6.1 of \cite{khanna2024Characterizations}.

The goal of this section is to advance the theory of Mal'tsev embeddings by deriving new properties. In particular, we make the following contributions:
\begin{itemize}
\item We introduce a new family of polymorphism identities known as the \emph{Catalan identities}. In particular, we show that every clone with a Mal'tsev term has \emph{Catalan terms} (or \emph{Catalan polymorphisms}) which satisfy these identities.
\item As an immediate corollary, we show the Catalan identities imply that every finite predicate with an infinite Mal'tsev embedding has an embedding into an infinite Coxeter group.
\item Furthermore, we can prove that every Boolean predicate with an infinite Mal'tsev embedding is in fact \emph{balanced} (i.e., has an Abelian embedding). This resolves an open question of \cite{chen2020BestCase}. 
\item Conversely, we construct a predicate over domain size three with a finite group embedding but no Abelian embedding. Invoking the main theorem of \cite{brakensiek2024Redundancy}, we obtain the first explicit example of a predicate with near-linear sparsifiability that does not derive from the framework of \cite{khanna2024Characterizations}. 
\item Furthermore, the Catalan identities streamline existing proofs in the literature that certain predicates lack infinite Mal'tsev embeddings.
\end{itemize}

See Figure~\ref{fig:linear} for an overview of existing results alongside our new results.

\begin{figure}[t]

\begin{center}
  \begin{tikzpicture}[auto, node distance=1.5cm,
    short/.style={shorten >=0.5mm, shorten <=0.5mm},
    normal_edge/.style={draw=black, ultra thick, short},
    theorem_edge/.style={draw=red, ultra thick, short},
    conj_edge/.style={draw=blue, ultra thick, short},
    class/.style={rectangle, rounded corners, draw, ultra thick,
      minimum width = 3 cm, minimum height = 1 cm}
    ]
  
   \node [class] (linNRD) {Linear NRD};
   \node [class, below right = of linNRD, align = center] (infMalt) {Infinite Mal'tsev\\Embedding};
   \node [class, below left = of infMalt, align = center] (finMalt) {Finite Mal'tsev\\Embedding};
   \node [class, below right = of infMalt, align = center] (infGroup) {Infinite Group\\Embedding};
   \node [class, below right = of finMalt, align = center] (finGroup) {Finite Group\\Embedding};
   \node [class, below = of finGroup, align = center] (affine) {Abelian\\Embedding};

   \draw [normal_edge,  <-] (linNRD) edge[bend right] node[left] {Bulatov-Dalmau} (finMalt);
   \draw [conj_edge, ->] (linNRD) edge[bend left] (finMalt);

   \draw [theorem_edge, <-, dashed] (linNRD) edge node[above right] {Theorem~\ref{thm:Mal'tsev-balanced}} (infMalt);
   \draw [normal_edge, <-] (finMalt) edge (finGroup);
   \draw [normal_edge, <-] (infGroup) edge (finGroup);
   \draw [normal_edge, <-] (infMalt) edge (finMalt);
   \draw [normal_edge, ->] (infGroup) edge (infMalt);
   \draw [theorem_edge, <-] (infGroup) edge[bend right=15] node[above right] {Corollary~\ref{cor:inf-malt-group}} (infMalt);
   \draw [normal_edge, <-] (finGroup) edge (affine);
   \draw [theorem_edge, ->] (finGroup) edge[bend left=35] coordinate(x) node[right] {\ Theorem~\ref{thm:non-Abelian}} (affine);
   \draw [red, ultra thick] (x)+(2.5mm,2.5mm) edge (x);
   \draw [red, ultra thick] (x)+(-2.5mm,-2.5mm) edge (x);
   \draw [conj_edge, <-] (finMalt) edge[bend right=20] node[below right] {Conjecture~\ref{conj:inf-to-fin}} (infMalt);
   \draw [theorem_edge, <-, dashed] (affine) edge[bend left=50] node[below left] {Theorem~\ref{thm:Mal'tsev-balanced}} (infMalt);

\end{tikzpicture}
\end{center}

    \caption{
Black edges are known results. Red edges are new results proved in this paper. Blue edges are conjectures. Dashed edges apply to Boolean predicates only. Note that the notion of an ``Abelian embedding'' is the same for finite and infinite groups due to a result of \cite{khanna2024Characterizations}.}\label{fig:linear}
\end{figure}

\subsection{Catalan Identities and Polymorphisms}

Recall that a ternary operator $\varphi : D^3 \to D$ is a Mal'tsev term if it satisfies the identities $\varphi(x,x,y) = \varphi(y,x,x) = y$ for all $x,y \in D$. One can generalize the Mal'tsev identity to arbitrary odd arity as follows. 

\begin{definition}[Catalan identities]\label{def:Catalan-identities}
We say that a sequence of terms $\{\psi_1, \psi_3, \psi_5, \hdots\}$ such that $\psi_m$ has arity $m$ for all odd $m \in \N$ satisfies the \emph{Catalan identities} if the following properties hold.
\begin{itemize}
\item For all $x$, $\psi_1(x) = x$.
\item For all odd $m \in \N$ at least $3$ and all $i \in [m-1]$, we have that
\begin{align}
\text{for all $x_1, \hdots, x_m$, } x_i = x_{i+1} \implies \psi_m(x_1, \hdots, x_{m}) = \psi_{m-2}(x_1, \hdots, x_{i-1}, x_{i+2}, \hdots, x_{m}).\label{eq:catalan}
\end{align}
\end{itemize}
Such terms are called \emph{Catalan terms} or \emph{Catalan polymorphisms}.
\end{definition}

Note that Definition~\ref{def:Catalan-identities} applies equally well to total and partial functions. In the latter case we naturally still require that $\psi_1$ is totally defined and the condition $\psi_m(x_1, \hdots, x_{m}) = \psi_{m-2}(x_1, \hdots, x_{i-1}, x_{i+2}, \hdots, x_{m})$ then means that $\psi_m$ and $\psi_{m-2}$ are both defined and agree or that both equal $\bot$.

\begin{remark}
We coin the name \emph{Catalan identities} because of the close relationship between these identities and the Catalan numbers (see, e.g., Chapter 14 of \cite{van2001course}). In particular, one definition of the $n$th Catalan number $c_n$ is the number of syntactically valid strings consisting of $n$ pairs of left and right parentheses. If one ``flattens'' the Catalan identities by recursively applying (\ref{eq:catalan}) $(m-1)/2$ times, we get exactly $c_{(m+1)/2}$ inequivalent identities. For example, when $m=5$ we get $c_3 = 5$ identities: 
\begin{align*}
\psi_5(x,x,y,y,z) &= z\\
\psi_5(x,x,z,y,y) &= z\\
\psi_5(z,x,x,y,y) &= z\\
\psi_5(x,y,y,x,z) &= z\\
\psi_5(z,x,y,y,x) &= z
\end{align*}
Furthermore, the derivation of Catalan terms from a Mal'tsev term (equations (\ref{eq:f-2k+1}) and (\ref{eq:g-2k+1})) follows a recursive structure similar to that of the recursion defining the Catalan numbers: $c_n = c_0 c_{n-1} + c_1 c_{n-2} + \cdots + c_{n-1} c_0$.
\end{remark}

The main result of this subsection is to show that any clone (over a finite or infinite domain) with a Mal'tsev term also has all the Catalan terms. Recall from Section~\ref{subsec:galois} that we for a set of functions $F$ let $[F]$ be the clone generated by $F$, which in particular contains all possible term operations over $F$. To ease the notation we write $[f]$ rather than $[\{f\}]$ for a singleton set $F = \{f\}$.

\begin{theorem}\label{thm:catalan}
Let $D$ be a (possibily infinite) domain and let $\varphi : D^3 \to D$ be a Mal'tsev term. Then there exist $f_1, f_3, f_5, \hdots \in [\varphi]$ satisfying the Catalan identities.
Conversely, if a sequence of terms $f_1, f_3, \ldots$ satisfies the Catalan identities, then $f_3$ is a Mal'tsev term.
\end{theorem}

Consider the following sequence of operators $\{f_{2i+1}\}_{i \ge 0}$ such that $f_{2i+1}$ has arity $2i+1$.
\begin{align}
f_1(x_1) &:= x_1\nonumber\\
f_3(x_1, x_2, x_3) &:= \varphi(x_1, x_2, x_3)\nonumber\\
f_{2k+1}(x_1, \hdots, x_{2k+1}) &:=  f_{2k-1}(g^{2k+1}_2(x_1, \hdots, x_{2k+1}), \hdots, g^{2k+1}_{2k}(x_1, \hdots, x_{2k+1})),\label{eq:f-2k+1}
\end{align}
where for $k \ge 2$ and $j \in [2k+1]$ we have that
\begin{align}
    g^{2k+1}_j(x_1, \hdots, x_{2k+1}) &:= \begin{cases}
    \varphi(f_{j-1}(x_1, \hdots, x_{j-1}), x_{j}, f_{2k+1-j}(x_{j+1}, \hdots, x_{2k+1})) & \text{$j$ even}\\
    \varphi(f_{j}(x_1, \hdots, x_{j}), x_{j}, f_{2k+2-j}(x_{j}, \hdots, x_{2k+1})) & \text{$j$ odd}.
    \end{cases}\label{eq:g-2k+1}
\end{align}
\begin{remark}
Note that for $j \in \{1, 2k+1\}$, we define $g^{2k+1}_j$, even though these do not appear in (\ref{eq:f-2k+1}). This is to facilitate the inductive argument in the proof of Lemma~\ref{lem:cancel}. Note by (\ref{eq:g-2k+1}) we have that
\begin{align}
g^{2k+1}_1(x_1, \hdots, x_{2k+1}) = \varphi(x_1, x_1, f_{2k+1}(x_1, \hdots, x_{2k+1})) = f_{2k+1}(x_1, \hdots, x_{2k+1})\label{eq:g-2k+1-1}\\
g^{2k+1}_{2k+1}(x_1, \hdots, x_{2k+1}) = \varphi(f_{2k+1}(x_1, \hdots, x_{2k+1}), x_{2k+1}, x_{2k+1}) = f_{2k+1}(x_1, \hdots, x_{2k+1})\label{eq:g-2k+1-2k+1}
\end{align}
\end{remark}
As a concrete example, when $k=2$, we have that (\ref{eq:f-2k+1}) becomes
\begin{align*}
    f_5(x_1, x_2, x_3, x_4, x_5) = \varphi(&\varphi(x_1, x_2, \varphi(x_3, x_4, x_5)),\\
    &\varphi(\varphi(x_1,x_2,x_3), x_3, \varphi(x_3,x_4,x_5)),\\ &\varphi(\varphi(x_1, x_2, x_3), x_4, x_5)).
\end{align*}

Theorem~\ref{thm:catalan} immediately follows by proving the following lemma.
\begin{lemma}[Cancellation Lemma]\label{lem:cancel}
For all $k \ge 1$ and all $i \in [2k]$, we have that
\begin{align}
    x_i = x_{i+1} \implies f_{2k+1}(x_1, \hdots, x_{2k+1}) = f_{2k-1}(x_1, \hdots, x_{i-1}, x_{i+2}, \hdots, x_{2k+1}).\label{eq:cancel}
\end{align}
\end{lemma}
\begin{example}\label{ex:cat5}
To help get intuition for the proof, we include a couple of worked examples.

If $k = 2$ and $i = 1$, we have that
\begin{align*}
    f_5(x_1, x_1, x_3, x_4, x_5) &= \varphi(\varphi(x_1, x_1, \varphi(x_3, x_4, x_5)),\\
    &\qquad\varphi(\varphi(x_1,x_1,x_3), x_3, \varphi(x_3,x_4,x_5)),\\ &\qquad \varphi(\varphi(x_1, x_1, x_3), x_4, x_5)).\\
    &= \varphi(\varphi(x_3,x_4,x_5), \varphi(x_3, x_3, \varphi(x_3, x_4, x_5)), \varphi(x_3, x_4, x_5))\\
    &= \varphi(\varphi(x_3, x_4, x_5), \varphi(x_3, x_4, x_5), \varphi(x_3, x_4, x_5))\\
    &= \varphi(x_3, x_4, x_5) = f_3(x_3, x_4, x_5).
\end{align*}

If $k=2$ and $i=2$, we have that
\begin{align*}
    f_5(x_1, x_2, x_2, x_4, x_5) &= \varphi(\varphi(x_1, x_2, \varphi(x_2, x_4, x_5)),\\
    &\qquad\varphi(\varphi(x_1,x_2,x_2), x_2, \varphi(x_2,x_4,x_5)),\\ &\qquad \varphi(\varphi(x_1, x_2, x_2), x_4, x_5)).\\
    &= \varphi(\varphi(x_1, x_2, \varphi(x_2, x_4, x_5)), \varphi(x_1, x_2, \varphi(x_2, x_4, x_5)), \varphi(x_1, x_4, x_5))\\
    &= \varphi(x_1, x_4, x_5) = f_3(x_1, x_4, x_5),
\end{align*}
where the last line uses that that first two arguments of the outer $\varphi$ are equal.
\end{example}

\begin{proof}[Proof of Lemma~\ref{lem:cancel}]
We prove this by strong induction on $k$. The case $k=1$ asserts that $\varphi(x,x,y) = \varphi(y,x,x) = y$, which is precisely the definition of a Mal'tsev term. Now assume $k \ge 2$ and fix $i \in [2k]$ (i.e., the assumption $x_i = x_{i+1}$). To start, we seek to compute what happens for any $j \in \{2, \hdots, 2k\}$ how $g^{2k+1}_j$ simplifies when $x_i = x_{i+1}$. We break this into cases.

\paragraph{Case 1, $j \le i-1$.} In this case, we seek to show that
\begin{align}
x_i = x_{i+1} \implies g^{2k+1}_j(x_1, \hdots, x_{2k+1}) = g^{2k-1}_j(x_1, \hdots, x_{i-1},x_{i+2}, \hdots, x_{2k+1}).\label{eq:j-le-i-1}
\end{align}
We split into cases based on the parity of $j$. First, if $j$ is even, $x_i = x_{i+1}$ implies that
\begin{align*}
g^{2k+1}_j(x_1, \hdots, x_{2k+1}) &= \varphi(f_{j-1}(x_1, \hdots, x_{j-1}), x_{j}, f_{2k+1-j}(x_{j+1}, \hdots, x_{2k+1}))\\
&= \varphi(f_{j-1}(x_1, \hdots, x_{j-1}), x_{j}, f_{2k-1-j}(x_{j+1}, \hdots, x_{i-1}, x_{i+2}, \hdots, x_{2k+1}))\\
&= g^{2k-1}_j(x_1, \hdots, x_{i-1}, x_{i+2}, \hdots, x_{2k+1}),
\end{align*}
where the second equality uses the strong induction hypothesis on $f_{2k+1-j}$. Second, if $j$ is odd, then $x_i = x_{i+1}$ implies that
\begin{align*}
g^{2k+1}_j(x_1, \hdots, x_{2k+1}) &= \varphi(f_{j}(x_1, \hdots, x_{j}), x_{j}, f_{2k+2-j}(x_{j}, \hdots, x_{2k+1}))\\
&= \varphi(f_{j}(x_1, \hdots, x_{j}), x_{j}, f_{2k-j}(x_{j}, \hdots, x_{i-1}, x_{i+2}, \hdots, x_{2k+1}))\\
&= g^{2k-1}_j(x_1, \hdots, x_{i-1}, x_{i+2}, \hdots, x_{2k+1}),
\end{align*}
where the second equality uses the strong induction hypothesis on $f_{2k+2-j}$ (note that $2k+2-j \le 2k-1$ since $j \ge 2$ is odd). This proves (\ref{eq:j-le-i-1}).

\paragraph{Case 2, $j \ge i+2$.} In this case, we seek to show that
\begin{align}
x_i = x_{i+1} \implies g^{2k+1}_j(x_1, \hdots, x_{2k+1}) = g^{2k-1}_{j-2}(x_1, \hdots, x_{i-1},x_{i+2}, \hdots, x_{2k+1}).\label{eq:j-ge-i+2}
\end{align}
Like in Case 1, we do casework based on the parity of $j$. The relevant calculations for $j$ even are
\begin{align*}
g^{2k+1}_j(x_1, \hdots, x_{2k+1}) &= \varphi(f_{j-1}(x_1, \hdots, x_{j-1}), x_{j}, f_{2k+1-j}(x_{j+1}, \hdots, x_{2k+1}))\\
&= \varphi(f_{j-3}(x_1, \hdots, x_{i-1}, x_{i+2}, \hdots, x_{j-1}), x_{j}, f_{2k-1-(j-2)}(x_{j+1}, \hdots, x_{2k+1}))\\
&= g^{2k-1}_{j-2}(x_1, \hdots, x_{i-1}, x_{i+2}, \hdots, x_{2k+1}),
\end{align*}
and the calculations for $j$ odd are
\begin{align*}
g^{2k+1}_j(x_1, \hdots, x_{2k+1}) &= \varphi(f_{j}(x_1, \hdots, x_{j}), x_{j}, f_{2k+2-j}(x_{j}, \hdots, x_{2k+1}))\\
&= \varphi(f_{j-2}(x_1, \hdots, x_{i-1}, x_{i+2}, \hdots, x_{j}), x_{j}, f_{2k-(j-2)}(x_{j}, \hdots, x_{2k+1}))\\
&= g^{2k-1}_{j-2}(x_1, \hdots, x_{i-1}, x_{i+2}, \hdots, x_{2k+1}).
\end{align*}

\paragraph{Case 3, $j \in \{i, i+1\}$ and $i \not\in \{1,2k\}$.} In this case, instead of simplifying $g^{2k+1}_j$, we seek to show an equality which will help for proving (\ref{eq:cancel}):
\begin{align}
x_i = x_{i+1} \wedge i \not\in \{1,2k\} \implies g^{2k+1}_i(x_1, \hdots, x_{2k+1}) = g^{2k+1}_{i+1}(x_1, \hdots, x_{2k+1}).\label{eq:j-i-i+1}
\end{align}
If $i$ is even, then $x_i = x_{i+1}$ implies that
\begin{align*}
g^{2k+1}_i(x_1, \hdots, x_{2k+1}) &= \varphi(f_{i-1}(x_1, \hdots, x_{i-1}), x_i, f_{2k+1-i}(x_{i+1}, \hdots, x_{2k+1}))\\
&= \varphi(f_{i+1}(x_1, \hdots, x_{i-1}, x_i, x_{i+1}), x_{i+1}, f_{2k+1-i}(x_{i+1}, \hdots, x_{2k+1}))\\
&= g^{2k+1}_{i+1}(x_1, \hdots, x_{2k+1}),
\end{align*}
where we use $i \neq 2k$ so that $i+1 \le 2k-1$, allowing for us to invoke the induction hypothesis.
If $i$ is odd, then $x_i = x_{i+1}$ implies that
\begin{align*}
g^{2k+1}_i(x_1, \hdots, x_{2k+1}) &= \varphi(f_{i}(x_1, \hdots, x_{i}), x_i, f_{2k+2-i}(x_{i}, x_{i+1}, x_{i+2}, \hdots, x_{2k+1}))\\
&= \varphi(f_{i}(x_1, \hdots, x_{i}), x_{i+1}, f_{2k-i}(x_{i+2}, \hdots, x_{2k+1}))\\
&= g^{2k+1}_{i+1}(x_1, \hdots, x_{2k+1}),
\end{align*}
where we use $i \neq 1$ so that $2k+2-i \le 2k-1$, allowing for us to invoke the induction hypothesis.

\paragraph{Finishing the proof.} We now prove (\ref{eq:cancel}). We first consider the boundary cases that $i \in \{1,2k\}$. If $i = 1$ (that is $x_1 = x_2$), note that
\begin{align}
    g^{2k+1}_2(x_1, \hdots, x_{2k+1}) = \varphi(x_1, x_2, f_{2k-1}(x_3, \hdots, x_{2k+1})) = f_{2k-1}(x_3, \hdots, x_{2k+1}).\label{eq:x1-x2}
\end{align}
By the definition of $f_{2k+1}$, we have that
\begin{align*}
x_1 = x_2 \implies&f_{2k+1}(x_1, \hdots, x_{2k+1})\\
&=  f_{2k-1}(g^{2k+1}_2(x_1, \hdots, x_{2k+1}), g^{2k+1}_3(x_1, \hdots, x_{2k+1}),\\&\quad\quad\quad g^{2k+1}_4(x_1, \hdots, x_{2k+1}),\hdots, g^{2k+1}_{2k}(x_1, \hdots, x_{2k+1}))\\
&= f_{2k-1}(f_{2k-1}(x_3, \hdots, x_{2k+1}), g^{2k-1}_{1}(x_3, \hdots, x_{2k+1}),\\
&\quad\quad\quad g^{2k-2}_{2}(x_3, \hdots, x_{2k+1}), \hdots, g^{2k-1}_{2k-2}(x_3, \hdots, x_{2k+1})) \text{ by (\ref{eq:x1-x2}) and (\ref{eq:j-ge-i+2})}\\
&= f_{2k-1}(f_{2k-1}(x_3, \hdots, x_{2k+1}), f_{2k-1}(x_3, \hdots, x_{2k+1}),\\
&\quad\quad\quad g^{2k-1}_2(x_3, \hdots, x_{2k+1}), \hdots, g^{2k-1}_{2k-2}(x_3, \hdots, x_{2k+1})) \text{ by (\ref{eq:g-2k+1-1})}\\
&= f_{2k-3}(g^{2k-1}_2(x_3, \hdots, x_{2k+1}), \hdots, g^{2k-1}_{2k-2}(x_3, \hdots, x_{2k+1})) \text{ by induction hypothesis}\\
&= f_{2k-1}(x_3, \hdots, x_{2k+1}),
\end{align*}
as desired. 
If $i = 2k$ (that is $x_{2k} = x_{2k+1}$), note that
\begin{align}
    g^{2k+1}_{2k}(x_1, \hdots, x_{2k+1}) = \varphi(f_{2k-1}(x_1, \hdots, x_{2k-1}), x_{2k}, x_{2k+1}) = f_{2k-1}(x_1, \hdots, x_{2k-1}).\label{eq:x2k-x2k+1}
\end{align}
Likewise, we have that
\begin{align*}
x_{2k} = x_{2k+1} \implies &f_{2k+1}(x_1, \hdots, x_{2k+1})\\
&=  f_{2k-1}(g^{2k+1}_2(x_1, \hdots, x_{2k+1}), \hdots, g^{2k+1}_{2k-2}(x_1, \hdots, x_{2k+1}),\\
&\quad\quad\quad g^{2k+1}_{2k-1}(x_1, \hdots, x_{2k+1}), g^{2k+1}_{2k}(x_1, \hdots, x_{2k+1}))\\
&= f_{2k-1}(g^{2k-1}_2(x_1, \hdots, x_{2k-1}), \hdots, g^{2k-1}_{2k-2}(x_1, \hdots, x_{2k-1}),\\
&\quad\quad\quad g^{2k-1}_{2k-1}(x_1, \hdots, x_{2k-1}), f_{2k-1}(x_1, \hdots, x_{2k-1})) \text{ by (\ref{eq:j-le-i-1}) and (\ref{eq:x2k-x2k+1})}\\
&= f_{2k-1}(g^{2k-1}_2(x_1, \hdots, x_{2k-1}), \hdots, g^{2k-1}_{2k-2}(x_1, \hdots, x_{2k-1}),\\ &\quad\quad\quad f_{2k-1}(x_1, \hdots, x_{2k-1}), f_{2k-1}(x_1, \hdots, x_{2k-1})) \text{ by (\ref{eq:g-2k+1-2k+1})}\\
&= f_{2k-3}(g^{2k-1}_2(x_1, \hdots, x_{2k-1}), \hdots, g^{2k-1}_{2k-2}(x_1, \hdots, x_{2k-1})) \text{ by induction hypothesis}\\
&= f_{2k-1}(x_1, \hdots, x_{2k-1}),
\end{align*}
as desired. Now assume that $i \in \{2, \hdots, 2k-1\}$. By (\ref{eq:j-le-i-1}), (\ref{eq:j-ge-i+2}), and (\ref{eq:j-i-i+1}), we have that
\begin{align*}
&x_{i} = x_{i+1} \implies f_{2k+1}(x_1, \hdots, x_{2k+1})\\
&=  f_{2k-1}(g^{2k+1}_2(x_1, \hdots, x_{2k+1}), \hdots, g^{2k+1}_{i-1}(x_1, \hdots, x_{2k+1}),\\
&\quad\quad\quad g^{2k+1}_i(x_1, \hdots, x_{2k+1}), g^{2k+1}_{i+1}(x_1, \hdots, x_{2k+1}),\\
&\quad\quad\quad g^{2k+1}_{i+2}(x_1, \hdots, x_{2k+1}), \hdots, g^{2k+1}_{2k}(x_1, \hdots, x_{2k+1}))\\
&= f_{2k-1}(g^{2k-1}_2(x_1, \hdots, x_{i-1}, x_{i+2}, \hdots x_{2k+1}), \hdots, g^{2k-1}_{i-1}(x_1, \hdots, x_{i-1}, x_{i+2}, \hdots, x_{2k+1}),\\
&\quad\quad\quad g^{2k+1}_{i}(x_1, \hdots, x_{2k+1}), g^{2k+1}_{i+1}(x_1, \hdots, x_{2k+1}),\\
&\quad\quad\quad g^{2k-1}_{i}(x_1, \hdots, x_{i-1}, x_{i+2}, \hdots x_{2k+1}), \hdots, g^{2k-1}_{2k-2}(x_1, \hdots, x_{i-1}, x_{i+2}, \hdots x_{2k+1})) \text{ by (\ref{eq:j-le-i-1}) and (\ref{eq:j-ge-i+2})}.\\
&= f_{2k-3}(g^{2k-1}_2(x_1, \hdots, x_{i-1}, x_{i+2}, \hdots x_{2k+1}), \hdots, g^{2k-1}_{i-1}(x_1, \hdots, x_{i-1}, x_{i+2}, \hdots, x_{2k+1}),\\
&\quad\quad\quad g^{2k-1}_i(x_1, \hdots, x_{i-1}, x_{i+2}, \hdots, x_{2k+1}), \hdots, g^{2k-1}_{2k-2}(x_1, \hdots, x_{i-1}, x_{i+2}, \hdots x_{2k+1})) \text{ by (\ref{eq:j-i-i+1}) and IH}\\
&= f_{2k-1}(x_1, \hdots, x_{i-1}, x_{i+2}, \hdots, x_{2k+1}),
\end{align*}
as desired. This completes the proof of (\ref{eq:cancel}) assuming the induction hypothesis.
\end{proof}

\subsection{Polymorphism Pattern Perspective: Boolean Mal'tsev Embeddings are Balanced}

Theorem~\ref{thm:catalan} has a number of implications on the structure of Mal'tsev embeddings. To develop this theory properly, we first consider the Catalan identities from the perspective of partial polymorphism patterns. In particular, we shall view the Catalan partial polymorphisms as special cases of \emph{group polymorphism patterns}. 

\begin{definition}[Group Polymorphism Patterns]
Let $V$ be a (possibly infinite) set of variables, and let $G$ be a (possibly infinite) group with a corresponding map $\eta : V \to G$. For every odd $m \in \N$, we define the polymorphism pattern $\Grp_m^{V,G,\eta}$ to contain all patterns $((x_1, \hdots, x_m), y)$ for any $x_1, \hdots, x_m, y \in V$ such that
\[
    \eta(x_1) \cdot \eta(x_2)^{-1} \cdot \eta(x_3) \cdot \hdots \cdot \eta(x_{m-1})^{-1} \cdot \eta(x_m) = \eta(y).
\]
In other words $\Grp_m^{V,G,\eta}$ captures the arity-$m$ patterns which arise from a group embedding into $G$ via $\eta$. 
\end{definition}

Group polymorphism patterns can capture many properties of predicates found in the literature. For example, Chen, Jansen, and Pieterse~\cite{chen2020BestCase} consider the notion of a \emph{balanced} Boolean predicate as a potential characterization of Boolean predicates with linear non-redundancy.
We say that a Boolean predicate $P \subseteq \{0,1\}^r$ is \emph{balanced} if for any odd $m \in \N$ and tuples $t^1, \hdots, t^m \in P$, we have the following property: if for each $i \in [r]$ we have $t^1_i - t^2_i + t^3_i - \cdots + t^m_i \in \{0,1\}$, then $t^1 - t^2 + t^3 - \cdots + t^m \in P$. In other words, the operation defined by 
\[
f(x_1,x_2,\dots,x_m) = \begin{cases} x_1 - x_2 + \cdots +x_m & \text{if in $\{0,1\}$}\\
\bot &\text{otherwise}
\end{cases}
\] is a partial polymorphism of $P$. We now demonstrate that this property of being balanced can be captured by the following group polymorphism pattern
\[ \Bal_m := \Grp_m^{\{x,y\}, \mathbb Z, \eta} \] where $\eta(x) = 0$ and $\eta(y) = 1$.

\begin{proposition}\label{prop:balance}
Any $P \subseteq \{0,1\}^r$ is balanced if and only if $I_{\{0,1\}}(\Bal_m) \subseteq \pPol(P)$ for all odd $m \in \N$.
\end{proposition}
\begin{proof}
The condition $I_{\{0,1\}}(\Bal_m) \subseteq \pPol(P)$ is equivalent to the following condition: for all $t^1, \hdots, t^m \in P$ satisfying
\[
I_{\{0,1\}}(\Bal_m)(t^1_i, t^2_i, \hdots, t^m_i) \in \{0,1\}
\]
for each $i \in [r]$,
we have $I_{\{0,1\}}(\Bal_m)(t^1, t^2 \hdots, t^m) \in P$. However, by definition of $\Bal_m$, we have that $I_{\{0,1\}}(\Bal_m)(t^1_i, t^2_i, \hdots, t^m_i) = t^1_i - t^2_i + t^3_i - \cdots + t^m_i$ whenever the latter is contained in $\{0,1\}$. Therefore, this condition is equivalent to $P$ being balanced.
\end{proof}

\subsubsection{Alternating Sums and Abelian Embeddings}

However, since $\Bal_m$ only considers patterns with two symbols, there are multiple natural generalizations into more general patterns. One such generalization, roughly equivalent to the approach taken by \cite{khanna2024Characterizations}, is to let $G$ be the free Abelian group generated by $V$, which we will denote as $\mathbb Z^V$ with coordinate-wise addition as the group operation. In other words, the elements of $G$ consist of all (finite) formal sums of the form $z_1 \cdot x_1 + \hdots z_\ell \cdot x_\ell$, where $x_1, \hdots, x_\ell \in V$ and $z_1, \hdots, z_\ell \in \mathbb Z$, with addition and subtraction defined in the obvious manner. We let $\eta$ be the ``identity'' map from $V$ to $\Z^V$ and call the resulting group polymorphism pattern to be an \emph{alternating sum pattern}.\footnote{We chose this name to be analogous to the ``alternating threshold'' polymorphisms considered in the theory of promise CSPs~\cite{brakensiek2021Promise}.} That is, for all odd $m \in \N$, we define $\AS^V_m := \Grp_m^{V,\Z^V,\eta}$.
This pattern is most general when $V$ is countably infinite. We denote this family of polymorphism patterns by $\{\AS_m\}$. We now show this family of polymorphism patterns exactly captures Abelian embeddings.

\begin{proposition}\label{prop:abelian-AS}
For any $P \subseteq D^r$, the following are equivalent.
\begin{itemize}
\item $P$ has a finite Abelian embedding.
\item $P$ has an embedding into the free Abelian group $\Z^D$ via $\eta(d) := e_d$.
\item $I_{D}(\AS_m) \subseteq \pPol(P)$ for all odd $m \in \N$.
\end{itemize}
\end{proposition}
\begin{proof}
First, we show that $P$ has an Abelian embedding if and only if it has an embedding into $\Z^D$ in the canonical manner (i.e., each element of $D$ maps to its corresponding generator). Assume that $P$ has an embedding into some group $H$ via a map $\eta : D \to H$. By the fundamental property of the free group, there exists a homomorphism $\sigma : \Z^D \to H$ such that $\eta = \sigma \circ \iota$, where $\iota : D \to \Z^D$ is the canonical embedding. As such, $\sigma$ implies that $\eta$ being an embedding of $P$ into $H$ implies that $\iota$ is an embedding of $P$ into $\Z^D$. %
Note that $P$ must also embed into a finite Abelian group using the lattice-based techniques developed in Section~6.1 of \cite{khanna2024Characterizations}. We omit further details on how this finite group is constructed.

We now proceed to prove the equivalence. Let $Q \subseteq (\Z^D)^r$ be the subgroup of $(\Z^D)^r$ generated by $P$. We claim that $Q \cap D^r = P$ if and only if $I_{D}(\AS_m) \subseteq \pPol(P)$.

We first prove the ``if'' direction. Consider any $t \in Q \cap D^r$. Since $Q$ is generated by $P$, we have there exists nonzero $z_1, \hdots, z_\ell \in \Z$ and corresponding $t^1, \hdots, t^\ell \in P$ such that $t = z_1 \cdot t^1 + \cdots + z_\ell \cdot t^\ell.$ Note that the sum of the coordinates of each $t_i$ (and $t$) is equal to $1$. Thus, we must have that $z_1 + \cdots + z_\ell = 1$. Let $m = |z_1| + \cdots + |z_\ell|$. We can see that the positive $z_i$'s sum to precisely $(m+1)/2$ and negative $z_i$'s sum to precisely $-(m-1)/2$. Thus, there exists a tuple $s \in [\ell]^m$ such that each $j \in [\ell]$ appears in exactly $|z_j|$ positions of $s$. Furthermore, $j$ only appears in odd positions if $z_j$ is positive and only in even positions if $z_j$ is negative.

For each $i \in [r]$, we consider $s^i \in D^m$ where $s^i_j = t^{s_j}_i$. Then, since 
\[
t_i = z_1 \cdot t_i^1 + \cdots + z_\ell \cdot t_i^\ell = \sum_{j=1}^m (-1)^{j-1} t^{s_j}_i = \sum_{j=1}^m (-1)^{j-1} s^i_j.
\]
Therefore, $(s^i_j, t_i) \in \AS_m$ for all $i \in [r].$ In particular, this means that $I_D(\AS_m) \subseteq \pPol(P)$ implies $t \in P$. Therefore, $Q \cap D^r = P$.

For the ``only if'' direction, assume that $I_{D}(\AS_m) \not\subseteq \pPol(P)$. In particular, this means there exists $p \in I_D(\AS_m)$ and $t^1, \hdots, t^m \in P$ but $t \not\in P$ such that $p(t^1_i, t^2_i, \cdots, t^m_i) = t_i$ for all $i \in [r]$. Since $p \in I_D(\AS_m)$, we must have that $t^1_i - t^2_i + \cdots - t^m_i = t_i$ for all $i \in [r]$, where the sum is taken in $\Z^D$. Therefore, $t \in Q \cap D^r$, so $Q \cap D^r \neq P$. By taking the contrapositive, we have that $Q \cap D^r = P$ implies $I_{D}(\AS_m) \subseteq \pPol(P).$
\end{proof}

\subsubsection{Catalan Patterns and Group Embedings}

However, we can also generalize the property of being balanced in a non-Abelian manner. In particular, given a set $V$ of variables, we let $\Cox(V)$ denote the group generated by the elements of $V$ with the condition that for each $v \in V$, we have that $v^2$ is the identity element. We note that $\Cox(V)$ is a special example of \emph{Coxeter} groups~\cite{Humphreys_1990}.
In particular, it is isomorphic to the group generated by $|V|$ reflections in general position in $\R^2$. The group can also be thought of as the \emph{free product} of $V$ copies of $\Z/2\Z$.

The simplest way to think about $\Cox(V)$ is that it consists of \emph{words} that lack immediate repetition. For example, if $V = \{x,y,z\}$, a valid computation is
\[
    xyz \cdot zxy \cdot yxyxy = y.
\]
Furthermore, inverses just involve reversing the symbols in a word, so $zyx$ is the inverse of $xyz$.

If we let $\eta$ denote the canonical embedding from $V$ into $\Cox(V)$, then for all $m \in \N$, we can define 
$$\Cat^V_m := \Grp_m^{V, \Cox(V), \eta} \ . $$ If $V$ is countably infinite, we let $\Cat_m := \Cat^V_m.$ As the name suggests, these \emph{Catalan polymorphism patterns} precisely capture the \emph{Catalan identities}.

\begin{proposition}\label{prop:catalan-pattern}
For any $P \subseteq D^r$, there exists $\psi_1, \psi_3, \psi_5, \hdots \in \pPol(P)$ satisfying the Catalan identities if and only if $I_D(\Cat_m) \subseteq \pPol(P)$ for all odd $m$. 
\end{proposition}
\begin{proof}
First assume $I_D(\Cat_m) \subseteq \pPol(P)$ for all odd $m \in \N$. Thus, for each odd $m \in \N$, there exists a polymorphism $\psi_m \in I_D(\Cat_m)$ such that $\psi_m(x_1, \hdots, x_m) = x_1 \cdots x_m$ if the product in the Coxeter group evaluates to a single generator of $\Cox(D)$. Otherwise, $\psi_m(x_1, \hdots, x_m) = \bot$. We claim that $\psi_1, \psi_3, \hdots$ satisfy the Catalan identities. First, clearly $\psi_1(x) = x$. Second, for any odd $m \in \N$ and $i \in [m-1]$, we have that $x_{i} = x_{i+1}$ implies that $x_{i}x_{i+1}$ in the Coxeter group, so $x_1 \cdots x_m = x_1 \cdots x_{i-1}x_{i+2} \cdots x_m$. In particular, $x_i = x_{i+1}$ implies that $\psi_m(x_1, \hdots, x_m) = \psi_{m-2}(x_1, \hdots, x_{i-2}, x_{i+1}, \hdots, x_m)$ as either both equal the same monomial or both equal $\bot$.

Conversely, assume that there exists $\psi_1, \psi_3, \psi_5, \hdots \in \pPol(P)$ satisfying the Catalan identities. For each odd $m \in \N$ and each identity $((x_1, \hdots, x_m),y) \in \Cat_m$ for some $x_1, \hdots, x_m, y \in V$, we seek to prove that every interpretation of $(t,x)$ over $D$ is satisfied by $\psi_m$. We prove this by induction on $m$. The case $m=1$ is trivial as $\psi_1(d) = d$ for all $d \in D$.

In order for $((x_1, \hdots, x_m),y) \in \Cat_m$, we must have that $x_1 \cdots x_m = y$ in $\Cox(V)$. For this to be the case, we must have that some cancellation occurs. That is, there exists $i \in [m-1]$ such that $x_{i-1} = x_i$. Furthermore, we must have that $((x_1, \hdots, x_{i-1},x_{i+2}, \hdots, x_m), y) \in \Cat_{m-2}$. Thus, for every interpretation $\iota : V \to D$, we have by the Catalan identities that
\[
    \psi_m(\iota(x_1), \hdots, \iota(x_m)) = \psi_{m-2}(\iota(x_1), \hdots, \iota(x_{i-1}), \iota(x_{i+2}), \hdots, \iota(x_m)) = y,
\]
where for the last equality we use the induction hypothesis. Thus, $I_D(\Cat_m) \subseteq \pPol(P)$ for all odd $m$, as desired.
\end{proof}

\subsubsection{Corollary: Infinite Mal'tsev to Infinite Group}

As an immediately corollary of Theorem~\ref{thm:catalan} and Proposition~\ref{prop:catalan-pattern}, we have that infinite Mal'tsev embeddings are equivalent to infinite group embeddings. This sharpens the classification by Lagerkvist and Wahlstr\"om~\cite{lagerkvist2020Sparsification} who relate the existence of a Mal'tsev embedding to a Mal'tsev embedding over a certain free algebra over a countable infinite domain whose term operations coincide with the so-called {\em universal partial Mal'tsev} operations. As these ``UPM'' operations are rather opaque, our characterization in terms of an embedding into $\Cox(D)$ and the resulting partial Catalan operations appears to be a much more useful description.

\begin{corollary}\label{cor:inf-malt-group}
A finite predicate $P \subseteq D^r$ has an infinite Mal'tsev embedding if and only if $P$ embeds into $\Cox(D)$.
\end{corollary}
\begin{proof}
Since a group embedding is a special case of Mal'tsev embedding, the ``if'' direction is trivial.

For the ``only if'' direction, by Theorem~\ref{thm:catalan}, $P$ has partial polymorphisms satisfying the Catalan identities. By Proposition~\ref{prop:catalan-pattern}, this implies that $I_D(\Grp_m^{V,\Cox(V),\eta}) = I_D(\Grp_m^{D,\Cox(D),\eta}) \subseteq \pPol(P)$ for all odd $m \in \N$. Let $Q$ be the subgroup of $\Cox(D)^r$ generated by $P$. Any $t \in Q \cap D^r$ must be expressed as a product $t^1 \cdots t^m$ of tuples in $P$. By a simple parity argument, $m$ must be odd. Thus, there exists $p \in I_D(\Grp_m^{D,\Cox(D),\eta})$ such that $p(t^1_i, \hdots, t^m_i) = t_i$ for all $i \in [r]$. Thus, $t \in P$, so $P$ has an embedding into $\Cox(D)$.
\end{proof}

\begin{remark}\label{rem:coset}
We can make \Cref{cor:inf-malt-group} slightly more precise. Let $\widetilde{\Cox}(D)$ be the index-$2$ coset of $\Cox(D)$ consisting of words of odd length. Since each $d \in D$ embeds into an element of $\Cox(D)$ of length $1$, we in fact constructed an embedding into $\widetilde{\Cox}(D)$. The upcoming Theorem~\ref{thm:Mal'tsev-balanced} is equivalent to the fact that $\widetilde{\Cox}(\{0,1\}) \cong (\Z, +)$.
\end{remark}

\subsubsection{Corollary: Boolean Mal'tsev Embeddings are Balanced}

As another corollary of our Catalan machinery, we show that for Boolean predicate $P$, having a Mal'tsev embedding (even over an infinite domain!) is sufficient to prove that $P$ is balanced.

\begin{theorem}\label{thm:Mal'tsev-balanced}
Let $P \subseteq \{0,1\}^k$ be a predicate which embeds into a predicate $A \subseteq D^k$ with a Mal'tsev term $\varphi \in \Pol(A)$. Then, $P$ is balanced.
\end{theorem}

This resolves an open question of \cite{chen2020BestCase} who only proved Theorem~\ref{thm:Mal'tsev-balanced} in the special case that $\varphi(x,y,z) = x \cdot y^{-1} \cdot z$ where $D$ embeds into a (possibly infinite) group. In fact, immediately applying Corollary~\ref{cor:inf-malt-group} to their group embedding result implies Theorem~\ref{thm:Mal'tsev-balanced}.

However, to make the argument a bit more transparent, we give a more structural proof of Theorem~\ref{thm:Mal'tsev-balanced}. In particular, we show the identities implying that a predicate is balanced are identical to the identities implying a Boolean predicate has partial polymorphisms satisfying the Catalan identities.
\begin{lemma}\label{lem:bal-cat}
For all odd $m \in \N$, we have that $I_{\{0,1\}}(\Bal_m) = I_{\{0,1\}}(\Cat_m)$.
\end{lemma}
Theorem~\ref{thm:Mal'tsev-balanced} then immediately follows by combining Lemma~\ref{lem:bal-cat} with Proposition~\ref{prop:balance}, Proposition~\ref{prop:catalan-pattern}, and Theorem~\ref{thm:catalan}. In fact, the technical core of Lemma~\ref{lem:bal-cat} was already observed by \cite{chen2020BestCase}.
\begin{observation}[Implicit in Claim~7.9 of \cite{chen2020BestCase}.]\label{obs:Boolean-cancel}
Let $m \ge 3$ be odd and let $x_1, \hdots, x_m \in \{0,1\}$ be such that
\begin{align}
    x_1 - x_2 + x_3 - \cdots + x_m \in \{0,1\},\label{eq:alt-bool}
\end{align}
then there exists $i \in [m-1]$ such that $x_i = x_{i+1}$.
\end{observation}
The proof of Observation~\ref{obs:Boolean-cancel} is rather elementary, as the only Boolean strings which lack a neighboring pair which are equal are the alternating strings $101\cdots 101$ and $010\cdots 010$, but it is straightforward to verify these alternating strings can never satisfy (\ref{eq:alt-bool}) unless $m=1$.

\begin{proof}[Proof of Lemma~\ref{lem:bal-cat}]
We prove this statement by induction on $m$. The base case is trivial as both $\Bal_1$ and $\Cat_1$ only have the unary pattern $((x),x)$. Fix $m \in \N$ at least $3$ and consider an arbitrary pattern $(t,y) := ((x_1, \hdots, x_m), y)$ using only two symbols $0$ and $1$. We have that $(t,y) \in \Bal_m$ if and only if $x_1 - x_2 + \cdots + x_m = y$.  Likewise, $(t,y) \in \Cat_m$ if $x_1x_2 \cdots x_m = y$ in $\Cox(\{0,1\})$ (i.e., there is a series of cancellations resulting in $y$). It is clear the latter always implies the former, so $I_{\{0,1\}}(\Cat_m) \subseteq I_{\{0,1\}}(\Bal_m)$.

For the other direction, if $(t,y) \in \Bal_m$, by Observation~\ref{obs:Boolean-cancel}, we have that there exists $i \in [m-1]$ such that $x_i = x_{i+1}$. Thus, $x_1 - x_2 + \cdots \pm x_{i-1} \mp x_{i+2} \pm \hdots + x_m = y$, so $((x_1, \hdots, x_{i-1}, x_{i+2}, \hdots, x_m), y) \in \Bal_{m-2}$. By the induction hypothesis we have that $((x_1, \hdots, x_{i-1},$ $x_{i+2}, \hdots, x_m), y) \in \Cat_{m-2}$. Thus, since $x_i$ and $x_{i+1}$ cancel, we have that $((x_1, \hdots, x_m), y) \in \Cat_{m},$ as desired.
\end{proof}

\begin{remark}\label{rem:BCK}
As an attempt toward proving Theorem~\ref{thm:Mal'tsev-balanced}, \cite{chen2020BestCase} considered the following function similar to the arity-$5$ Catalan polymorphism: \[f(x_1, x_2, x_3, x_4, x_5) = \varphi(x_1, \varphi(x_2, x_3, \varphi(x_1, x_2, x_3)), \varphi(x_5, x_4, \varphi(x_3, x_2, x_1))).\]
One can verify that $f$ satisfies the Boolean interpretation of $\Cat_5$, but it fails to interpret $\Cat_5$ over larger domains. In particular assuming $x_3 = x_4$ does not simplify $f$. This relative lack of structure in $f$ perhaps made it difficult to find a more general polymorphisms.

Of note, this function $f$ was independently found (except for the final $x_1$ and $x_3$ being swapped) by \cite{bessiere2020Chain} to prove that the non-Boolean predicate $P := \{111,222,012,120,201\}$ lacks a Mal'tsev embedding. We explore such arguments in more detail in Section~\ref{subsec:excl}.
\end{remark}

\subsection{A Non-Abelian Group Embedding}

Given that Theorem~\ref{thm:Mal'tsev-balanced} (or more precisely Lemma~\ref{lem:bal-cat}) shows that the Catalan identities imply Abelian structure for Boolean predicates, one may naturally ask if such an argument extends to non-Boolean domains. If true, such a result could lead to an elegant boundary between predicates with linear non-redundancy and those with superlinear non-redundancy. Sadly, such a world does not exist.

\begin{theorem}\label{thm:non-Abelian}
There exists an explicit relation $R \subseteq D^6$ with $|D| = 3$ for which $R$ has a group embedding but lacks an Abelian embedding.
\end{theorem}

We motivate the construction from the perspective of polymorphism patterns. If one were to generalize Lemma~\ref{lem:bal-cat} to domain size three, one would hope that for all odd $m \in \N$,
\[
    I_{\{0,1,2\}}(\AS_{m}) = I_{\{0,1,2\}}(\Cat_m). 
\]
It turns out this is not even true for $m = 5$. Up to isomorphism, the identities defining $\Cat_5$ are
\begin{align*}
    \Cat_5 = \{\quad&((x,x,y,y,z),z),\\
    &((x,x,z,y,y),z),\\
    &((z,x,x,y,y),z),\\
    &((x,y,y,x,z),z),\\
    &((z,x,y,y,x),z)\quad \}
\end{align*}
However, $\AS_5$ has these five identities along with $((x,y,z,x,y), z)$---notice the two $x$'s and the two $y$'s have opposite signs in an alternating sum. We can immediately lift this mismatch in the identities into an explicit predicate. Consider the domain $D = \{x,y,z\}$ and the relation $\PAULI \subset D^6$ defined by
\[
    \PAULI := \begin{pmatrix}
    x & x & y & y & z\\
    x & x & z & y & y\\
    z & x & x & y & y\\
    x & y & y & x & z\\
    z & x & y & y & x\\
    x & y & z & x & y
    \end{pmatrix},
\]
where the tuples in the relation are the columns of the matrix. Because of the aforementioned list of identities, one can verify that there exists $p \in I_D(\AS_5)$ such that applying $p$ to rows of above matrix gives $(z,z,z,z,z,z) \not\in \PAULI$.
More transparently, any Abelian group containing $R$ must also contain $(z,z,z,z,z,z)$. Thus, $R$ cannot be embedded into any Abelian group. However, the same argument fails with $I_D(\Cat_5)$ as $p(x,y,z,x,y)$ cannot be simplified using Catalan identities. In fact, as we shall see $\PAULI$ embeds into a non-Abelian group, specifically the Pauli group.

\paragraph{The Pauli Group.} 
Consider the matrices
\begin{align*}
I &= \begin{pmatrix} 1 & 0\\0 & 1\end{pmatrix}&
X &= \begin{pmatrix} 0 & 1\\1 & 0\end{pmatrix}&
Y &= \begin{pmatrix} 0 & -\sqrt{-1}\\\sqrt{-1} & 0\end{pmatrix}&
Z &= \begin{pmatrix} 1 & 0\\0 & -1\end{pmatrix}.
\end{align*}
The matrices $X,Y,Z$ generate a group known as the \emph{Pauli group} $P_1$ with $16$ elements. In particular, every element of $P_1$ is of the form $(-1)^a X^b Y^c Z^d$ for $a,b,c,d \in \{0,1\}$. We also have the following identities
\begin{align*}
X^2 = Y^2 = Z^2 &= I,&
XY &= -YX,&
XZ &= -ZX,&
YZ &= -ZY.
\end{align*}

Now, we claim that the map which send $x,y,z$ to $X,Y,Z$ is a group embedding of $\PAULI$. In particular define $g_1, \hdots, g_5 \in P_1^6$ as the embedding of the tuples defining $\PAULI$.
\begin{align*}
g_1 &= (X,X,Z,X,Z,X)\\
g_2 &= (X,X,X,Y,X,Y)\\
g_3 &= (Y,Z,X,Y,Y,Z)\\
g_4 &= (Y,Y,Y,X,Y,X)\\
g_5 &= (Z,Y,Y,Z,X,Y).
\end{align*}
Let $C$ be the minimal coset of $P_1^6$ generated by $g_1, \hdots, g_5$. That is,
\begin{align*}
    C &= \{g_{i_1} g_{i_2}^{-1} g_{i_3} \cdots g_{i_{m-1}}^{-1} g_{i_m} : i_1, \hdots, i_m \in [5]^m, m\text{ odd}\}\\
    &= \{g_{i_1} g_{i_2} g_{i_3} \cdots g_{i_{m-1}} g_{i_m} : i_1, \hdots, i_m \in [5]^m, m\text{ odd}\},
\end{align*}
where the latter equality uses the fact that $g_i^2 = (I,I,I,I,I,I)$ for all $I \in [5]$.

We next describe the elements of this coset. Let $Q \subseteq P_1^6$ be the subgroup of $P_1^6$ whose elements of the form
\[
    Q := \{((-1)^{a_1}I, \hdots, (-1)^{a_6}I) : a_1 + \cdots + a_6 \equiv 0 \mod 2\}.
\]
If we think of each of $g_1, \hdots, g_5$ as length 6 vectors over the alphabet $\{X,Y,Z\}$, then each pair of generators has even Hamming distance. Thus, by the (anti-)commutativity rules of $X,Y,Z$ we have that for all $i,j \in [5]$ there exists $q_{i,j} \in Q$ such that
\[
    g_i g_j = q_{i,j} g_j g_i.
\]
Thus, by a suitable induction, we can prove that every element of $c\in C$ has the form
\[
    c = q g_1^{a_1} \cdots g_5^{a_5},
\]
where $a_1, \hdots, a_5 \in \{0,1\}$ and $a_1 + \cdots +a_5 \equiv 1 \mod 2$. There are $16$ such products of the form $g_1^{a_1} \cdots g_5^{a_5}$ whose sum of exponents is odd. If $a_1+\cdots +a_5 = 1$, then we just get $g_1, \hdots, g_5$. If $a_1 + \cdots + a_5 = 3$, one can verify that there always exists $i \in [6]$ such that $(g^{a_1}g^{a_2}g^{a_3}g^{a_4}g^{a_5})_i = \pm XYZ$. Finally, we can compute that
\[
    g_1g_2g_3g_4g_5 = (Z,Z,Z,Z,Z,-Z).
\]
Thus, any $c \in C \cap \{X,Y,Z\}^6$ must be of the form $q g_1^{a_1} \cdots g_5^{a_5}$ for some $q \in Q$ and $a_1 + \cdots + a_5 \in \{1,5\}$. If $a_1 + \cdots + a_5 = 1$, then these products are just $c \in \{g_1, \hdots, g_5\}$. However, if $a_1 + \cdots + a_5 = 5$, observe that $g_1g_2g_3g_4g_5$ has an odd number of negative signs, but every $q \in Q$ has an even number of negative signs. Thus, $(Z,Z,Z,Z,Z,Z) \not\in C$. Therefore,
\[
C \cap \{X,Y,Z\}^6 = \{g_1, g_2, g_3, g_4, g_5\}.
\]
This shows that $\PAULI$ embeds into $P_1$, so we have proved Theorem~\ref{thm:non-Abelian}.
\begin{remark}
Analogous to Remark~\ref{rem:coset}, we are actually embedding into the index-$2$ coset of $P_1$ generated by $\{X, Y, Z\}$. This coset is isomorphic to the dihedral group $D_4$ (also called $D_8$) on $8$ elements.
\end{remark}

\begin{remark}
Recent breakthroughs by Lichter--Pago~\cite{lichter2024} and Zhuk~\cite{Zhuk25} show analogous results in the context of CSP satisfiability algorithms. In particular, they give examples of CSPs tractable due to non-Abelian group equations but failed to be solved by Affine Integer Programming (AIP)~\cite{brakensiek2021Promise,barto2021Algebraica} and even stronger variants (e.g., \cite{brakensiek2020Powera,ciardo2023CLAP}). Interestingly, $D_4$ is also the group which Zhuk bases his counterexample on, but otherwise we (the authors) are unaware of any direct technical connection between these algorithmic results and Theorem~\ref{thm:non-Abelian}.
\end{remark}

\begin{remark}
Another fgpp-equivalent representation of $\PAULI$ is
\begin{align*}
    R = \{(1,1,1,1,1,1), (2,2,2,2,2,2), (0,1,2,1,0,2), (2,0,1,2,1,0), (1,2,0,0,2,1)\}
\end{align*}
with $(0,0,0,0,0,0)$ present in an Abelian embedding of $R$. Note that $R$ can be viewed as two copies of $\BCK := \{(1,1,1),(2,2,2),(0,1,2),(2,0,1),(1,2,0)\}$ (see \cite{bessiere2020Chain,brakensiek2024Redundancy}) stuck together except the permutations (such as $(0,1,2)$ and $(1,0,2)$) have opposite sign. If you project out any coordinate of $R$, the result fails to have a Mal'tsev embedding (by applying $\Cat_5$ to some permutation of the tuples).
\end{remark}

\subsection{Application: Excluding Mal'tsev Embeddings}\label{subsec:excl}

Based on the current understanding of which predicates have (near) linear non-redundancy, an important tool for identifying which predicates may have nontrivial non-redundancy lower bounds is proving there does not exist an embedding into an infinite Mal'tsev domain. There are a few such arguments in the literature~\cite{bessiere2020Chain,brakensiek2024Redundancy,chen2020BestCase,lagerkvist2020Sparsification} but they are rather ad-hoc based on the particular predicate involved. In this section, we give a few examples of how Catalan identities can greatly streamline such arguments.

For instance, the predicate $\BCK := \{111,222,012,120,201\}$ which first appeared in \cite{bessiere2020Chain} (see Section 7.1 in \cite{brakensiek2024Redundancy} for a more extensive discussion) was excluded from having a Mal'tsev embedding by using the operator $f$ from \Cref{rem:BCK}. However, using the Catalan polymorphism patterns $\Cat_5$, we can make the same deduction in a more principled manner. In particular,
\begin{align*}
        \Cat_5 \begin{pmatrix}0&2&1&1&2\\1&2&2&1&0\\2&2&0&1&1\end{pmatrix} = \begin{pmatrix}0\\0\\0\\\end{pmatrix}\not\in \BCK,
\end{align*}
where $\Cat_5$ is applied to the rows of the matrix.
Therefore, $\BCK$ lacks any Mal'tsev embedding. We now present some more complex examples.

\subsubsection{Proof of Theorem~\ref{thm:CYC-Mal'tsev}}

With the Catalan machinery we have developed, we can now prove Theorem~\ref{thm:CYC-Mal'tsev} with a relatively simple combinatorial argument. First, recall for odd $m \ge 3$, we defined
\[
    \CYCs_m := \{(x,y,z) \in (\Z/m\Z)^3 : x+y+z = 0 \wedge y-x \in \{0,1\}\} \setminus \{(0,0,0)\}.
\]
We seek to prove that by applying $\Cat_{2m-1}$ to a suitable permutation of the $2m-1$ tuples of $\CYCs_m$ we can deduce the tuple $(0,0,0) \not\in \CYCs_m$, proving that $\CYCs_m$ lacks a Mal'tsev embedding by Theorem~\ref{thm:catalan}. In particular, we claim the following.

\setcounter{MaxMatrixCols}{19}
\[
    \Cat_{2m-1} \begin{pmatrix}
    0&1&m-1&2&m-2&3&m-3&\cdots&m-1&1\\
    1&1&0&2&m-1&3&m-2&\cdots&m-1&2\\
    m-1&m-2&1&m-4&3&m-6&5&\cdots&2&m-3
    \end{pmatrix} = \begin{pmatrix}
    0\\
    0\\
    0
    \end{pmatrix}.
\]
More precisely, we build a matrix $M \in (\Z/m\Z)^{3 \times (2m-1)}$ such that for each $i \in [2m-1]$, the $i$th column of $M$ is equal to (modulo $m$)
\begin{itemize}
\item $(-j,-j+1,2j-1)$ if $i = 2j+1$ for $j \in \Z$, or
\item $(j,j, -2j)$ if $i = 2j$ for $j \in \Z$.
\end{itemize}
As a concrete example, when $m=7$, we can verify by inspection that
\[
    \Cat_{13} \begin{pmatrix}
    0&1&6&2&5&3&4&4&3&5&2&6&1\\
    1&1&0&2&6&3&5&4&4&5&3&6&2\\
    6&5&1&3&3&1&5&6&0&4&2&2&4\\
    \end{pmatrix} = \begin{pmatrix}
    0\\
    0\\
    0
    \end{pmatrix}.
\]
\setcounter{MaxMatrixCols}{10}

To finish the proof, it suffices to show for each row of $M$ a series of cancellations which leave a single $0$ remaining.

\paragraph{First row.} Note that $M_{1,m} = M_{1,m+1} = \frac{m+1}{2}.$ Likewise, for all $i \in [m-1]$, we have that $M_{1,m+1-i} = M_{1,m+i}$.\footnote{This can be proved with a case split. If $m+1-i$ is odd, then $M_{1,m+1-i} \equiv -\frac{m-i}{2} \equiv \frac{m+i}{2} \equiv M_{1,m+i} \mod m$ because $m+i$ is even. Likewise, if $m+1-i$ is even, then $M_{1,m+1-i} \equiv \frac{m+1-i}{2} \equiv -\frac{m+i-1}{2} \equiv M_{1,m+i} \mod m$ because $m+i$ is odd. Analogous assertions in this proof are justified in a similar manner.} After applying these nested cancellations, we are just left with $M_{1,1} = 0$.

\paragraph{Second row.} Since $M_{2,1} = M_{2,2} = 1$, those entries cancel. For $i \in [m-2]$, we can verify that $M_{2,m+2-i} = M_{2,m+1+i}$. After applying these nested cancellations, we are left with $M_{2,3} = 0$.

\paragraph{Third row.} Note that $M_{3,m+2} = 0$, so we need to demonstrate two series of cancellations\footnote{Note that if $m$ is even, there is an odd number of terms both to the left and right of $0$, so the cancellation is not possible in this case. This is why we assume $m$ is odd.} to the left and right of this entry. In particular, for all $i \in [(m+1)/2]$, we have that $M_{3,i} = M_{3,m+2-i},$ so all entries to the left cancel out. Likewise, for all $i \in [(m-3)/2]$, we have that $M_{3,m+2+i} = M_{3,2m-i}$, so all the entries to the right cancel out.

This completes the proof of Theorem~\ref{thm:CYC-Mal'tsev}.

\section{The Algebraic Structure of Patterns and Pattern Powers} \label{sec:powers-and-fgpp}

In this section, we delve deeper into the algebraic framework
investigations initiated in Section~\ref{sec:theory} and the structural
restrictions associated with polymorphism patterns. As a first step,
we give the proof of Theorem~\ref{thm:fgppp-pattern}  and its multisorted generalization in Corollary~\ref{cor:pattern_galois} (announced already in Section~\ref{sec:theory}),
making the connection between fgppp-definitions and (multisorted or
regular) polymorphism patterns. However, we also push further, in two
directions. First, generalizing the constructions in Section~\ref{sec:frac}
(the \emph{all rational exponents} results), we study a notion of
\emph{$c$-tupled fgppp-definitions} ($c$-fgppp-definitions) where we allow arbitrary $c$-ary
maps from tuples of variables in one domain to variables in another domain.
We establish a Galois connection for this notion of expressive power
to a novel \emph{power} operation on a polymorphism pattern, and
prove that $c$-fgppp-definitions behave as expected with respect to conditional
non-redundancy; see Section~\ref{sec:fgpp} and~\ref{sec:gadget_reductions}.
We pay special attention to the powers of the \emph{$k$-cube} (or \emph{$k$-universal})
partial polymorphism patterns, which regulate precisely when a lower
bound of $\NRD(R, n) = \Omega(n^{p/q})$, $p/q>1$, can be shown by
$c$-fgppp-reduction from $\OR_k$; see Section~\ref{sec:cube-power}. The case when $c = k - 1$ is particularly interesting since they, in a certain precise sense, are the weakest patterns that exclude linear NRD. We relate them to Mal'tsev embeddings and give an explicit construction showing that any language with a Mal'tsev embedding is invariant under these cube powers.

Second, generalizing the graph girth connection from Section~\ref{sec:nrd-bin}, we study multisorted polymorphism patterns from a more
abstract perspective, and show that for any finite set of multisorted
polymorphism patterns $Q$, the maximum non-redundancy of promise relations
invariant under $Q$ corresponds to a multipartite hypergraph Tur\'an problem;
see Section~\ref{subsec:turan}. This generalizes results of Carbonnel~\cite{carbonnel2022Redundancy}.
This class of problems involves some
famously difficult questions. For example, the maximum
value of $\NRD(R \mid S, n)$ for ternary promise relations $(R,S) \in \inv(Q)$
corresponds to the \emph{Erd\H{o}s box problem}; see Section~\ref{subsec:turan_examples}
for this and other examples. 

As an application, we show that promise relations preserved by the \emph{$k$-edge} polymorphism patterns have non-redundancy $O(n^{r-\lfloor r/k \rfloor /2})$, where $r$ is the arity of the relation. These are generalizations of Mal'tsev polymorphisms, with similar consequences; see Section~\ref{subsec:turan_examples}.
    
    \subsection{FGPP Definitions and Pattern Powers}
    \label{sec:fgpp}

We first consider the relational counterpart to patterns, which turns out to be a  common generalization of {\em cone-definitions}~\cite{chen2020BestCase} and {\em functionally guarded pp-definitions}~\cite{carbonnel2022Redundancy} (fgpp-definitions) in the promise setting,
and its generalization to $c$-tupled fgppp-definitions.
First, we generalize the notion of an atom $R(x_1, \ldots, x_r)$ in a qfpp-formula to allow terms of the form $f_i(x_{i_1}, \ldots, x_{i_m})$ (where $f_i$ is a function symbol of arity $i_m$) in place of a variable $x_i$, and thus allow expressions of the form $R(f_1(x_{1_1}, \ldots, x_{1_m}), \ldots, f_r(x_{1_r}, \ldots, x_{r_m})$. If each $f_i$ is a function  $f_i \colon D^{i_m}_1 \to D_2$ for two sets $D_1, D_2$, and $R_2 \subseteq D^r_2$,  then $R(f_1(x_{1_1}, \ldots, x_{1_m}), \ldots, f_r(x_{1_r}, \ldots, x_{r_m}))$ is interpreted in the obvious way, i.e., an assignment $h$ from the variables of the atom to $D_1$ satisfies the atom if \[(f_1(h(x_{1_1}), \ldots, h(x_{1_m})), \ldots, f_r(h(x_{1_r}), \ldots, h(x_{r_m}))) \in R.\] 

\begin{definition} \label{def:fgppp}
    Let $(S, T)$ be an $r$-ary promise relation over a domain $D_1$ and let  $(\cla, \clb)$ be a promise constraint language over $D_2$. We say that $(\cla, \clb)$ {\em $c$-tuple fgpp-defines} ($c$-fgpp-defines) $(S, T)$ if there are functions $g_1, \ldots, g_m \colon D^c_1 \to D_2$ where
    \[S(x_1, \ldots, x_r) \equiv \varphi^{\cla}(\ell_1, \ldots, \ell_k)\]
    and  
        \[T(x_1, \ldots, x_r) \equiv \varphi^{\clb}(\ell_1, \ldots, \ell_k)\] 
    where each $\ell_i$ is a term of the form $g_i(x_{i_1}, \ldots, x_{i_c})$ for $g_i \in \{g_1, \ldots, g_m\}$, and where $\varphi(\ell_1, \ldots, \ell_k)^{\cla}$ is a qfpp-formula over $\cla$, and  $\varphi(\ell_1, \ldots, \ell_k)^{\cla}$ the corresponding qfpp-formula over $\clb$.
\end{definition}

If $c = 1$ then we say fgpp-definition rather than $1$-fgpp-definition. We remark that if $D_1 = D_2$ then by letting $g_i = \pi^1_1$ be the identity function we may effectively also allow variables from $x_1, \ldots, x_r$ to directly appear in an fgpp-definition.
Note that nothing in Definition~\ref{def:fgppp} prevents multisorted relations: if $(P_1, Q_1)$ is a multisorted promise relation over $D_1 \times \ldots \times D_r$ and $(P_2, Q_2)$ a multisorted promise relation over $E_1 \times \ldots \times E_s$, then a $c$-tupled fgpp-definition of $(P_1, Q_1)$ over $(P_2, Q_2)$ are allowed to use functions from $D^c_i$ to $E_j$ for $i \in [r], j \in [s]$. However, in the multisorted setting we are primarily interested in the case when $c = 1$ due to the strong connection to hypergraph Tur{\'a}n bounds in Section~\ref{sec:hypergraph-turan}.

Similar to qfppp-definitions we introduce the stronger notion of an {\em $c$-fgppp-definition} as a sequence of strict relaxations and $c$-fgpp-definitions. Let us now continue by relating fgppp-definitions to polymorphism patterns.
To accomplish this we from $P$ define a new pattern $P^c$ which under certain assumptions is guaranteed to commute with $P$. Hence, let $P$ be a pattern and for $c \geq 1$ we define a new pattern $P^c$ via $c$-tuples of
patterns from $P$. More precisely, assume $P$ has arity $n$ and consists of $m$ patterns
$P_i=(X_i,y_i)$ where $P_i$ is over a set of variables $V_i$. 
Let $S \in \binom{[m]}{c}$ with $S=\{i_1, \ldots, i_c\}$ where $i_1 < \ldots < i_c$. 
For each such $S$, $P^c$ contains a pattern $P_S=(X_S,y_S)$ over
variables $V_S=V_{i_1} \times \ldots \times V_{i_c}$, where
for each $p \in [n]$ we have $X_S[p]=(X_{i_1}[p], \ldots, X_{i_c}[p])$
and $y_S=(y_{i_1}, \ldots, y_{i_c})$. Here, we treat each distinct
tuple of variables as a distinct variable for notational convenience
(e.g., instead of a tuple $(x_{a_1}, \ldots, x_{a_c})$ we could
use a single variable name $x_{a_1,\ldots,a_c}$ but that introduces
unreasonable overhead in the proofs to ``extract'' parts of the
variable index). We call the pattern $P^c$ the \emph{$c$th power} of $P$ and the basic question is now when a given promise language is invariant under $P^c$ for some $c \geq 1$. 

\begin{definition}
    Let $(S, T)$ be a promise language over $D$ and let $P$ be a polymorphism pattern. For $c \geq 1$ we say that $P$ \emph{$c$-preserves} $(S,T)$ if (1) $P^c$ is consistent, (2) $|\interpretation{D}{P^c}| = 1$, and (3) and $P^c \in \pattern(S,T)$.
\end{definition}

We remark that while $|\interpretation{D}{P}|  > 1$ can sometimes be useful for patterns (for the multisorted application in Section~\ref{subsec:turan}) it does not seem to be useful for powers of patterns since the power is then trivial. We illustrate these complications with an example.

\begin{example} \label{ex:pattern_power}
    A \emph{majority} polymorphism is a ternary polymorphism $m$
that satisfies $m(x,x,y)=m(x,y,x)=m(y,x,x)=x$ for all $x$.
As a total polymorphism, it guarantees that a language is
\emph{2-decomposable}, i.e., all relations in the language can be
written as intersections of binary constraints (and hence the
non-redundancy is at most quadratic). 
Let $P=\{((x,x,y),x), ((x,y,x),x), ((y,x,x),x)\}$
be the corresponding polymorphism pattern.
Since $\interpretation{D}{P}$ is partial (if $D$ is not Boolean) the presence of $P$ does not imply that a relation is 2-decomposable. However, 
the absence of $P$ guarantees that a relation is not 2-decomposable.
The square of $P$ is the pattern
\[
P^2=  \{(((x,x), (x,y), (y,x)), (x,x)), (((x,y), (x,x), (y,x)), (x,x)),
  (((x,y), (y,x), (x,x)), (x,x))\}\]
  and if we rename variables we obtain
  \[\{((x,y,z), x), ((x,y,z), y), ((x,y,z),z)
  \}.\]
 Hence, $P^2$ is
inconsistent. This agrees well with the fact that the relation 
$a+b+c+d=1$ is 2-fgpp-definable over the equality relation, e.g., using maps $g_1(a,b)=a+b$ and $g_2(c,d)=1-c-d$. Note also that $P^3$ (after renaming variables) yields only the identity $((x,y,z), v)$. Thus, $\interpretation{D}{P^3}$ contains every total, ternary function over $D$, and $P^3$ is thus trivial. 

On the other hand, let \[P=\{((x,x,x,x,y), x), \ldots,((y,x,x,x,x),x)\}\]
be the \emph{near-unanimity} pattern of arity 5 (5-NU).
Similarly as above, the absence of $P$ as a polymorphism pattern
guarantees that a relation is not 4-decomposable.
Then $P^2$ is the polymorphism pattern that defines the 5-ary partial
threshold operation $p(x_1,\ldots,x_5)$ whose output is the majority
value if at least three of its arguments take the same value,
otherwise $p$ is undefined. Thus by Lemma~\ref{lm:pattern-power-consequence},
any relation not preserved by $p$ can 2-fgpp-define a relation $R$
that is not 4-decomposable.

\end{example}

There is a strong similarity between pattern powers and the standard power operation in universal algebra even though the two operations are not equivalent.

\begin{remark} \label{example:3cube}
    Recall that the well-known algebraic $c$-th {\em power} operation of an $n$-ary function $f$ over a domain $D$ produces the function $\mathrm{pow}(f,c)$ on $D^c$ which for any sequence $t^{(1)}, \ldots, t^{(n)} \in D^c$ is defined as $f(t^{(1)}, \ldots, t^{(n)}) = (f(t^{(1)}_1, \ldots, t^{(n)}_1), \ldots, f(t^{(1)}_c, \ldots, t^{(n)}_c))$. For a partial function $p$ we can similarly define its power $\mathrm{pow}(p,c)$ to be defined if and only if each application of $p$ is defined. Then, if $p \in \interpretation{D}{P}$ for a pattern $P$, one can verify that for any $c \geq 1$ $\mathrm{pow}(p,c)$ is a {\em superfunction} of some function  in $\interpretation{D^c}{P}$ in the sense that it satisfies $P$ over $D^c$ but might be defined for some additional tuples. However, perhaps contrary to intuition, $\mathrm{pow}(p,c) \in \interpretation{D^c}{P^c}$ does in general {\em not} hold. To see this, consider e.g. the identities defining the ternary {\em cube} (also called {\em 3-universal} in Lagerkvist and Wahlstr\"om~\cite{lagerkvist2021Coarse})
    \[U_3 = 
\left\{
\begin{array}{c}
((x,x,x,y,y,y,y), x), \\
((x,y,y,x,x,y,y), x), \\
((y,x,y,x,y,x,y), x) \\
\end{array}
\right\}.
\]
It is easy to verify that $\interpretation{\{0,1\}}{U_3} = \{u_3\}$ for a uniquely defined partial operation $u_3$ and that $\mathrm{pow}(u_3, 2)$ satisfies $U_3$ over $\{0,1\}^2$. However, if we e.g.\ consider $U^2_3$ then we obtain a system of identities equivalent to 

\[
\left\{
\begin{array}{c}
(((x,x), (x,y), (x,y), (y,x), (y,x), (y,y), (y,y)), (x,x)), \\
(((x,y), (x,x), (x,y), (y,x), (y,y), (y,x), (y,y)), (x,x)), \\
(((x,y), (y,x), (y,y), (x,x), (x,y), (y,x), (y,y)), (x,x)) \\
\end{array}
\right\}
\] 
over variables $\{(x,x), (x,y), (y,x), (y,y)\}$, and we see that $\mathrm{pow}(u_3, 2)$ does not satisfy these identities over $\{0,1\}^2$ since it e.g. is undefined for the instantiation $\alpha((x,x)) = (0,1), \alpha((x,y)) = (0,0)$, $\alpha((y,x)) = (1,0)$, $\alpha((y,y)) = (1,1)$ since
\[
\begin{aligned}
(\alpha((x,x)), \alpha((x,y)), \alpha((x,y)), \alpha((y,x)), \alpha((y,x)), \alpha((y,y)), \alpha((y,y))) \\
= ((0,1), (0,0), (0,0), (1,0), (1,0), (1,1), (1,1))
\end{aligned}
\]
but $u^2_3((0,1), (0,0), (0,0), (1,0), (1,0) (1,1), (1,1))$ is undefined.
\end{remark}

We now prove that $P^c$ and $P$ commute in the following sense.

\begin{definition}
    Let $f_1$ and $f_2$ be $n$-ary partial functions over $D_1$ and $D_2$, respectively. We say that $f_1,f_2$ {\em commute} with a $c$-ary $g \colon D_1^c \to D_2$ if 
    \[f_2(g(x^1_1, \ldots, x^c_1), \ldots, g(x^1_n, \ldots, x^c_n)) = g(f_1(x^1_1, \ldots, x^1_n), \ldots, f_1(x^c_1, \ldots, x^c_n))\] and is defined for all $(x^1_1, \ldots, x^1_n), \ldots, (x^c_1, \ldots, x^c_n) \in \dom(f_1)$.
\end{definition}

\begin{proposition} \label{prop:commutes_c_ary}
    Let $P$ be an $n$-ary polymorphism pattern, let $c \geq 1$ such that $P^c$ is consistent, let $D_1$ and $D_2$ be finite sets. For every $f_1 \in \interpretation{D_1}{P}$ there exists $f_2 \in \interpretation{D_2}{P^c}$ such that $f_1, f_2$ commute with every $c$-ary $g \colon D^c_1 \to D_2$.
\end{proposition}

\begin{proof}
    Let $P_1, \ldots, P_m$ be an enumeration of the elements in $P$ and let $f_1 \in \interpretation{D_1}{P}$.
Assume that $(x^1_1, \ldots, x^1_n), \ldots, (x^c_1, \ldots, x^c_n) \in \dom(f_1)$, and let $P_{i_1}, \ldots, P_{i_c}$ be witnessing patterns (where there may be multiple matching patterns for each $P_{i_j}$, but the value of $f_1(x^j_1, \ldots, x^j_n)$ is invariant under the choice of pattern, since otherwise $P$ would be inconsistent and we could not choose $f_1 \in \interpretation{D_1}{P}$). 
This corresponds to the pattern $P'=P_{\{i_1, \ldots, i_c\}} \in P^c$; note that for every function $g \colon D_1^c \to D_2$ the tuple $(g(x_1^1, \ldots, x_1^c), \ldots, g(x_n^1, \ldots, x_n^c)) \in D_2^n$ matches the pattern $P'$. Now, either $P'=(X,y)$ is a pattern where $y$ does not occur in $X$, in which case we can choose $f_2 \in \interpretation{D_2}{P^c}$ so that the criterion is satisfied,
or otherwise $f_2(g(x^1_1, \ldots, x^1_c), \ldots, g(x^n_1, \ldots, x^n_c))$ is an instantiation of $P_{\{i_1, \ldots, i_c\}}$ and \[g(f_1(x^1_1, \ldots, x^1_n), \ldots, f_1(x^n_1, \ldots, x^n_c)) = f_2(g(x^1_1, \ldots, x^1_c), \ldots, g(x^n_1, \ldots, x^n_c)\]
by a similar reasoning to Proposition~\ref{prop:commutes}.
\end{proof}

\begin{proposition} \label{prop:if_fgpp_then_pattern}
    Let $P$ be an $n$-ary polymorphism pattern and let $c \geq 1$. Let $(S_1, T_1)$ and $(S_2, T_2)$ be promise languages over $D_1$ and $D_2$.
    If $(S_1, T_1)$ is $c$-fgppp-definable over $(S_2, T_2)$, and $(S_2, T_2)$ is $c$-preserved by $P$, then $(S_1, T_1)$ is preserved by $P$. 
\end{proposition}

\begin{proof}
    We only consider the case when $(S_1, T_1)$ is $c$-fgpp-definable since relaxations follow via arguments completely analogous to those in Lemma~\ref{lemma:is_qfpp_closed}.
    Thus, let 
    \[S_1(x_1, \ldots, x_r) \equiv \varphi^{S_2}(g_1(x_{1_1}, \ldots, x_{1_c}), \ldots, g_k(x_{k_1}, \ldots, x_{k_c}))\] 
    and 
    \[T_1(x_1, \ldots, x_r) \equiv \varphi^{T_2}(g_1(x_{1_1}, \ldots, x_{1_c}), \ldots, g_k(x_{k_1}, \ldots, x_{k_c}))\] 
    for guarding function maps $g_i \colon D^c_1 \to D_2$. Let 
    \[S'_1(y_1, \ldots, y_k) \equiv \varphi^{S_2}(y_1, \ldots, y_k)\] 
    and 
    \[T'_1(y_1, \ldots, y_k) \equiv \varphi^{T_2}(y_1, \ldots, y_k).\] 
    Then $(S'_1, T'_1)$ is qfpp-definable over $(S_2, T_2)$ and is by Theorem~\ref{thm:galois} thus invariant under $\interpretation{D_2}{P}$.

    Assume that there exists $P_1^{D_1} \in \interpretation{D_1}{P}$ and $t^{(1)}, \ldots, t^{(n)} \in S_1$ such that $P^{D_1}(t^{(1)}, \ldots, t^{(n)}) = t \notin T_1$. 
    We observe that for each $t^{(i)}$, $1 \leq i \leq n$, we must have 
    \[(g_1(t^{(i)}_{1_1}, \ldots, t^{(i)}_{1_c}), \ldots, g_k(t^{(i)}_{k_1}, \ldots, t^{(i)}_{k_c})) \in S'_1.\] 
    Proposition~\ref{prop:commutes_c_ary} then implies that there exists $P^{D_2} \in \interpretation{D_2}{P}$ such that $P^{D_1}$ and $P^{D_2}$ commute with all $c$-ary guarding functions from $D_1$ to $D_2$. This, in turn, implies that 
     \begin{equation}
    \begin{aligned}
    & \big( P^{D_2}\big(g_1(t^{(1)}_{1_1}, \ldots, t^{(1)}_{1_c}), \ldots, g_1(t^{(n)}_{1_1}, \ldots, t^{(n)}_{1_c})\big), \ldots,  
    P^{D_2}\big(g_k(t^{(1)}_{k_1}, \ldots, t^{(1)}_{k_c}), \ldots, g_k(t^{(n)}_{k_1}, \ldots, t^{(n)}_{k_c})\big) \big) \\
    &= \big(g_1(P^{D_1}(t^{(1)}_{1_1}, \ldots, t^{(n)}_{1_1}), \ldots, P^{D_1}(t^{(1)}_{1_c}, \ldots, t^{(n)}_{1_c})), \ldots, g_k(P^{D_1}(t^{(1)}_{k_1}, \ldots, t^{(n)}_{k_1}), \ldots, P^{D_1}(t^{(1)}_{k_c}, \ldots, t^{(n)}_{k_c}))\big) \\
    &= \big(g_1(t_{1_1}, \ldots, t_{1_c}), \ldots, g_k(t_{k_1}, \ldots, t_{k_c})\big) \notin T'_1.
    \end{aligned}
    \end{equation}
    This contradicts the assumption that $P^{D_2}$ preserves $(S'_1, T'_1)$. 
\end{proof}

Next, we prove a converse result which shows that we can $c$-fgppp-define $(S_1, T_1)$ if it is invariant under $P$ whenever $(S_2, T_2)$ is invariant under $P^c$. For simplicity we present it for individual relations but the generalization to promise constraint languages is effortless.

\begin{lemma} \label{lemma:can_fgppp_define}
Let $(S_1, T_1)$, $(S_2, T_2)$ be promise relations over $D_1$, $D_2$. If $P \in \pattern{(S_1, T_1)}$  whenever $P$ $c$-preserves $(S_2, T_2)$ then $(S_2, T_2)$ $c$-fgppp-defines $(S_1, T_1)$.
\end{lemma}

\begin{proof}
  Assume that $P \in \pattern(S_1, T_1)$ whenever $P^c \in \pattern(S_2, T_2)$, for a consistent pattern $P^c$.
  We begin by introducing some notation.

    \begin{enumerate}
        \item 
        Let $|S_1| = m$ and let $t_1, \ldots, t_m \in S_1$ be an enumeration of the tuples in $S_1$.
        \item 
        Let $r_1 \geq 1$ be the arity of $S_1$ and $r_2 \geq 1$ the arity of $S_2$. For each $1 \leq i \leq r_1$ we let $c^i = (t^{(1)}_{i}, \ldots, t^{(m)}_{i})$ be the $i$th column of $S_1$ under the enumeration $t_1, \ldots, t_m$. We assume without loss of generality that $|\{c^1, \ldots, c^{r_1}\}| = r_1$ (all columns are distinct) since we can otherwise qfppp-define relations of lower arity by identifying arguments. These relations can, in turn, qfppp-define the original relation by equality constraints.
    \end{enumerate}

Consider the $r_1$-ary relations $S'_1$ and $T'_1$ defined via $c$-fgpp-formulas $\varphi^{S_2}(x_1, \ldots, x_{r_1})$ and $\varphi^{T_2}(x_1, \ldots, x_{r_1})$ over variables $x_1, \ldots, x_{r_1}$ where we for every sequence $s^{(1)}, \ldots, s^{(m)} \in S_2$, with columns \[d^1 = (s^{(1)}_{1}, \ldots, s^{(m)}_{1}), \ldots, d^{r_2} = (s^{(1)}_{r_2}, \ldots, s^{(m)}_{r_2})\]
add

\begin{enumerate}
    \item 
    $S_2(g_1(x_{h(1, 1)}, \ldots, x_{h(1,c)}), \ldots, g_{r_2}(x_{h(r_2, 1)}, \ldots, x_{h(r_2, c)}))$ to $\varphi^{S_2}$, and 
    \item 
    $T_2(g_1(x_{h(1, 1)}, \ldots, x_{h(1,c)}), \ldots, g_{r_2}(x_{h(r_2, 1)}, \ldots, x_{h(r_2, c)}))$ to $\varphi^{T_2}$, 
\end{enumerate}
    if there exist $g_1, \ldots, g_{r_2} \colon D^c_1 \to D_2$ and $h \colon [r_2] \times [c] \to [r_1]$ such that $g_1(c^{h(1,1)}, \ldots, c^{h(1,c)}) = d^1, \ldots, g_{r_2}(c^{h(r_2, 1)}, \ldots, c^{h(r_2, c)}) = d^{r_2}$. 

Similarly, we add
    the equality constraint \[g_1(x_{h(1, 1)}, \ldots, x_{h(1,c)}) = g_{2}(x_{h(2, 1)}, \ldots, x_{h(2, c)})\] to $\varphi^{S_2}$ and $\varphi^{T_2}$ if there are  $g_1, g_2 \colon D^c_1 \to D_2$ and $h \colon [2] \times [c] \to [r_1]$ such that $g_1(c^{h(1,1)}, \ldots, c^{h(1,c)}) = g_{2}(c^{h(2, 1)}, \ldots, c^{h(2, c)})$. The equality constraints are necessary to make sure that the resulting pattern can be defined in a consistent way. 
    
    We then view $c^1, \ldots, c^{r_1}$ as an incomplete pattern $P'$ by viewing $D_1$ as a set of variables $V$ and each $c^i$ as a tuple of variables. Say that a pattern $P$ over $V$ is an {\em instantiation}  of $P'$ if $((x_1, \ldots, x_{m}), y) \in P$ for some $y \in V$ if and only if $(x_1, \ldots, x_m) \in P'$. Then the following statements are true.
    \begin{enumerate}
        \item 
        If $t \in S'_1$ then there is an instantiation $P$ of $P'$ such that  $P^c$ is consistent and $P^c \in \pattern(S_2, S_2)$ and such that $t$ (viewed as an $m$-ary function) is a subfunction of a partial function in $\interpretation{D_1}{P}$. To see this, we construct the pattern $P$ from $P'$ by letting it contain $(c^i, t[i])$ for each $1 \leq i \leq r_1$. We first claim that $P$ does not contain any pattern $(c^i, c)$, for $c \notin \{c^i_1, \ldots, c^i_{m}\}$. To see this, consider maps $h_c,h_c'$ defined such that $h_c(c^i) = h'_c(c^i) = c^i$ but such that $h_c(c) \neq h'_c(c)$. Then the above construction will contain $h_c(x_i) = h'_c(x_i)$ (or, to be precise, $c$-ary maps defined such that $h_c(x_i, \ldots, x_i) = h'_c(x_i, \ldots, x_i)$) and will not be satisfied by $t$. Similarly, assume that $P$ is inconsistent, i.e., that there are $(c^i, a), (c^j, b) \in P$ and maps $h, h' \colon D_1 \to D_1$ such that $h(c^i) = h'(c^j)$ but such that $h(a) \neq h(b)$. It again follows that the above construction contains an equality constraint that forbids this possibility. Entirely analogous arguments can be used for $P^c$ to show that it must be consistent.
        \item 
        If $t \in T'_1$ then there is an instantiation $P$ of $P'$ where $P^c$ is consistent and $P^c \in \pattern(S_2, T_2)$ such that $t$ (viewed as an $m$-ary function) is a subfunction of a partial function in $\interpretation{D_1}{P}$. The arguments are entirely analogous to the case above.
      \end{enumerate}
    Hence, if there exists $f \in T'_1$ where $f \notin T_1$ then there exists $m$-ary $P^c \in \pattern(S_2, T_2)$ such that $f$ is a subfunction of a partial function in $\interpretation{D_1}{P}$ but where $f \notin \pPol(S_1, T_1)$.

    We conclude that $T'_1 = T_1, S_1 \subseteq S'_1$, and that $(S_1, T_1)$ is fgppp-definable via a strict relaxation of $(S'_1, T'_1)$.
\end{proof}

For $c = 1$ we obtain the following precise characterization which  generalizes Carbonnel~\cite{carbonnel2022Redundancy} in the promise setting. 

\begin{theorem}
\label{thm:fgppp-pattern}
Let $(S_1, T_1)$ be a promise relation over $D_1$ and $(S_2, T_2)$ a promise relation over $D_2$. Then $\pattern(S_2, T_2) \subseteq \pattern(S_1, T_1)$ if and only if $(S_2, T_2)$ fgppp-defines $(S_1, T_1)$.
\end{theorem}

For $c \geq 1$ we for a promise relation $(S,T)$ let \[\pattern^c(S,T) = \{P \mid P \text{ is a pattern},  P \, c\text{-preserves } (S,T)\}.\] 
Hence, $\pattern^{c}(S,T)$ is the set of patterns $P$ where $P^c$ is consistent and preserves $(S,T)$. We obtain the following.

\begin{theorem}
\label{thm:fgppp-pattern_c}
Let $(S_1, T_1)$ be a promise relation over $D_1$ and $(S_2, T_2)$ a promise relation over $D_2$, and let $c \geq 1$. Then $\pattern^{c}(S_2, T_2) \subseteq \pattern(S_1, T_1)$ if and only if $(S_2, T_2)$ $c$-fgppp-defines $(S_1, T_1)$.
\end{theorem}

For multisorted patterns we concentrate on $c = 1$ and obtain the following characterization. 

\begin{corollary}
        Let $Q$ be a set of multisorted polymorphism patterns. If $(S_1, T_1)$ over disjoint $D_1, \ldots, D_r$ is fgppp-definable over $(S_2, T_2)$ over disjoint $E_1, \ldots, E_s$ which is preserved by $Q$, then $(S_1, T_1)$ is preserved by $Q$ ($(S_1, T_1) \in \inv(Q)$).
\end{corollary}

Note that this implies that $\inv(Q)$, when $Q$ is a set of multisorted patterns, is closed under fgppp-definitions.

\begin{corollary} \label{cor:pattern_galois}
Let $(S_1, T_1)$ be a multisorted promise relation over disjoint $D_1, \ldots, D_r$ and $(S_2, T_2)$ a multisorted promise relation over disjoint $E_1, \ldots, E_s$. Then $\mpattern(S_2, T_2) \subseteq \mpattern(S_1, T_1)$ if and  only if $(S_2, T_2)$ fgppp-defines $(S_1, T_1)$.
\end{corollary}

Finally, the following statement is useful for constructing a canonical relation not invariant under a given pattern (power) (cf. Lemma~\ref{lemma:pippenger1}). We state it in the non-multisorted setting but it naturally extends to this setting, too.

\begin{lemma} \label{lm:pattern-power-consequence}
  Let $P$ be a polymorphism pattern and $c \in \mathbb{N}$ such that
  $P^c$ is consistent and $|\interpretation{D}{P^c}|=1$ for every domain $D$. 
  Let $(S_1,T_1)$ be a promise relation such that $P \notin \pattern^c(S_1,T_1)$. 
  Then $(S_1,T_1)$ $c$-fgpp-defines a promise relation $(S_2,T_2)$ such that $P \notin \pattern(S_2,T_2)$.
\end{lemma}
\begin{proof}
  Let $P$ be an $n$-ary pattern and $c \in \mathbb{N}$ such that $P^c$ is consistent and $P \notin \pattern^c(S_1,T_1)$
  for a promise relation $(S_1,T_1)$ of arity $r$ over some domain $D$. That is, $P^c \notin \pattern(S_1,T_1)$.
  Let $p \in \interpretation{D}{P^c}$ and let $t^1, \ldots t^n \in S_1$ be such that
  $p(t^1,\ldots,t^n)=u$ is defined and $u \notin T_1$.
  For every $i \in [r]$, let $(X^i,y^i) \in P^c$ be a pair that matches $((t_i^1, \ldots, t_i^n), u_i)$.
  Let $y^i=X_{n_i}^i$; this is defined since $|\interpretation{D}{P^c}|=1$.
  Furthermore, let $(X^i,y^i)$ be the combination of patterns
  $(X^{i,1},y^{i,1}), \ldots, (X^{i,c},y^{i,c}) \in P$.
  Let $D_2$ be the domain consisting of all variables occurring in $P$
  and for $i \in [r]$ let $g_i \colon D_2^c \to D$ be any function
  such that $g_i(X^{i,1}_j, \ldots, X^{i,c}_j)=t_j^i$ for each $j \in [n]$.
  Consider the $c$-fgpp-definition
  $\varphi(v_{1,1},\ldots,v_{r,c})=R(\ell_1,\ldots,\ell_r)$ where
  $\ell_i=g_i(v_{i,1}, \ldots, v_{i,c})$ for every $i \in [r]$.
  Let $(S_2,T_2)$ be the promise relation defined by $\varphi$ using $(S_1,T_1)$.
  We claim that $P \notin \pattern(S_2,T_2)$. 
  Consider the tuples $a^1, \ldots, a^n \in D_2^{rc}$ where
  $a_{i,j}^d = X^{i,j}_d$ for every $i \in [r]$, $j \in [c]$, $d \in [n]$.
  Let $q \in \interpretation{D_2}{P}$ (and note that $q$ exists and is unique).
  We claim that $a^1, \ldots, a^n \in S_2$ and that $q(a^1,\ldots,a^n)=b$ is defined with $b \notin T_2$. 
  Indeed, by design, for every $i \in [r]$, $j \in [c]$ the tuple $(a_{i,j}^1,\ldots,a_{i,j}^n)$
  matches the pair $(X^{i,j},y^{i,j}) \in P$, so $q$ produces an output $b$
  where $b_{i,j}=a^{n_i}_{i,j}$ for every $i \in [r]$, $j \in [c]$.
  Thus $b$ under the maps $g_i$ maps to the tuple $u \notin T_1$,
  so $b \notin T_2$. 
\end{proof}

    \subsection{Gadget Reductions for Conditional Non-redundancy}
    \label{sec:gadget_reductions}
We continue by relating fgppp-definitions, and thus also patterns, to conditional non-redundancy. Our goal is to prove the following (compare with Proposition 2 in Carbonnel~\cite{carbonnel2022Redundancy}) statement.
The theorem naturally extends to finite sets of promise relations but we present and prove the results in the singleton setting to simplify the formal details.

\begin{theorem}\label{thm:fgppp-nrd}
Let $(S_1, T_1)$ and $(S_2, T_2)$ be promise relations over $D_1$, respectively, $D_2$. If $(S_2, \widetilde{T}_2)$ $c$-fgppp-defines $(S_1, \widetilde{T}_1)$ for some $c \geq 1$ then $\NRD(S_1 \mid T_1, n) = O_{D_1,D_2,r_1,r_2}(n + \NRD(S_2 \mid T_2, n^c))$.
\end{theorem}

As an immediate corollary of Theorem~\ref{thm:fgppp-pattern}, we have the following algebraic test for bounding non-redundancy.

\begin{corollary}
Let $(S_1, T_1)$ and $(S_2, T_2)$ be promise relations over $D_1$, respectively, $D_2$. 
If $\pattern(S_2, \widetilde{T}_2) \subseteq \pattern(S_1, \widetilde{T}_1)$ then $\NRD(S_1 \mid T_1, n) = O_{D_1, D_2, r_1,r_2}(n + \NRD(S_2 \mid T_2, n))$.
\end{corollary}

For pattern powers we obtain the more general statement.

\begin{corollary}
Let $(S_1, T_1)$ and $(S_2, T_2)$ be promise relations over $D_1$, respectively, $D_2$, and
let $c \geq 1$.
    If $\pattern^c(S_2, \widetilde{T}_2) \subseteq \pattern(S_1, \widetilde{T}_1)$ then $\NRD(S_1 \mid T_1, n) = O_{D_1, D_2, r_1,r_2}(n + \NRD(S_2 \mid T_2, n^c))$.
\end{corollary}

Our overall proof strategy is similar to that of Carbonnel~\cite{carbonnel2022Redundancy} but more involved due to the richer machinery.  We write our proof in a modular fashion by proving Theorem~\ref{thm:fgppp-nrd} for each of the ``atomic'' operations involved in an fgppp-definition. In certain cases the resulting bound is thus sharper than the general bound in Theorem~\ref{thm:fgppp-nrd}. We start with strict relaxations.

\begin{proposition}\label{prop:strict-relaxation-nrd}
Assume $D := D_1 = D_2$ and $r := r_1 = r_2$. If $S_1 \subseteq S_2$ and $\widetilde{T}_2 \subseteq \widetilde{T}_1$, then $\NRD(S_1 \mid T_1, n) \le \NRD(S_2 \mid T_2, n)$.
\end{proposition}
\begin{proof}
Consider any non-redundant instance $(X, Y \subseteq X^r)$ of $\CSP(S_1 \mid T_1)$ with witnessing assignments $\psi_y : X \to D$ for all $y \in Y$. If $\widetilde{T}_2 \subseteq \widetilde{T}_1$, then $T_1 \setminus S_1 \subseteq T_2 \setminus S_2$. Since we also have $S_1 \subseteq S_2$, the assignments $\{\psi_y : y \in Y\}$ also witness that $(X,Y)$ is a non-redundant instance of $\CSP(S_2 \mid T_2)$.
\end{proof}

Second, we handle equality constraints.

\begin{proposition}\label{prop:equality-nrd}
Assume $D := D_1 = D_2$ and $r := r_1 = r_2$. Assume that there exists distinct $i,j \in [r]$ such that
\begin{align}
    S_1(x_1, \hdots, x_r) &= S_2(x_1, \hdots, x_r) \wedge x_i = x_j,\nonumber\\
    \widetilde{T}_1(x_1, \hdots, x_r) &= \widetilde{T}_2(x_1, \hdots, x_r) \wedge x_i = x_j.\label{eq:wtQeq}
\end{align}
Then, $\NRD(S_1 \mid T_1, n) \le \NRD(S_2 \mid T_2, n) + n-1$.
\end{proposition}
\begin{proof}
Let $(X, Y)$ be a non-redundant instance of $\CSP(S_1 \mid T_1)$, and let $\{\psi_y : X \to D \mid y \in Y\}$ be the witnessing assignments. Let $Y_{\EQ} \subseteq Y$ be the set of clauses $y \in Y$ for which $\psi_y(y_i) \neq \psi_y(y_j)$. It is clear that the graph $\{(y_i, y_j) \in X^2 \mid y \in Y_{\EQ}\}$ is a acyclic, so $|Y_{\EQ}| \le |X|-1$.

By negating both sides of (\ref{eq:wtQeq}), we have that
\[
    (T_1 \setminus S_1)(x_1, \hdots, x_r) = (T_2 \setminus S_2)(x_1, \hdots, x_r) \vee x_i \neq x_j.
\]
For each $y \in Y$, we have that $\psi_y(y) \in T_1 \setminus S_1$. Further, by definition of $Y_{\EQ}$, for each $y \in Y \setminus Y_{\EQ}$, we have that $\psi_y(y_i) = \psi_y(y_j)$, so $\psi_y(y) \in T_2 \setminus S_2$. Therefore, $(X, Y \setminus Y_{\EQ})$ is a non-redundant instance of $\CSP(S_2 \mid T_2)$.
\end{proof}

Third, we handle the conjunction of two predicates, which immediately extends to an arbitrary number of predicates.

\begin{proposition}\label{prop:conjunction-nrd}
Assume $D := D_1 = D_2$. Let $f,g: [r_2] \to [r_1]$ be functions such that
\begin{align}
    S_1(x_1, \hdots, x_{r_1}) &= S_2(x_{g(1)}, \hdots, x_{g(r_2)}) \wedge S_2(x_{f(1)}, \hdots, x_{f(r_2)}),\nonumber\\
    \widetilde{T}_1(x_1, \hdots, x_{r_1}) &= \widetilde{T}_2(x_{g(1)}, \hdots, x_{g(r_2)}) \wedge \widetilde{T}_2(x_{f(1)}, \hdots, x_{f(r_2)}).\label{eq:wtQconj}
\end{align}
Then, $\NRD(S_1 \mid T_1, n) \le \NRD(S_2 \mid T_2, n)$.
\end{proposition}
\begin{proof}
Let $(X, Y \subseteq X^{r_1})$ be a non-redundant instance of $\CSP(S_1 \mid T_1)$, and let $\{\psi_y : X \to D \mid y \in Y\}$ be the witnessing assignments. For each $y \in Y$, let $y_f, y_g \in X^{r_2}$ be
\begin{align*}
    y_f &:= (y_{g(1)}, \hdots, y_{g(r_2)}),\\
    y_g &:= (y_{f(1)}, \hdots, y_{f(r_2)}).
\end{align*}
For each $y, y' \in Y$, we have that $\psi_y(y') \in S_1$ if and only if $\psi_y(y'_f) \in S_2$ and $\psi_y(y'_g) \in S_2$. Furthermore, by negating both sides of (\ref{eq:wtQconj}), we can deduce that
\begin{align*}
(T_1 \setminus S_1)(x_1, \hdots, x_{r_1}) &= (T_2 \setminus S_2)(x_{g(1)}, \hdots, x_{g(r_2)}) \vee (T_2 \setminus S_1)(x_{f(1)}, \hdots, x_{f(r_2)}).
\end{align*}
In particular, since $\psi_y(y) \in T_1 \setminus S_1$, we have that $\psi_y(y_f) \in T_2 \setminus S_2$ or $\psi_y(y_g) \in T_2 \setminus S_2$. Let $y_{\rep} \in \{y_f, y_g\}$ be an arbitrary representative for which $\psi_y(y_{\rep}) \in T_2 \setminus S_2$. Then define $Y_{\rep} := \{y_{\rep} : y \in Y\} \subseteq X^{r_2}$. From our discussion, we have that $\{\psi_y : y \in Y\}$ witness that $(X, Y_{\rep})$ is a non-redundant instance of $\CSP(S_2 \mid T_2)$. In particular, $|Y_{\rep}| = |Y|$ as each $\psi_y$ fails to $S_2$-satisfy a unique clause $y_{\rep}$ of $(X, Y_{\rep})$.
\end{proof}

As an immediate corollary of our progress so far, we have that conditional non-redundancy behaves well with respect to qfppp-reductions. However, we now conclude by showing that functional guarding also behaves well with respect to conditional non-redundancy.

\begin{proposition}\label{prop:functional-guarding-nrd}
For $c \geq 1$ let $h \colon [r_2] \times [c] \to [r_1]$ and $g_1, \hdots, g_{r_2} : D^c_1 \to D_2$ be functions such that
\begin{align}
    S_1(x_1, \hdots, x_{r_1}) &= S_2(g_1(x_{h(1,1)}, \ldots, x_{h(1,c)}), \hdots, g_{r_2}(x_{h(r_2, 1)}, \ldots, x_{h(r_2, c)})) \label{eq:123}\\
    \widetilde{T}_1(x_1, \hdots, x_{r_1}) &= \widetilde{T}_2(g_1(x_{h(1,1)}, \ldots, x_{h(1,c)}), \hdots, g_{r_2}(x_{h(r_2, 1)}, \ldots, x_{h(r_2, c)})).\label{eq:234}
\end{align}
Then, $\NRD(S_1 \mid T_1, n) \le O_{r_2}(\NRD(S_2 \mid T_2, n^c))$.
\end{proposition}

\begin{proof}
Let $(X, Y \subseteq X^{r_1})$ be a non-redundant instance of $\CSP(S_1 \mid T_1)$ with witnessing assignments $\{\psi_y : X \to D_2 \mid y \in Y\}$. Let $X_{+} := X^c \times [r_2]$ and for each $y \in Y$, let 
\[
y_{+} := ((y_{h(1,1)}, \ldots, y_{h(1,c)}, 1), \hdots, (y_{h(r_2, 1)}, \ldots, y_{h(r_2, c)}, r_2)) \in X_{+}^{r_2}.
\]
Let $Y_{+} := \{y_{+} \mid y \in Y\}$ and for each $y_{+} \in Y_{+}$ let $\psi_{y_{+}} : X_{+} \to D_1$ be such that for all $(x, i) \in X_{+}$, we have that $\psi_{y_{+}}(x, i) = g_i(\psi_y(x))$, where $\psi_y$ is applied componentwise to $x$ and therefore yields a $c$-ary tuple matching the arity of $g_i$. From (\ref{eq:123}) and (\ref{eq:234}), we can see that $\{\psi_{y_{+}} \mid y_{+} \in Y_{+}\}$ witness that $(X_{+}, Y_{+})$ is a non-redundant instance of $\CSP(S_2 \mid T_2)$ and thus $|Y_{+}| = |Y|$. Since $|X_{+}| = r_2 |X|^c$ we therefore have $\NRD(S_1 \mid T_1, n) \le \NRD(S_2 \mid T_2, n^c \cdot r_2)$, and since $\NRD(S_2 \mid T_2, n^c \cdot r_2) \in O_{r_2}(\NRD(S_2 \mid T_2, n^c))$ (\Cref{lem:NRD-scale}) 
the desired bound follows.
\end{proof}
Note that this bound, in contrast to Carbonnel~\cite{carbonnel2022Redundancy} who obtain a dependency on $|D|^{|D|}$, does {\em not} depend on the size of the domains.

Last, we conclude the proof of Theorem~\ref{thm:fgppp-nrd}.

\begin{proof}[Proof of Theorem~\ref{thm:fgppp-nrd}]
If $(S_2, T_2)$ fgppp-defines $(S_1, T_1)$, then this definition can be constructed using at most $O_{D_1,D_2,r_1,r_2}(1)$ iterations of the operations in Proposition~\ref{prop:strict-relaxation-nrd}, Proposition~\ref{prop:equality-nrd}, Proposition~\ref{prop:conjunction-nrd}, and Proposition~\ref{prop:functional-guarding-nrd}. Furthermore, the arity of any intermediate predicate pair is at most $O_{D_1,D_2,r_1,r_2}(1)$, so the constant factor bound of Proposition~\ref{prop:functional-guarding-nrd} still applies. Last, as a special case, consider the case when $(S_1, T_1)$ is fgppp-definable with qfppp-definition over $(S_2, T_2)$ which involves one or more $\mathsf{f} = \emptyset$ constraints. Then, $S_1 = T_1 = \emptyset$ and the claim trivially follows.
\end{proof}

\subsection{The Cube-power Polymorphism Patterns} \label{sec:cube-power}

We now review a particularly important class of polymorphism patterns, corresponding to the powers of the \emph{cube} identities mentioned in Remark~\ref{example:3cube}, studied by Berman et al.~\cite{berman2009Varieties}. The $k$-cube polymorphism patterns were referred to as \emph{$k$-universal} partial polymorphisms in previous work~\cite{carbonnel2022Redundancy,lagerkvist2021Coarse}.

\begin{definition} \label{def:kcube}
    Let $k \geq 2$. The \emph{$k$-cube} polymorphism pattern $U_k$ is the pattern with arity $2^k-1$ on variable set $V=\{x,y\}$, with patterns $P_i=(X_i,y)$ for $i \in [k]$ where
    for each $j \in [2^k-1]$, position $j$ of $X_i$ equals $x$ if the $i$th bit of the binary expansion of $j$ is 1, and $y$ otherwise.
\end{definition}

Since $|I_D(U_k)|=1$ for every domain $D$, we let $u_k$ refer to the unique member of $I_D(U_k)$, where the domain $D$ will frequently be understood from context. The significance of $U_k$ is that every relation $R \subseteq D^r$ for some domain $D$ and arity $r$ can fgpp-define the $k$-ary $\mathsf{OR}$ relation if and only if $u_k \notin \pPol(R)$; see~\cite{carbonnel2022Redundancy,lagerkvist2021Coarse}. Then Lemma~\ref{lm:pattern-power-consequence} gives the following. 

\begin{lemma} \label{lm:cube-power-lb}
    Let $R \subseteq D^r$ be a relation over some domain $D$, and assume $u_k^c \notin \pPol(R)$ for some $k > c \geq 1$. Then $\NRD(R,n)=\Omega(n^{k/c})$.
\end{lemma}

\begin{example}  
Consider the following relation from \cite{brakensiek2024Redundancy}.
Define $R=\text{3LIN}_{\Z/3\Z}=\{(x,y,z) \in \{0,1,2\}^3 \mid x+y+z+0 \pmod 3\}$
and $R^*=\text{3LIN}^*_{\Z/3\Z}= R \setminus \{(0,0,0)\}$.
It was shown in \cite{brakensiek2024Redundancy} that $\NRD(R^*,n)=\Omega(n^{1.5})$ and 
$\NRD(R^*,n)=\Oh(n^{1.6}\log n)$. The lower bound now follows from Lemma~\ref{lm:cube-power-lb} as $u_3^2((0,1,2), (1,0,2), (1,1,1), (1,2,0), (1,0,2), (2,2,2), (2,0,1))=(0,0,0)$.
\end{example}

Lemma~\ref{lm:cube-power-lb} is a powerful approach towards lower bounds on non-redundancy; e.g., beyond this example, it also covers all lower bounds of Section~\ref{sec:frac}. However, it is also not a ``universal'' method of showing tight lower bounds on the non-redundancy, as shown at the end of this section.

In our study of relations with near-linear NRD, we will be particularly interested in the power patterns $U_k^{k-1}$, since these are the weakest patterns that exclude near-linear NRD. Let us make some quick observations about these patterns.
The first is a direct observation that follows from the definition (we include the case $U_2^1=U_2$ for uniformity).

\begin{proposition} \label{prop:ukc-description}
  Let $k \in \N$, $k \geq 2$. The polymorphism pattern $U_k^{k-1}$ 
  has arity $2^k-1$, and its patterns are described as follows:
  For every $i \in [k]$, there is a pattern $P_i=(X_i,y_i)$
  where $X_i[p]=X_i[q]$ for $p, q \in [2^k-1]$ if and only if
  $p \oplus q = e_i$ when $p$ and $q$ are viewed as bitstrings
  and $e_i$ is the $i$th unit vector, and where
  $y_i$ is the unique variable that occurs only once in $X_i$.  
\end{proposition}

We continue with the following observations. 

\begin{lemma} \label{lm:ukc-basics}
  The following hold.
  \begin{enumerate}
  \item For any $k \geq 2$, $U_k^{k-1}$ is symmetric in the following
    sense: Let $P_1, \ldots, P_k$ be the patterns of $U_k^{k-1}$
    as in ~\Cref{prop:ukc-description}.
    For every permutation $\pi$ on $[k]$
    there is a permutation $\sigma$ on $[2^k-1]$ such that
    for every $i \in [k]$, the values of pattern $P_i$ read in order
    $\sigma$ match the values of pattern $P_{\pi(i)}$ read in the
    original order. 
  \item Over any domain $D$, $u_p^{p-1}$ implies $u_q^c$ 
    for every $p \geq q \geq 2$ and every $c < q$.
  \item Let $R \subset D^r$ be a relation for some domain $D$ and
    arity $r$. Then $R$ is preserved by $u_k^{c}$ for every $k \geq 2$ and $c<k$
    if and only if $R$ is preserved by $u_r^{r-1}$.
  \end{enumerate}
\end{lemma}
\begin{proof}
  For the first, if we view the arguments of $U_k^{k-1}$ as
  corresponding to non-zero bitstrings $b \in \{0,1\}^k$, 
  then applying $\pi$ directly to the bitstring yields the desired
  permutation $\sigma$.
    For the second, first let $R \subseteq D^r$ be an $r$-ary relation
  preserved by $u_p^{p-1}$ for some $p \geq 2$, 
  and let $t_1, \ldots, t_m \in R$ be such that
  $u_q^{q-1}(t_1,\ldots,t_m)=t$ is defined, $q < p$.
  Select $t' \in R$ arbitrarily. Then
  $u_p^{p-1}(t_1,\ldots,t_m,t',\ldots,t')=t$ is defined,
  by the description in ~\Cref{prop:ukc-description}.
  Thus $u_p^{p-1}$ implies $u_q^{q-1}$.
  To complete the property, we note that for any $k$,
  and any $a \leq b<k$, $u_k^b$ implies $u_k^a$:
  in line with ~\Cref{prop:ukc-description},
  the patterns of $u_k^c$ for any $c$ correspond to sets $S \subset [k]$
  containing $c$ indexes from $[k]$, and such a pattern has the same
  variable in two positions $i$, $j$ if and only if 
  $i$ and $j$ match on all bits present in $S$.
  Thus, for every pattern $P_S$ of $u_k^a$,
  indexing $a$ bits $S \subset [k]$, let $S' \subset [k]$ be any
  subset of $S$ with $b$ bits. Then the pattern
  $P_{S'}$ is strictly less restrictive than $P_S$,
  so any tuple that matches a pattern of $u_k^a$
  also matches a pattern of $u_k^b$. Similarly, the output
  of pattern $P_{S'}$ matches an argument that also
  matches the output of $P_S$. 

  For the last point, we note by the second property that 
  $R$ is preserved by all $u_k^c$, $k \geq 2$, $c<k$
  if and only if $R$ is preserved by $u_k^{k-1}$ for all $k \geq 2$.
  Thus, assume that $R$ is violated by $u_k^{k-1}$
  for some $k \geq 2$. If $k<r$, then $R$ is violated by $u_r^{r-1}$
  by the second property. Otherwise, assume $k>r$,
  and let $t_1, \ldots, t_m \in R$ be such that
  $u_k^{k-1}(t_1,\ldots,t_m)=t$ is defined where $t \notin R$.
  By the first property, we may assume that this application
  of $u_k^{k-1}$ only uses patterns $P_i$ corresponding to the $r$
  least significant bits according to ~\Cref{prop:ukc-description}.
  But then, all ``repeats'' in the pattern have distance at most $2^{r-1}$
  and $u_r^{r-1}(t_1,\ldots,t_{2^r-1})=t$ is defined, showing
  $u_r^{r-1} \notin \pPol(R)$.
\end{proof}

The last property allows us to verify with a finite computation that a relation cannot be shown to have superlinear NRD using $c$-fgppp-definitions of $\OR$ (i.e., via Lemma~\ref{lm:cube-power-lb}).
This may be interesting to compare with Malt'sev embeddings and the Catalan terms from Section~\ref{sec:linear-nrd}: to verify that a language admits a Malt'sev embedding it needs to be invariant under all (infinitely many) Catalan terms.

Unfortunately, the second property of Lemma~\ref{lm:ukc-basics} does
not generalize: in general, if $p/q \leq k/c$, it does not
necessarily follow that $u_p^q$ implies $u_k^c$. 
Let $P_1=(X_1,y_1), \ldots, P_6=(X_6,y_6)$ enumerate the patterns of $U_4^2$,
and let $V$ be the set of variables of $U_4^2$.
Define a relation $R \subseteq V^6$ by $R=\{t_1,\ldots,t_{15}\}$
where $t_i=(X_1[i], \ldots, X_6[i])$. 
Then by design $R$ is not preserved by $u_4^2$, but
one can verify that $R$ is preserved by $u_2^1=u_2$.

Finally, we note that the final property of Lemma~\ref{lm:ukc-basics}
represents a weakness in the method, in that the cube-power patterns $U_p^q$
cannot be used to prove a lower bound of $\Omega(n^{1+\varepsilon})$
for the non-redundancy of an $r$-ary relation $R$ for any $0 < \varepsilon < 1/(r-1)$,
despite the results of Section~\ref{sec:nrd-bin} showing that infinitely many such bounds exist, even for $r \leq 3$ (at least for promise relations, in Section~\ref{subsec:binary-class}, and conjecturally, in Section~\ref{subsec:approach-linear}).
We show a complementary upper bound, showing that $U_p^q$ cannot show a lower bound $\Omega(n^{r-\varepsilon})$ for the non-redundancy of $R$ for any $0 < \varepsilon < 1$ unless $\NRD(R,n)=\Theta(n^r)$.

\begin{lemma} \label{lm:upq-near-r}
  Let $R \subseteq D^r$ for $r \geq 2$ be a relation such that $U_r \in \pattern(R)$.
  Then $U_p^q \in \pattern(R)$ for every $p/q > r-1$. 
\end{lemma}
\begin{proof}
  Assume $u_p^q \notin \pPol(R)$ for some $p/q > r-1$. Then $R$
  $q$-fgpp-defines $\OR_p$, and this definition can be assumed to
  consist of a single constraint, i.e.,
  \[
    \OR_p(x_1,\ldots,x_p) \equiv R(g_1(X_1), \ldots, g_r(X_r))
  \]
  for some maps $g_i \colon \{0,1\}^q \to D$ and some scopes
  $X_i \subseteq X$ of $q$ variables each, where $X=\{x_1,\ldots,x_p\}$. 
  First assume that there is a scope $X_i$, say $X_r$, such that every
  variable in $X_r$ also occurs in $X_j$ for some $j<r$. Then
  $|\bigcup_{i=1}^{r-1} X_i|=|X|=p \leq (r-1)q$.
  Otherwise, let $x_{i_j}$, $j \in [r]$ be variables such that
  $x_{i_j}$ occurs only in $X_j$ for $j \in [r]$. Consider the set of
  assignments to $X$ where $x_t=0$ for every $t \in [r] \setminus \{i_1,\ldots,i_r\}$.
  Each map $g_i$, $i \in [r]$, takes two distinct values under these assignments,
  and this defines a subset $F=\prod_{i=1}^r \{a_i,b_i\} \subseteq D^r$
  such that $|F \setminus R|=1$. Thus we have a witness that
  $R$ violates $u_r$. 
\end{proof}

Last, it is interesting to relate cube powers to Mal'tsev embeddings. First, Lemma~\ref{lm:cube-power-lb} implies that   if $\pPol(P)$ lacks $u_k^{k-1}$, then $\NRD(P, n) = \Omega(n^{k/(k-1)})$, and $P$ lacks a finite Mal'tsev embedding. However, we can give a more direct argument as follows.

\begin{lemma}
  Let $R$ be a finite relation with a Mal'tsev embedding (possibly into an infinite domain). Then $R$ is preserved by $u_k^{k-1}$ for every $k \geq 2$.
\end{lemma}
\begin{proof}
  Let $R \subseteq D^r$ for some domain $D$ and arity $r$. 
  Fix $k \geq 2$ and let $u_k^{k-1}(t_1,\ldots,t_m)=t$ be defined,
  where $t_1, \ldots, t_m \in R$ and $m=2^k-1$.
  Assume that the arguments of $u_k^{k-1}$
  are ordered as in Prop.~\ref{prop:ukc-description}.
  Let $D'$ be the domain of the Mal'tsev embedding of $R$, let
  $\varphi \colon (D')^3 \to D'$ be a Mal'tsev term,
  and let $f_1, f_3, \ldots, f_{2^k-1} \in [\varphi]$ be Catalan terms,
  which exist by Theorem~\ref{thm:catalan}.
  Let $\sigma$ be an ordering on $\{0,\ldots,2^k-1\}$ that
  corresponds to the reflected binary Gray code, where 
  $\sigma(i)$ and $\sigma(i+1)$ differ in precisely one digit.
  The Gray code has a recursively defined structure as follows:
  $\sigma(0)=0$, $\sigma(1)=1$, and the $(n+1)$-bit Gray code
  is defined by letting $\sigma(0)$ through $\sigma(2^n-1)$
  be the $n$-bit Gray code, then letting the second half be the
  first half reflected with the top bit set. 
  We note that due to this symmetry, the terms in any tuple in the
  domain of $u_k^{k-1}$ cancel exhaustively if read in the order of $\sigma$.
  
  \begin{claim}
    Let $u_k^{k-1}(x_1,\ldots,x_{2^k-1})=y$ be defined for some
    tuple of values $x_1, \ldots, x_{2^k-1}, y \in D$. Then
    $f_{2^k-1}(x_{\sigma(1)}, \ldots, x_{\sigma(2^k-1)})=y$.
  \end{claim}
  \begin{proof}
    Let $(x_1,\ldots,x_{2^k-1})$ match the pattern $P_i=(X_i,y_i)$
    for $i \in [k]$ in $U_k^{k-1}$, i.e., $X_i[p]$ corresponds to the
    binary expansion of $p$ with bit $i$ removed,
    and consider the sequence $S=(y_i, x_{\sigma(1)}, \ldots, x_{\sigma(2^k-1)})$,
    where $y_i$ in a sense corresponds to $\sigma(0)=0$.     
    Let $p$ be the first position such that bit $i$ of $\sigma(p)$
    is set; then $p=2^q$ for some $q \in \{0,\ldots,k-1\}$.
    Now, for every $d \in \N$, the subsequence
    $S[d \cdot 2p:(d+1) \cdot 2p]$ of $S$ (exclusive of the final point)
    is symmetric around the pair $(p-1,p)$ and will vanish
    if identical adjacent terms are cancelled exhaustively. 
    Applying this to all pairs except the unique pair involving $y_i$
    shows that cancellations inside the $X_i$-part of $S$
    will leave a copy of $y_i$ as the unique uncancelled term.
    Since every tuple in the domain of $u_k^{k-1}$ comes from
    some such pattern $P_i$, the claim follows. 
  \end{proof}
  
  Thus, for every $i \in [r]$ we have that $f_{2^k-1}(t_1[i],\ldots,t_{2^k-1}[i])=t[i]$, 
  so $t \in R$ and $u_k^{k-1} \in \pPol(R)$.
\end{proof}

\subsection{Connections to Hypergraph Tur{\'a}n} \label{sec:hypergraph-turan}

Given a set of patterns $Q$ we are now interested in studying the non-redundancy of families of promise relations invariant under $Q$. 
We study this systematically in this section and exemplify our main result in Section~\ref{subsec:turan_examples}
\label{subsec:turan}

   Before we can state the main result we need to extend the algebraic vocabulary so that it is applicable to CSP instances and not only to relations. First, given $R \subseteq S \subseteq D^r$ and an instance $(X,Y)$ of $\CSP(R \mid S)$ recall that we (by Lemma~\ref{lemma:reduction:to:multipartite}) without loss of generality may assume that it is $r$-partite and that each $i \in [r]$ is associated with exactly one sort $X_i$.
   Given a multisorted pattern $P = (P_1, \ldots, P_r)$ of arity $n$ each $P_i$ then induces a set of partial functions $\interpretation{X_i}{P_i}$, and given $p_i \in \interpretation{X_i}{P_i}$ for every $1 \leq i \leq r$ we can then apply $p_1, \ldots, p_r$ to $y^1, \ldots, y^n$ as $(p_1, \ldots, p_r)(y^1, \ldots, y^n)$ which for each $1 \leq i \leq r$ yields $p_i(y^1_i, \ldots, y^n_i)$.

 \begin{lemma} \label{lem:cpol-redundancy}
  Let $R \subseteq S \subseteq D_1 \times \ldots \times D_r$ be multisorted relations with $r$ sorts $D_1 \cup \ldots \cup D_r = D$ and let $P=(P_1,\ldots, P_r)$ be 
  an $n$-ary multisorted pattern such that $(R,\widetilde{S}) \in \inv(P)$. 
  Let $(X,Y)$ be an $r$-partite instance of $\CSP(R \mid S)$ and let 
  $y^{(1)},\ldots, y^{(n)}, y \in Y$ such that $(p_1, \ldots, p_r)(y^{(1)},\ldots,y^{(n)})=y$ is defined
  and $y \notin \{y^{(1)},\ldots,y^{(n)}\}$ for some $p_1 \in \interpretation{X_1}{P_1}, \ldots, p_r \in \interpretation{X_r}{P_r}$. Then $y$ is redundant in $(X,Y)$. 
\end{lemma}

\begin{proof}
  Assume that there is an assignment $\sigma \colon X \to D$ such that $\sigma(y) \in S \setminus R$
  and $\sigma(y') \in R$ for every $y' \in Y - \{y\}$. 
  It follows that there exists $q_1 \in \interpretation{D_1}{P_1}, \ldots, q_r \in \interpretation{D_r}{P_r}$
    such that
  $\sigma(y) = (q_1, \ldots, q_r)(\sigma(y^{(1)}), \ldots, \sigma(y^{(n)}))$:
  for every position $i \in [r]$, we have $y_i = p_i(y^{(1)}_i,\ldots,y^{(n)}_i)$,
  hence there exists a pair $((x_1,\ldots,x_n),x_{j_i}) \in P_i$, so that
  $y_i=y^{{(j_i)}}_i$ and:
  \begin{enumerate}
      \item 
      if $x_{j_i} \in \{x_1, \ldots, x_n\}$ then $q_i(\sigma(y^{(1)})_i,\ldots,\sigma(y^{(n)})_i)=\sigma(y^{(j_i)})_i$ for every possible $q_i \in \interpretation{D_i}{P_i}$, and 
      \item 
      if $x_{j_i} \notin \{x_1, \ldots, x_n\}$ then $\interpretation{D_i}{P_i}$ for every $d \in |D_i|$ contains a partial function $q^d$ such that $q^d(\sigma(x_1), \ldots, \sigma(x_n)) = d$, and we choose $d = \sigma(y_i)$.
  \end{enumerate}
  But now $\sigma(y^{(1)}), \ldots, \sigma(y^{(n)}) \in R$ and $\sigma(y) \in S$,
  which implies that $\sigma(y) \in R$, contrary to our assumption. 
\end{proof}   

Lemma~\ref{lem:cpol-redundancy} requires a fair amount of notation to state but as the following example shows is actually a reasonable and straightforward property on non-redundant constraints.

\begin{example}
We recall the relations $R=\text{3LIN}_{\Z/3\Z}=\{(x,y,z) \in \{0,1,2\}^3 \mid x+y+z \equiv 0 \pmod 3\}$
and $R^*=\text{3LIN}^*_{\Z/3\Z}= R \setminus \{(0,0,0)\}$ from Brakensiek \& Guruswami~\cite{brakensiek2024Redundancy}, introduced in Section~\ref{sec:cube-power}. 
Consider the multisorted polymorphism pattern $P = (P_1, P_2, P_3)$ where (1) $P_1 = P_2 = \{((x), x)\}$ is the pattern satisfied only by the unary projection and (2) $P_3 = \{((x), y)\}$ is an unrestricted pattern, i.e., one satisfied by any unary function. 
We claim that $(R^*, \widetilde{R}) \in \inv(P)$. 
To see this, note first that $\interpretation{\{0,1,2\}}{P_1} = \interpretation{\{0,1,2\}}{P_2} = \{\pi^1_1\}$, while $\interpretation{\{0,1,2\}}{P_3} = \{0,1,2\}^{\{0,1,2\}}$. 
Thus, the preservation condition boils to the following: for any tuple $t \in R^*$ we must have that $(t_1, t_2, d) \in \widetilde{R}$ for any $d \in \{0,1,2\}$, i.e., that $(t_1, t_2, d)$ is not included in $R \setminus R^* = \{(0,0,0)\}$. 
This is clearly true since if at least only entry in $t$ is non-zero and the sum is 0, $t_1$ and $t_2$ cannot both be 0, and the claim follows.

Now, consider an instance $(\{x,y,z,v\}, \{(x,y,z), (x,y,v)\})$ of $\CSP(R^* \mid R)$. We introduce disjoint copies $X_1, X_2, X_3$ of $X$ and then use functions
\begin{enumerate}
\item
$\pi^1_1 \in \interpretation{X_1}{P_1}$,
\item
$\pi^1_1 \in \interpretation{X_2}{P_2}$,
\item
$f \in \interpretation{X_3}{P_3}$ where $f(d) = v$ for every $d \in X_3$.
\end{enumerate}
It then follows that $(p_1, p_2, p_3)((x,y,z)) = (x,y,v) \in Y$ and that $(x,y,v)$ is redundant in $\CSP(R^* \mid R)$.
\end{example}

To connect to the hypergraph structure, we need the following assumptions.
Let $P=(P_1,\ldots,P_r)$ be an $n$-ary multisorted polymorphism pattern.
Say that $P$ is \emph{unit} if $|P_i|=1$ for every $i \in [r]$, i.e., each $P_i$ consists of a single pattern $((x_{i,1},\ldots,x_{i,n}),y_i)$. For a set of patterns $Q$ with $r$ sorts we let $\inv^{r}(Q) \subseteq \inv(Q)$ be the set of promise relations of arity $r$ (with exactly $r$ sorts).

\begin{proposition} \label{prop:unit_pattern_suffices}
  Let $P = (P_1, \ldots, P_r)$ be an $n$-ary multisorted polymorphism pattern. Then there exists a set of unit patterns $U \subseteq \minor{P}$ such that $\inv^r(U) = \inv^r(P)$.
\end{proposition}

\begin{proof}
  First, observe that $U = \{(U_1, \ldots, U_r) \mid 1 \leq i \leq r, U_i \subseteq P_i, |U_i| = 1\} \subseteq \minor{P}$ since we allow subsets in the minor definition. Due to Propositition~\ref{prop:minors_preserve} this implies that $\inv(P) \subseteq \inv(U)$ and hence also that $\inv^r(P) \subseteq \inv^r(U)$. For the other direction, assume there exists $(S,T) \in \inv^r(U)$ where $(S,T) \notin \inv^r_D(P)$ for any finite set $D = D_1 \cup \ldots \cup D_r$, and let $p_1 \in \interpretation{D_1}{P_1}, \ldots, p_r \in \interpretation{D_r}{P_r}$ witness this, i.e., $(p_1, \ldots p_r)(s^{(1)}, \ldots, s^{(n)}) = t \notin T$ for $s^{(1)}, \ldots, s^{(n)} \in S$. Each such application of a $p_i$ is witnessed by a unit pattern, and it follows that there exist subfunctions $u_1 \in \interpretation{D_1}{U_1}, \ldots, u_r \in \interpretation{D_r}{U_r}$ of $p_1, \ldots, p_r$ such that $(u_1, \ldots u_r)(s^{(1)}, \ldots, s^{(n)}) = t \notin T$, contradicting the original assumption.
\end{proof}

Let $P=(P_1,\ldots,P_r)$ be a unit multisorted pattern
with patterns $P_i=\{((x_{i,1},\ldots,x_{i,n}),y_i)\}$ (where the 
$x_{i,j}$ are not necessarily distinct). 
Say that $P$ is a \emph{partial projection} if there is $j \in [n]$ such that for every
$i \in [r]$, $y_i=x_{i,j}$, and \emph{non-trivial} otherwise. 

We define a hypergraph $\cH(P)$ from $P$ as follows.
For $i \in [r]$, let $X_i$ be the variable set that $P_i$ is defined
over, and assume (syntactically) that $X_i \cap X_j = \emptyset$
for $i \neq j$. Then $\cH(P)$ is the hypergraph on 
vertex set $X=X_1 \cup \ldots \cup X_r$ and with $n+1$
hyperedges given by $Y=\{(x_{1,j},\ldots,x_{r,j}) \mid j \in [n]\} \cup \{(y_1,\ldots,y_r)\}$.
Refer to $(y_1,\ldots,y_r)$ as the \emph{output edge} of $\cH(P)$.

We can now observe the following. 

\begin{lemma} \label{lem:unit-pattern-hypergraph}
  Let $R \subseteq S \subseteq D^r$ be a pair of multisorted relations and let
  $P$ be a unit multisorted pattern that is not a partial projection, such that $P \in \mpattern(R, \widetilde{S})$. 
  Let $(X,Y)$ be a non-redundant instance of $\CSP(R \mid S)$.
  Then $(X,Y)$ is $\cH(P)$-free (when interpreted as a hypergraph).
\end{lemma}

\begin{proof}
    Let $\{y^{(1)},\ldots,y^{(n)},y\} \subseteq Y$ form a hypergraph isomorphic to $\cH(P)$
    such that each $y^{(i)}$ maps to the $i$th edge of $\cH(P)$ in the order of construction 
    and where $y$ maps to the output edge. Then, by construction, $(p_1, \ldots, p_r)(y^{(1)},\ldots, y^{(n)})$ is defined for all $p_1 \in \interpretation{X_1}{P_1}, \ldots, p_r \in \interpretation{X_r}{P_r}$, and it is clear that there exists such functions which yield $y$. Since $P$ is not a partial projection it follows that $y \notin \{y^{(1)},\ldots,y^{(n)}\}$ and that 
    $(X,Y)$ is redundant by Lemma~\ref{lem:cpol-redundancy}.    
\end{proof}

Recall by Proposition~\ref{prop:unit_pattern_suffices} that for $r$-ary relations a general multisorted pattern
$P=(P_1,\ldots,P_r)$ is equivalent to the collection of unit patterns
produced by selecting one pair from $P_i$ for every $i \in [r]$. Hence, working with only unit patterns is not restrictive, and
Lemma~\ref{lem:unit-pattern-hypergraph} then implies that every multisorted
polymorphism pattern implies a finite collection of excluded multipartite hypergraphs.

To show the result in the other direction, consider the following. 
Let an instance $(X,Y)$ of $\CSP(R \mid S)$ be \emph{minimal redundant}
if $(X,Y)$ is redundant but for every $y \in Y$,
$(X,Y-\{y\})$ is a non-redundant instance of $\CSP(R \mid S)$. 
Note that this is a well-behaved notion since non-redundant instances $(X,Y)$ are closed under taking subsets of $Y$. We now show the following.

\begin{lemma}\label{lem:H-to-P}
  Let $(X,Y)$ be a minimal redundant instance of $\CSP(R \mid S)$
  for some $R \subseteq S \subseteq D^r$.
  Let $y \in Y$ be such that there is no assignment $\psi$
  that satisfies every $y' \in Y-\{y\}$ (i.e., $\psi(y') \in R$) but has $\psi(y) \in S \setminus R$. 
  Then there is a unit multisorted polymorphism pattern $P$ and a
  permutation $y_1, \ldots, y_m$ of $Y-\{y\}$ such that $P \in \mpattern(R,\widetilde{S})$ and $p_1 \in \interpretation{X_1}{P_1}, \ldots, p_r \in \interpretation{X_r}{P_r}$ such that
  $(p_1, \ldots, p_r)(y_1,\ldots,y_m)=y$ and $\cH(P)=(X,Y)$.
  \end{lemma}

\begin{proof}
  Let $Y=\{y_1,\ldots,y_m,y\}$ and let $y$ be such that there is no
  assignment $\psi$ such that $\psi(y') \in R$ for every $y' \in Y-\{y\}$
  and $\psi(y) \in S \setminus R$. This exists, as otherwise $(X,Y)$
  is non-redundant. Define a multisorted pattern
  $P=(P_1,\ldots,P_r)$ on variable set $X$ such that for $i \in [r]$
  we have $P_i=\{((y_1(i), \ldots, y_m(i)),y(i)\}$. Observe that $P$ is a unit and non-trivial.  Next, we show that $P \in \mpattern(R,\widetilde{S})$.

  Assume for the sake of contradiction that $P \not\in \mpattern(R,\widetilde{S})$. Then, there exist tuples $t_1, \ldots, t_m \in R$ and $p_1 \in \interpretation{X_1}{P_1}, \ldots, p_r \in \interpretation{X_r}{P_r}$
  such that $t = (p_1, \ldots, p_r)(t_1,\ldots,t_m)$ is defined and $t \in S \setminus R$.
  Define $\psi \colon X \to D$ by letting $\psi(y_i(j))=t_i(j)$ for
  every $i \in [m]$, $j \in [r]$, and likewise $\psi(y(j))=t(j)$.
  Note that this is consistently defined by construction of $P$.
  Then $\psi(y_i) \in R$ for every $i \in [m]$, and $\psi(y) \in S \setminus R$, 
  which is contrary to our assumptions. Thus $P \in \mpattern(R,\widetilde{S})$.
  The remaining statements hold by construction. 
\end{proof}
To finish our characterization we first define the following tensor product  which allows us to combine two promise relations $(S_1, T_1)$  and $(S_2, T_2)$ into a promise relation of the same arity but over a larger domain. First, given $R$ and $S$ we let
\[R \otimes S := (\{((a_1,b_1),\hdots, (a_r,b_r)) : a \in R, b \in S\})\]
and then
\[(S_1, T_1) \otimes (S_2, T_2) := (S_1 \otimes S_2, T_1 \otimes T_2).\]

\begin{proposition}
    Let $(S_1, T_1)$ and $(S_2, T_2)$ be two $r$-ary multisorted relations with $r$ sorts. Then $(S_1, T_1) \otimes (S_2, T_2)$ is fgppp-equivalent to $\{(S_1, T_1), (S_2, T_2)\}$.
\end{proposition}

\begin{proof}
  Assume that $S_1 \subseteq T_1 \subseteq D_1 \times \ldots \times D_r$ and that $S_2 \subseteq T_2 \subseteq E_1 \times \ldots \times E_r$. For $1 \leq i \leq r$ let $g^i_1 \colon D_i \times E_i \to D_i$ be the first projection and $g^i_2 \colon D_i \times E_i \to E_i$ the second projection. Then, 
    \[
    (S_1 \otimes S_2)(x_1, \hdots, x_r) \equiv S_1(g^1_1(x_1), \hdots, g^r_1(x_r)) \wedge S_2(g^1_2(x_1), \hdots, g^r_2(x_r)),
  \]
  and
      \[
    (T_1 \otimes T_2)(x_1, \hdots, x_r) \equiv T_1(g^1_1(x_1), \hdots, g^r_1(x_r)) \wedge T_2(g^1_2(x_1), \hdots, g^r_2(x_r)).
  \]
    Thus, $\{(S_1, T_1), (S_2, T_2)\}$ fgppp-defines $(S_1, T_1) \otimes (S_2, T_2)$. For the other direction, fix $s \in S_2 \subseteq T_2$, and for $i \in [r]$, let $h_i : D_i \to D_i \times E_i$ be the map $h_i(x) = (x, s_i)$. Then,
    \[
    S_1(x_1, \hdots, x_r) \equiv (S_1 \otimes S_2)(h_1(x_1), \hdots, h_r(x_r)),
  \]
  and
      \[
    T_1(x_1, \hdots, x_r) \equiv (T_1 \otimes T_2)(h_1(x_1), \hdots, h_r(x_r)).
  \]
    By an analogous formula, $(S_1, T_1) \otimes (S_2, T_2)$ fgppp-defines $(S_2, T_2)$ as well.
\end{proof}

Via Corollary~\ref{cor:pattern_galois} we obtain the following general description.

\begin{corollary} \label{cor:tensor-pol} 
Let $(S_1, T_1)$ and $(S_2, T_2)$ be two multisorted promise relations with $r$ sorts. Then 
\begin{enumerate}
 \item $(S_1, T_1) \otimes (S_2, T_2)$ is fgppp-equivalent to $\{(S_1, T_1), (S_2, T_2)\}$, and
 \item $\mpattern(\{(S_1, T_1), (S_2, T_2)\}) = \mpattern((S_1, T_1) \otimes (S_2, T_2))$.
 \end{enumerate}
\end{corollary}

Given a set of multisorted patterns $Q$ (all with exactly $r$ sorts), we define $\inv^{= r}(Q)$ to be the set of multisorted pairs $(R,S)$ of arity $r$ over any domain with exactly $r$ sorts, one for each argument, such that $(R,S) \in \inv(Q)$. We then define \[\NRD(\inv^{= r}(Q), n) := \max_{(R,S) \in \inv^{= r}(Q)}\NRD(R \mid S, n).\] We now prove  that $\NRD(\inv^{= r}(Q, n)$ is precisely captured by a hypergraph Tur{\'a}n problem, generalizing observations of \cite{bessiere2020Chain,carbonnel2022Redundancy}. 

\begin{theorem}\label{lem:inv-turan}
  Let $Q$ be a finite set of multisorted polymorphism patterns with exactly $r$ sorts. Then one can construct a set $\cH(Q)$ of $r$-partite hypergraphs with the following property: 
  for all $n \in N$,
  \begin{align}
    \NRD(\inv^{= r}(Q), n) = \ex_r(n, \cH(Q)),\label{eq:nrd-turan}
  \end{align}
  where $\ex_r(n, \cH)$ is the size of the largest $r$-partite $r$-uniform hypergraph on $n$ vertices excluding $\cH$.
\end{theorem}

\begin{proof}
Let $[Q]_{\surj}$ be the closure of $Q$ with respect to \emph{surjective} minors: in other words, if $P \in Q$ of arity $m$, then $P'$ of arity $n$ is in $[Q]_{\surj}$ if there exists a surjective function $h \colon [m] \to [n]$ (so $n \le m$) such that for $i \in [r]$ and every $((x_1,\hdots, x_n), x) \in P'_i$ we have that $((x_{h(1)}, \hdots, x_{h(m)}), x) \in P_i$.

Recall that if $P$ is a unit polymorphism pattern, then there is a corresponding hypergraph $\cH(P)$. Slightly overloading the notation we then define
\[
    \cH(Q) := \{\cH(P) : P \in [Q]_{\surj}, P\text{ unit}\}.
\]

We prove (\ref{eq:nrd-turan}) as a pair of inequalities. Since $\inv^{= r}(Q) = \inv^{= r}([Q]_{\surj})$ (this is a special case of Proposition~\ref{prop:minors_preserve}), we have that the $\le$ direction of (\ref{eq:nrd-turan}) is an immediate consequence of Lemma~\ref{lem:unit-pattern-hypergraph}.

Let $(X,Y)$ be the largest $r$-partite hypergraph on $n$ vertices which is $\cH(Q)$-free. For every $Y' \subseteq Y$ and $y' \in Y'$, let $P(Y',y')$ be the $r$-sorted unit polymorphism pattern of arity $\ar(P(Y',Y')= |Y'|-1$ which maps $Y' \setminus \{y'\}$ to $y'$.

The key observation is that $P(Y', y') \not\in [Q]_{\min}$.
If not, there exists $P' \in Q$ and $P'' \subseteq P'$ unit such that $P(Y', y') = (P'')_{/h}$ for some \emph{non}-surjective $h \colon [\ar(P'')] \to [\ar(P(Y',y'))]$. Note then there must exist $Y'' \subseteq Y'$ and a surjective map $h' \colon [\ar(P'')] \to [\ar(P(Y'',y')]$ such that $h$ and $h'$ are equivalent in the sense that $P(Y'', y') = (P'')_{/h'}$.
In that case, $P(Y'', y') \in [Q]_{\min}$, which contradicts the fact that $(X,Y)$ and thus $(X,Y'')$ is $\cH(Q)$-free.

Thus, $P(Y', y') \not\in [Q]_{\min}$, so $\inv^{= r}(Q) \supsetneq \inv^{= r}(Q \cup \{P(Y',y')\})$. Thus, it is possible to find $(R(Y',y'), S(Y',y')) \in \inv^{= r}(Q)$ such that $P(Y',y') \not\in \mpattern(R(Y',y'), S(Y',y'))$. Now define
\[
    (R, S) := \bigotimes_{Y' \subseteq Y, y' \in Y'} (R(Y',y'), S(Y',y')).
\]
By Corollary~\ref{cor:tensor-pol}, we know that $Q \subseteq \mpattern(R, S)$, but $P(Y',y') \not\in \mpattern(R, S)$ for all $Y'$ and $y'$. To finish, it suffices to prove that $(X,Y)$ is a conditionally non-redundant instance of $\CSP(R \mid S)$.
If not, there exists a minimally redundant subinstance $(X,Y')$. By Lemma~\ref{lem:H-to-P}, we then have that $P(Y', y') \in \mpattern(R,\widetilde{S}),$ a contradiction. Thus, $(X,Y)$ is non-redundant.
\end{proof}

\subsection{Hypergraph Turán Examples and Applications}
\label{subsec:turan_examples}

Let us now consider some prominent examples of this connection.

\begin{example}
  One of the most basic polymorphism patterns is the $k$-NU 
  pattern \[P=\{((y,x,\ldots,x),x), ((x,y,\ldots,x),x), \ldots, ((x,x,\ldots,y),x)\}\] (introduced in Example~\ref{ex:pattern_power}).
  If we turn this into a $k$-partite, $k$-uniform hypergraph by placing
  each pattern in a different vertex class, then we get the hypergraph
  $\cH$ with edges \[\{(x, x, \ldots, x), (x, x, \ldots, y), \ldots,
  (y, x, \ldots, x)\}.\] It is easy to see that a $k$-partite, $k$-uniform
  hypergraph $(X,Y)$ is $\cH$-free if and only if the following holds:
  for each edge $E \in Y$ there is a strict subedge $E' \subset E$
  such that $E'$ only occurs in $E$. Therefore, $|Y| = \Oh(|X|^{k-1})$
  and $\ex_r(n, \cH)=\Theta(n^{k-1})$. In this case, the bound is tight,
  in that $\OR_{k-1}$ (e.g., with a fictitious added argument) is
  preserved by the $k$-NU pattern. 
\end{example}

\begin{example}
  Another class of tight examples comes from Section~\ref{sec:nrd-bin}.
  First consider the promise relation $(C_{2k}^*, C_{2k})$ for $k \geq 2$.
  In this case, the conditional polymorphisms are the patterns
  with $P_1=\{((0,1,1,2,2,\ldots, t-1, t-1),0)\}$
  and $P_2=\{((1,0,2,1,\ldots,t-2,t-1), 0)\}$ for $t<k$, 
  and the corresponding forbidden hypergraphs (which are graphs, since $r=2$)
  are the even cycles of length less than $2k$.
  By the results in Section~\ref{subsec:binary-class}, we have
  $\NRD(C_{2k}^* \mid C_{2k}, n) = \Theta(\ex_2(n, \{C_2,C_4,\ldots,C_{2k-2}\}))$.  
  Similarly, for the relation $R_{2k}$ in Section~\ref{subsec:inf-ternary},
  we get a multisorted pattern with $P_1=\{((x,\ldots,x),x)\}$
  and where $P_2$ and $P_3$ enumerate a cycle $C_{2t}$, $t<k$ as above.
  Let $\cH(Q)$ be the corresponding set of hypergraphs. Then a
  3-uniform, 3-partite hypergraph $(X,Y)$ is $\cH(Q)$-free if and only if
  the neighbourhood of every vertex in part 1 forms a graph with girth
  at least $2k$, and again, the results of Section~\ref{subsec:inf-ternary}
  imply that $\NRD(R_{2k}, N)= \Theta(\ex_3(n, \cH(Q)))$. 
\end{example}

\begin{example}
  For a complementary example where a similar tight result does not hold (for finite $\cH(Q)$),
  consider the case of binary relations preserved by the partial Mal'tsev operation $u_2$. 
  That is, let $P=\{((x,x,y),y), ((x,y,y),x)\}$ and let  $R \subseteq D^2$ with $P \in \pattern(R)$. 
  The only non-trivial multisorted unit pattern is $P'=(P_1,P_2)$ where 
  $P_1=\{((x,x,y),y)\}$ and $P_2=\{((x,y,y),x)\}$, and 
  the corresponding forbidden graph is $\cH(P')=C_4$.
  Then $\ex_2(n, \cH(P))=\Theta(n^{1.5})$. 
  However, if $R$ is preserved by $u_2$ as a (non-conditional) partial polymorphism,
  then $\NRD(R, n)=\Oh(n)$, as shown by previous work~\cite{butti2020,bessiere2020Chain}.
  This can be viewed as follows: For any relation $R \subseteq D^2$, view $R$ as a bipartite graph $G$.
  Applying $u_2$ exhaustively to $R$ (i.e., adding $u_2(t_1,t_2,t_3)$ to $R$ 
  whenever this is defined and not already present in $R$)
  yields a rectangular relation, where every connected component of $G$ is 
  completed into a biclique. The same argument applies to any instance
  $(X,Y)$ of $\CSP(R)$. Thus a non-redundant instance $(X,Y)$ of $\CSP(R)$ 
  must be acyclic. This ``closure'' aspect of repeated application of $u_2$
  is missed by the Turán connection.  One could include further patterns that capture 
  the consequences of such propagation, which would then correspond to adding the graphs
  $C_6$, $C_8$, \ldots to the collection of forbidden subgraphs, but for any finite
  collection of such subgraphs we get a Turán exponent of $\ex_2(n,\cH(Q))=\Theta(n^{1+\varepsilon})$
  for some $\varepsilon>0$ as above. Thus no finite collection of forbidden hypergraphs
  captures the consequences of the partial polymorphism $u_2$. 
\end{example}

\begin{example}\label{ex:r-LIN}
  For a much more drastic example, let $r \geq 3$ and let $R$ be an
  $r$-ary linear equation, e.g., $R \subset \{0,1\}^r$ defined as
  $R(x_1,\ldots,x_r) \equiv (x_1 \oplus \ldots \oplus x_r=0)$. 
  Then a hypergraph $\cH$ is a forbidden configuration for a
  non-redundant instance $(X,Y)$ of $\CSP(R)$ if and only if 
  the edges of $\cH$ are linearly dependent (as a set of vectors over GF$(2)$).
  Unfortunately, large systems of $r$-ary equations are possible where no constant-sized 
  subsystem is linearly dependent. For example, let $\varepsilon>0$
  and consider a random instance $(X,Y)$ of $\CSP(R)$ with
  $|X|=n$ and edges chosen uniformly at random so that $E[|Y|]=n^{r/2-\varepsilon}$.
  Then it can be shown that for any $t \in \N$, with high probability
  (asymptotically in $n$) every set $F \subseteq Y$ with $|F| \leq t$
  has $|\bigcup F| > r|F|/2$, and consequently has a variable that occurs only once in $F$.
  Since every linearly dependent set contains a minimal linearly dependent set,
  and no such set can contain a variable that occurs only once,
  it follows that every set of $\leq t$ edges from $Y$ is independent.
  Thus, any finite collection $\cH(Q)$ of forbidden hypergraphs for $\CSP(R)$
  has $\ex_r(n,\cH(Q)) = \Omega(n^{r/2-o(1)})$ (at least), despite $\NRD(R,n)=\Oh(n)$. Note that $n^{r/2}$ is also the threshold at which constant degrees of the Sum-of-Squares hierarchy fails to refute $r$-ary linear equations (e.g., \cite{kothari2017Sum}). Similar insights can also be applied to the \emph{range avoidance} problem in computational complexity (e.g., \cite{kuntewar2025Range}).
\end{example}

Given these examples, one of the more interesting questions
corresponds to the $k$-cube patterns $U_k$. For these patterns,
the corresponding hypergraph $\cH(U_k)$ is the complete $k$-partite $k$-uniform
hypergraph $K_{2,2,\ldots,2}$ with two vertices per part, and the value
of $\ex(n, \cH(U_k))$ is a notorious open question in extremal combinatorics,
even for $k=3$, known as the \emph{Erd\H{o}s box problem}~\cite{erdos1964extremal,katz2002Remarks,furedi2013history,conlon2021Random,gordeev2024Combinatorial}.
This connection was first noticed by Carbonnel~\cite{carbonnel2022Redundancy}, who used it
to show that every $r$-ary relation $R$ has either $\NRD(R,n)=\Theta(n^r)$ or $\NRD(R,n)=\Oh(n^{r-\varepsilon_r})$, where $\varepsilon_r \geq 2^{1-r}$ comes from $\ex(\cH(U_r),n)=\ex(K_{2,\ldots,2},n)=\Oh(n^{r-\varepsilon_r})$~\cite{erdos1964extremal}, and is based on the assumption that $R$ is preserved by the cube pattern $U_r$.
Focusing on $r=3$, our results in Section~\ref{subsec:turan} imply that there is an infinite family of ternary promise relations whose non-redundancy matches $\ex(K_{2,2,2},n)$,
but what about a single such promise relation? At the moment, our largest known bound for $r=3$ below  $\NRD(\OR_3,n)=\Theta(n^3)$ is $\NRD(R,n)=\Theta(n^{2.5})$ from Section~\ref{subsec:inf-ternary}. 

Based on these connections it is natural to investigate if stronger $\Oh(n^{r - \eps})$ bounds can be obtained for stronger patterns than $U_r$.
Promising candidates are \emph{$k$-edge} operations, that are known to generalize $k$-NU
and Mal'tsev operations~\cite{berman2009Varieties}. It is a $(k+1)$-ary operation that satisfies
the following identities ($k \geq 2$) for all $x, y$. 
\begin{align*}
  p(x,x,y,y,y,\ldots,y) &= y \\
  p(x,y,x,y,y,\ldots,y) &= y \\
  p(y,y,y,x,y,\ldots,y) &= y \\
  p(y,y,y,y,x,\ldots,y) &= y \\
  \cdots \\
  p(y,y,y,y,y,\ldots,x) &= y
\end{align*}
In particular, a 2-edge operation is a Mal'tsev operation. 
For example, the forbidden hypergraph for the 3-edge polymorphism pattern has edge set $\{(x,x,x), (,
x,x,y), (x,y,x), (y,x,x), (y,y,x)\}$, i.e., the same as 3-NU plus the edge $(y,y,x)$.
A language has a (total) $k$-edge polymorphism if and only if it has
\emph{few subpowers}, i.e., only $2^{n^{O(1)}}$ distinct $n$-ary relations
can be pp-defined in the language; see Berman et al.~\cite{berman2009Varieties}
or the survey by Barto, Krokhin and Willard~\cite{barto2017Polymorphisms}.
Lagerkvist and Wahlström~\cite{lagerkvist2020Sparsification} showed that
a language with a \emph{$k$-edge embedding}, generalizing Mal'tsev embeddings,
has non-redundancy $O(n^{k-1})$.
Here, we show that a partial $k$-edge polymorphism (i.e., the partial polymorphism patterns
defined from the above identities) also implies a non-trivial non-redundancy result.

\begin{lemma}
  Let $R \subseteq S \subseteq D^r$ be a pair of relations such that
  $(R, \widetilde{S})$ is preserved by the $k$-edge polymorphism pattern,
  where $r \geq k$. Then $\NRD(R \mid S,n)=O(n^{r-\lfloor r/k \rfloor/2})$. 
\end{lemma}

\begin{proof}
  We first show the result for the case $r=k$.
  Let $I=(X,Y)$ be a $k$-partite instance of $\CSP(R \mid S)$ and consider a constraint $(v_1,\ldots,v_k) \in Y$. Via Lemma~\ref{lem:unit-pattern-hypergraph}, we find the following. 
  Select $k-2$ positions, w.l.o.g.~$1, \ldots, k-2$, and consider the
  bipartite graph $G$ where $uw \in E(G)$ if and only if there is a constraint o
  $(v_1,\ldots,v_{k-2},u,w)$ in $I$.
  Then the hypergraph representation of the $k$-edge pattern
  for $k$-ary relations allows us to discard $R$ if two things
  apply: (1) $v_{k-1}v_k$ participates in a $C_4$ in $G$ and (2)
  for every $i \in [k-2]$ there is a constraint $(v_1,\ldots,v_i',\ldots,v_k)$
  in $I$ where $v_i' \neq v_i$. To argue the density we proceed as follows.
  Given $I=(X,Y)$, we first mark one constraint from $Y$ for
  every $(k-1)$-ary projection; this makes for $O(n^{k-1})$ constraints.
  Put these constraints aside, then form the bipartite graph $G$ as
  above using only the remaining constraints $Y'$. Let $C=(v_1,\ldots,v_k)$
  be any constraint in $Y'$. Then the second condition is automatically
  satisfied: since $C \in Y'$ there is for every $i \in [k]$ 
  a constraint that we already set aside that only differs from $C$ in
  coordinate $i$. Therefore it suffices to find $k-2$ variables in $C$,
  say $v_1, \ldots, v_{k-2}$ (again w.l.o.g.\ in terms of the identity of the parts $1$, \ldots, $k-2$)
  such that their neighborhood induces a graph $G$ with a $C_4$. 
  If no such selection exists, then for every selection $v_1, \ldots, v_{k-2}$
  the graph $G$ has $O(n^{1.5})$ edges, and $|Y|=O(|X|^{k-1/2})$.

  For the case $r>k$, we first partition the arguments of the constraint
  into $k$ approximately equal-sized groups, with $\lfloor r/k \rfloor$ or $\lceil r/k \rceil$
  arguments per group, then we reinterpret $(R,S)$ as a $k$-ary
  promise constraint $(R',S')$ over the product domain $D'=D^{\lceil r/k \rceil}$. 
  It is easy to check that this operation preserves polymorphism patterns
  (cf.~\cite{lagerkvist2021Coarse} and Example 11 of~\cite{carbonnel2022Redundancy} for the non-promise case).
  Furthermore, performing the corresponding product operation on $I=(X,Y)$,
  we get a $k$-partite instance with $O(n^{\lceil r/k \rceil})$ variables
  and $|Y|$ constraints, and this operation clearly does not decrease non-redundancy.
  Following the analysis above, if $X=X_1 \cup \ldots \cup X_k$, the non-redundancy is upper-bounded (up to ordering of arguments) by $O(\prod_{i=1}^{k-2} |X_i| \cdot (|X_{k-1}| + |X_k|)^{1.5})$, which is less than $|X|^k$ by a factor of $\min_i |X_i|^{0.5}$. 
  \end{proof}
  
\section{Conclusion}\label{sec:concl}

In this paper, we systematically studied the non-redundancy of arbitrary finite-domain (promise) constraints. Earlier techniques for bounding non-redundancy have typically yielded bounds of the form $\Theta(n^k)$ for integer exponents $k$, but as we have proven in this paper such bounds are the exception rather than the norm, and the full non-redundancy landscape is far richer and allows all possible rational exponents. To systematically study these questions we developed a novel algebraic approach for conditional non-redundancy and proved strong connections to hypergraph Tur{\'a}n problems. We took a special interest in near linear non-redundancy and solved an open question in the literature~\cite{chen2020BestCase}, namely, whether the set of Boolean predicates with Mal'tsev embeddings coincides with the set of Boolean predicates which are balanced. To prove this, we introduced the Catalan terms which provide a novel and useful reinterpretation of a Mal'tsev embedding that more generally also allowed us to prove that a Mal'tsev embedding always implies a group embedding. For non-Boolean domains we showed that Abelian and Mal'tsev embeddings differ and gave the first example of a non-Abelian embedding via the Pauli group. Currently, we do not know whether Mal'tsev embeddings capture all predicates with linear non-redundancy, but we now know that if this is the case then it is barely true. In particular, we constructed an explicit family of ternary predicates provably lacking Mal'tsev embeddings whose non-redundancy approaches linear.
We now conclude with a few exciting directions for future investigation.

\paragraph{Mal'tsev embeddings and linear redundancy.}

The central open question in the field of non-redundancy currently is a classification of which predicates have (near) linear non-redundancy. One potential next step toward this question is to unify concepts of infinite and finite Mal'tsev embeddings.

\begin{conjecture}\label{conj:inf-to-fin}
Any predicate over a finite domain with an infinite Mal'tsev embedding also has a finite Mal'tsev embedding. In particular, any predicate with an infinite Mal'tsev embedding has linear non-redundancy.
\end{conjecture}

In particular, we hope that the techniques involving Catalan terms in Section~\ref{sec:linear-nrd} could lead to solid progress toward this conjecture. Conversely, one can hope that (in)finite Mal'tsev embeddings capture all predicates with linear non-redundancy. However, to establish such a fact, one needs to first be able to analyze predicates such as $\BCK := \{111,222,012,120,201\}$ from \cite{bessiere2020Chain} and independently discovered in a ``Booleanized'' form by many others \cite{lagerkvist2020Sparsification,chen2020BestCase,khanna2024Characterizations}, and the family of $\CYCs_m$ predicates constructed in Section~\ref{subsec:approach-linear}. See Section 7.1 of \cite{brakensiek2024Redundancy} for a more thorough discussion of such predicates. Note that $u^{k-1}_k \in \pPol(\BCK)$ for all $k\ge 2$, so any non-redundancy lower bound for $\BCK$ will necessarily require techniques beyond those studied in this paper.

We also leave as an open question whether there is a polynomial-time algorithm (where the input is the truth table of the predicate) for determining whether a predicate satisfies the Catalan identities. We note that tools such as SMT solvers like Z3~\cite{demoura2008Z3} can assist in identifying partial polymorphisms, but the worst-case behavior of such techniques is likely exponential time.

\paragraph{Non-redundancy of Boolean languages.}
Is there a Boolean relation with fractional NRD exponent?

As we have seen, for relations over arbitrary finite domains
the NRD exponent can be any fractional number $p/q \geq 1$,
with an infinite number of exponents occurring even for arity 3.
However, no such relation is known in the Boolean domain,
where up to arity 3 all relations have integral NRD exponents~\cite{chen2020BestCase,khanna2024Characterizations}.
Boolean relations are also more well-behaved than general languages
in other ways, e.g.,~\cite{chen2020BestCase,khanna2024Characterizations} showed that for every $R \subseteq \{0,1\}^r$,
either $R$ fgpp-defines $\OR_r$ and $\NRD(R,n)=\Theta(n^r)$,
or $\NRD(R,n)=\Oh(n^{r-1})$ (in fact, $R$ is captured as the roots of a polynomial of degree $r-1$). 
Again, this is false for general languages
even with $r=3$ and domain $|D|=3$, cf.~$R_6$ of Section~\ref{subsec:inf-ternary}.
There is also the aforementioned fact that all Boolean relations with
Mal'tsev embeddings have affine embeddings, which is again false over general domains.
Does the structural simplicity of Boolean languages stretch so far
that all Boolean languages have integral NRD exponents?
It seems unlikely, but we cannot exclude it at the moment.
A similar consequence was shown to hold for Boolean MaxCSP~\cite{JansenW24}.

One particularly interesting case is the Booleanization
BCK$_{\mathbb{B}} \subset \{0,1\}^9$ of the BCK predicate
(discussed in \cite{chen2020BestCase,lagerkvist2020Sparsification,khanna2024Characterizations,brakensiek2024Redundancy}). Using a SAT solver, an SMT solver~\cite{demoura2008Z3}, and Lemma~\ref{lm:ukc-basics}
we have verified that BCK$_{\mathbb{B}}$ is preserved by $u_k^c$
for every $k/c > 1$, hence there are no simple lower bounds
on its non-redundancy. However, since it fgpp-defines BCK
it also does not satisfy the Catalan conditions
and does not have a Mal'tsev embedding.

\paragraph{Algorithmic aspects of non-redundancy.}

An important open question raised in the work of Carbonnel~\cite{carbonnel2022Redundancy} is whether for every predicate $P$, every instance of $\CSP(P)$ on $n$ variables can be efficiently kernelized to $\Oh(\NRD(P, n))$ clauses. As observed by Brakensiek and Guruswami~\cite{brakensiek2024Redundancy}, answering such a question is a necessary step toward efficiently constructing optimal-sized sparsifiers for all CSPs. Note that Appendix~\ref{app:ker} resolves Carbonnel's conjecture in a special case.

An important observation to keep in mind is that \emph{local consistency} methods such as arc consistency~\cite{chen2006rendezvous,barto2017Polymorphisms}, the Sherali-Adams hierarchy, and sum of squares (i.e., the Lasserre hierarchy) are \emph{not} sufficient to resolve Carbonnel's open question. In particular, one cannot use a constant number of rounds of sum of squares to kernelize \emph{random} instances of many CSPs due to strong gaps 
\cite{grigoriev2001Linear,schoenebeck2008Linear,tulsiani2009CSP,kothari2017Sum} (see also Example~\ref{ex:r-LIN}). As such, any ``universal'' kernelization algorithm must combine techniques from both Mal'tsev and local consistency algorithms. In the study of the satisfiability of (promise) CSPs, understanding such a combination is a very active topic of research \cite{bulatov2017Dichotomy,zhuk2020Proof,brakensiek2020Powera,ciardo2023CLAP,zhuk2024simplifieda,lichter2024,Zhuk25}.

\paragraph{Algebra and non-redundancy.} Our algebraic approach is based on a generalized form of conjunctive formulas that allows  functional atoms of fixed arity $c$. We obtained a Galois correspondence to sets of polymorphism patterns which, curiously, turned out to be interesting algebraic objects in their own right since they can also be seen as a type of partial minion. Can strong partial minions $\pPol(S,T)$ and $\pattern(S,T)$ both be defined as a form of abstract partial minion, or are partially defined objects inherently too brittle to admit such a unifying treatment? 

It is worth remarking that $c$-fggpp-definitions are {\em orthogonal} to algebraic reductions from the CSP literature in the sense that they (1) allow for arbitrary guarding functions regardless of whether these are definable from the language itself but (2) are weaker in the sense that they do not allow arbitrary existential quantification. Nevertheless, this framework is algebraically well-behaved and fundamentally builds on the principle that patterns always produce functions which commute with each other. Thus, algebraically, is commutativity all you need, or are there algebraic constructions compatible with non-redundancy which break commutativity? 

\paragraph{Ternary predicates and Erdős box problem.}

As mentioned, previous works \cite{bessiere2020Chain,butti2020} gave a dichotomy theorem for the non-redundancy of binary predicates. In this paper, we extended their classification to all binary promise non-redundancy problems. A natural next step is to classify the non-redundancy of all ternary predicates. As we demonstrate in this paper, even for the non-redundancy of non-promise relations, infinitely many exponents are possible. As such, a complete classification seems unlikely in the near future, but we highlight a few interesting directions in the study of ternary predicates.

\begin{itemize}
\item Is there any subset of $[1,3]$ for which the non-redundancy exponents of (promise) ternary predicates are dense?

\item What is the non-redundancy of the $\CYCs_m$ predicates studied in this paper?

\item What is the non-redundancy of the $\tLIN^*_{G}$ predicates studied in \cite{brakensiek2024Redundancy}? Could the non-redundancy exponents even be irrational? Intriguingly, we currently know that $\NRD(\tLIN^*_{\F_3}, n) \in [\Omega(n^{1.5}), \widetilde{\Oh}(n^{1.6})],$ which is consistent with $\NRD(\tLIN^*_{\F_3}, n) = \Theta(n^{\log_2 3})$.

\item A result of Carbonnel~\cite{carbonnel2022Redundancy} states that if $P$ is a ternary predicate and $\NRD(P, n) = o(n^3)$, then $\NRD(P, n) = \Oh(\ex(n, K_{2,2,2}))$ which is notoriously known as the Erdős box problem~\cite{erdos1964extremal,katz2002Remarks,gordeev2024Combinatorial,furedi2013history,conlon2021Random}. We complement Carbonnel's result in Section~\ref{sec:hypergraph-turan} by showing that an infinite-domain ternary promise predicate has non-redundancy asymptotic to $\ex(n, K_{2,2,2})$. Does an analogous finite-domain predicate also have this property? Even more ambitiously, can we leverage the algebraic tools developed in this paper to improve the longstanding lower and upper bounds of $\Omega(n^{8/3})$ and $\Oh(n^{11/4})$, respectively for $\ex(n, K_{2,2,2})$?
\end{itemize}

With all these different directions of exploration, we believe the study of non-redundancy and its related questions will be the topic of active research for years to come.

\bibliographystyle{alpha}
\bibliography{ref}

\clearpage

\appendix

\section{Omitted Details from Section~\ref{sec:frac}}\label{app:frac}

\subsection{Connections to Kernelization}\label{app:ker}

Recall that that for $p \ge q$, we define $\SATDP_{p,q}$ to be $\{\ORDP_{p,q}, \CUT\}$ where $\ORDP_{p,q}$ is defined in \Cref{def:Rpq} and $\CUT = \{(0,1), (1,0)\}$. We now work toward showing that the kernelization exponent of $\CSP(\SATDP_{p,q})$ is $p/q$. 

\begin{lemma} \label{lem:lowerbound:transform}
For each fixed~$p,q \in \mathbb{N}$ with~$p \geq q$, there is a polynomial-time reduction that transforms any $n$-variable instance~$\phi$ of CNF-$p$-SAT into an equivalent instance~$(V,C)$ of $\csp(\SATDP_{p,q})$ with~$|V| = 2n + (2n)^q$ variables.
\end{lemma}
\begin{proof}
Consider an instance~$\phi$ of~CNF-$p$-SAT. It consists of a set of clauses~$A_1, \ldots, A_m$ over Boolean variables~$x_1, \ldots, x_n$, with each clause~$A_i$ being a disjunction of literals. Each literal is either a variable~$x_i$ or its negation~$\neg x_i$. 

We construct an equivalent instance of~$\csp(\SATDP_{p,q})$ in the following way. Let~$X := \{x_i, \overline{x}_i \mid i \in [n]\}$ and~$\widetilde{X} := \{\tilde{x}_\mathbf{x} \mid \mathbf{x} \in X^q\}$. So~$X$ contains one element for each possible literal of~$\phi$, while~$\widetilde{X}$ contains one element for each $q$-tuple~$\mathbf{x}$ of literals. We use~$V = X \cup \widetilde{X}$ as the set of variables for the CSP under construction. We will refer to the variables in~$X$ as the \emph{Boolean} variables and to those in~$\widetilde{X}$ as the \emph{padding} variables. Observe that~$|V| = 2n + (2n)^q$, as claimed. In the remainder, it will be convenient to let~$\mathsf{var}(\ell) \in X$ denote the variable in~$X$ corresponding to a literal~$\ell$. So if~$\ell = x_i$ is a positive literal, we have~$\mathsf{var}(\ell) = x_i$, while for negative literals~$\neg x_i$ we have~$\mathsf{var}(\ell) = \overline{x}_i$.

To ensure that the variables~$x_i$ and~$\overline{x}_i$ act as each others negations, we add the constraint~$\CUT(x_i, \overline{x}_i)$ to~$C$ for each~$i \in [n]$.

We transform clauses of~$\phi$ into constraints using the relation~$\ORDP_{p,q}$, as follows. For each clause~$A_j = \ell_1 \vee \ldots \vee \ell_p$, where each~$\ell_k$ is a literal of the form~$x_i$ or~$\neg x_i$, we add a corresponding constraint~$\ORDP_{p,q}(\mathbf{v}_j)$ to~$C$, where~$\mathbf{v}_j$ is a tuple of~$p + p^q$ elements of~$V$. The $i$th value of~$\mathbf{v}_j$, denoted~$\mathbf{v}_j[i]$, is defined as follows:
\begin{itemize}
    \item if~$i \in [p]$, then~$\mathbf{v}_j[i] = \mathsf{var}(\ell_i)$.
    \item if~$i > p$, then~$\mathbf{v}_j[i] = \tilde{x}_{\mathbf{x}}$, where~$\mathbf{x} \in X^q$ is the $q$-tuple of Boolean variables defined as follows. Let~$I := (i_1, \ldots, i_q) = \idx_{p,q}(i-p)$ be the tuple of~$q$ indices in the range~$[p]$ that is labeled~$i-p$ by the chosen bijection~$\idx_{p,q}$. Then we define:
    \begin{equation}
        \mathbf{x} = (\mathsf{var}(\ell_{i_1}), \ldots, \mathsf{var}(\ell_{i_q})).
    \end{equation}
\end{itemize}
So the tuple of variables to which the constraint associated to clause~$A_j$ is applied is computed as follows: the first~$p$ variables are Boolean variables corresponding to the literals occurring in~$A_j$; the remaining variables are padding variables. The padding variable~$\tilde{x}_{\mathbf{x}}$ used at position~$p + i$ is the one that was created for the $q$-tuple of literals that appear in clause~$A_j$ at the positions encoded by~$\idx_{p,q}(i)$.

This concludes the description of the~$\csp(\SATDP_{p,q})$-instance~$(V,C)$, which can easily be done in polynomial time for fixed~$p$ and~$q$. It remains to argue that the Boolean formula~$\phi$ is satisfiable if and only if~$(C,V)$ is satisfiable. We first argue the converse. If~$\sigma \colon V \to D_{p,q}$ is an assignment that satisfies all constraints of~$C$, then its restriction to the original variables~$\{x_1, \ldots, x_n\}$ satisfies all clauses of~$\phi$. To see this, consider an arbitrary clause~$A_j$ of~$\phi$. Since~$\sigma$ satisfies all constraints of~$C$, in particular it satisfies the constraint~$\ORDP_{p,q}(\mathbf{v}_j)$ that was inserted into~$C$ on account of clause~$A_j$. By construction, the first~$p$ variables of~$\mathbf{v}_j$ correspond to the literals in~$A_j$. By definition of~$A_j$, the constraint~$\ORDP_{p,q}(\mathbf{v}_j)$ can only be satisfied if the values of the first~$p$ variables of~$\mathbf{v}_j$ are Boolean, and at least one of them is~$1$. Hence there is a literal~$\ell_i$ in clause~$A_j$ such that the corresponding variable (either~$x_i$ or~$\overline{x}_i$) in~$X$ has value~$1$ under~$\sigma$. If~$\ell_i$ is positive, this directly shows that the restriction of~$\sigma$ to~$\{x_1, \ldots, x_n\}$ satisfies clause~$A_j$. If~$\ell_i$ is a negative literal of the form~$\neg x_k$, then we know~$\sigma(\overline{x}_k) = 1$. Since~$\sigma$ also satisfies the constraint~$(\CUT, (x_k, \overline{x}_k))$ that was inserted into~$C$, we find~$\sigma(x_k) = 0$. This again guarantees that the restriction of~$\sigma$ to~$\{x_1, \ldots, x_n\}$ satisfies~$A_j$. As this argument applies for an arbitrary clause~$C_j$, the entire formula~$\phi$ is indeed satisfied by the stated assignment.

To conclude the proof, we establish correctness in the converse direction: if~$\phi$ has a satisfying assignment~$f \colon \{x_1, \ldots, x_n\} \to \{0,1\}$, then~$(C,V)$ has a satisfying assignment~$\sigma \colon V \to D_{i,j}$. The assignment~$\sigma$ is derived from~$f$ as follows.
\begin{itemize}
    \item For each positive Boolean variable~$x_i \in V \cap X$ we set~$\sigma(x_i) = f(x_i)$.
    \item For each negative Boolean variable~$\overline{x}_i \in V \cap X$ we set~$\sigma(\overline{x}_i) = 1 - f(x_i)$.
    \item For each padding variable~$\tilde{x}_\mathbf{x}$ where~$\mathbf{x} = (x_{i_1}, \ldots, x_{i_q}) \in X^q$, the value~$\sigma(\tilde{x}_\mathbf{x})$ is the bitstring~$(f(x_{i_1}) \ldots f(x_{i_q}))$.
\end{itemize}
Hence the value of a padding variable~$\tilde{x}_{\mathbf{x}}$ that was created on account of a tuple of~$q$ Boolean variables is simply the bitstring obtained by concatenating the $f$-values of these Boolean values. From these definitions, it is straightforward to verify that~$\sigma$ satisfies~$(V,C)$. Since each clause of~$\phi$ is satisfied by~$f$, in each constraint of the form~$\ORDP_{p,q}(\mathbf{v}_j)$ created for a clause~$A_j$, one of the first~$p$ Boolean variables will be set to~$1$ since~$f$ satisfies a literal of~$A_j$. Since the values of the padding variables appearing in the rest of the constraint are defined precisely as the bitstrings of the values whose positions they control, the remaining conditions of \Cref{def:Rpq} are also satisfied. This concludes the proof of \Cref{lem:lowerbound:transform}.
\end{proof}

By combining the reduction from in the previous lemma with known kernelization lower bounds for CNF-$p$-SAT due to Dell and van Melkebeek~\cite{dell2014Satisfiability}, we obtain kernelization lower bounds for $\CSP(\SATDP_{p,q})$. We introduce some terminology that allows us to draw this connection. In the following, a \emph{parameterized problem}~$Q$ is a subset of $\Sigma^* \times \mathbb{N}_{+}$, where $\Sigma$ is a finite alphabet. When applying tools of parameterized complexity to our study of CSPs, we use the number of variables~$n$ in the CSP as the parameter~$k$ of the parameterized problem.

\begin{definition} \label{def:gen:kernel}
Let $Q, Q' \subseteq \Sigma^*\times\mathbb{N}_{+}$ be parameterized problems and let $h \colon \mathbb{N}_{+}\rightarrow\mathbb{N}_{+}$  be a computable function. A \emph{generalized kernelization for $Q$ into $Q'$ of size $h(k)$} is an algorithm that, on input $(x,k) \in \Sigma^*\times\mathbb{N}_{+}$, takes time polynomial in $|x|+k$ and outputs an instance $(x',k')$ such that:
\begin{enumerate}
\item $|x'|$ and $k'$ are bounded by $h(k)$, and
\item $(x',k')\in Q'$ if and only if $(x,k) \in Q$.
\end{enumerate}
The algorithm is a \emph{polynomial generalized kernelization} if $h(k)$ is a polynomial.
\end{definition}

The adjective \emph{generalized} corresponds to the fact that the algorithm is allowed to map instances from the source problem~$Q$ into equivalent instances of a different (but fixed) decision problem~$Q'$. Using this terminology, the lower bound by Dell and van Melkebeek can be stated as follows.

\begin{theorem}[\cite{dell2014Satisfiability}] \label{cnf:dell:lowerbound}
    For each integer~$p \geq 3$ and~$\varepsilon > 0$, the CNF-$p$-SAT problem parameterized by the number of variables~$n$ does not admit a generalized kernelization of size~$\Oh(n^{p-\varepsilon})$ unless \containment.
\end{theorem}

In this theorem, and in \Cref{def:gen:kernel}, the size of the instances is measured in terms of the number of symbols over alphabet~$\Sigma$, i.e., in terms of the number of bits needed to encode the instance. Our definition of the size of a kernelization for $\CSP(\Gamma)$ in the introduction was measuring the number of \emph{constraints} in the reduced instance. Notice that in our setting of CSPs over finite constraint languages consisting of finite relations, the difference between these two types of measurements is small: an instance that is encoded in~$f(n)$ bits has at most~$f(n)$ constraints, while an instance of~$f(n)$ constraints can be encoded in~$\Oh_\Gamma (f(n) \log n)$ bits by encoding, for each constraint, which relation~$R \in \Gamma$ it uses ($\Oh_\Gamma(1)$ bits) and to which variables the constraint is applied ($\Oh_R(1)$ variables, each encoded in~$\Oh(\log n)$ bits). Since our notion of kernelization exponent is oblivious to~$\log n$ factors, when considering the kernelization exponent of nonuniform CSPs it does not matter whether we measure kernelization size in bits or constraints. With this knowledge, we present the lower bound in the following lemma.

\begin{lemma} \label{lem:kernel:lowerbound}
Assuming \ncontainment, for each fixed~$p,q \in \mathbb{N}$ with~$p \geq q$ and~$p \geq 3$, the kernelization exponent of~$\csp(\SATDP_{p,q})$ is at least~$\frac{p}{q}$.
\end{lemma}
\begin{proof}
Assume for a contradiction that the kernelization exponent of~$\csp(\SATDP_{p,q})$ is less than~$\frac{p}{q}$. It follows that there exists a positive real~$\eps$ such that~$\csp(\SATDP_{p,q})$ has a kernelization that reduces to equisatisfiable instances with~$\Oh(n^{\frac{p}{q}-\eps})$ constraints.

By \Cref{cnf:dell:lowerbound}, to contradict the assumption \ncontainment it suffices to leverage \Cref{lem:lowerbound:transform} together with the assumed kernelization for $\csp(\SATDP_{p,q})$ into a generalized kernelization for CNF-$p$-SAT of bitsize~$\Oh(n^{p-\delta})$ for some~$\delta > 0$. This can be done as follows.

Given an $n$-variable instance~$\phi$ of CNF-$p$-SAT, the generalized kernelization transforms it into an equivalent instance~$(V,C)$ of~$\csp(\SATDP_{p,q})$ on~$|V| = 2n+(2n)^q$ variables by invoking \Cref{lem:lowerbound:transform}. It then applies the assumed kernelization to~$(V,C)$, resulting in an equivalent instance~$(V',C')$ on at most~$\Oh(|V|^{\frac{p}{q} - \eps}) \leq \Oh(n^q)^{\frac{p}{q} - \eps} \leq \Oh(n^{p - q \eps})$ constraints. Since variables not occurring in any constraint can be removed without changing the answer to the instance, while the arity of the constraints is~$\Oh(1)$ for fixed~$p$ and~$q$, we may assume without loss of generality that~$|V'| \leq \Oh(|C'|) \leq \Oh(n^{p - q \eps})$. Hence the identity of a single variable can be encoded in~$\log_2 \Oh(n^{p - q \eps}) \leq \Oh(\log_2 n)$ bits; here we use the fact that~$\log_2(n^{p - q\eps}) = (p - q\eps) \cdot \log_2 (n)$ while~$p \in \Oh(1)$. As the arity of each constraint is constant, each constraint of the instance~$(V',C')$ can be encoded in~$\Oh(\log n)$ bits by listing the binary representations of the indices of all the variables occurring in the constraint. Hence instance~$(V',C')$ can be encoded in~$\Oh(n^{p - q\eps} \log n)$ bits and is equivalent to the CNF-$p$-SAT instance~$\phi$. Since both transformation steps can be carried out in polynomial time, this results in a generalized kernelization for CNF-$p$-SAT of bitsize~$\Oh(n^{p - q\eps} \log n)$. Since~$\log n \in \Oh(n^\beta)$ for each~$\beta > 0$, this shows that CNF-$p$-SAT has a generalized kernelization of bitsize~$\Oh(n^{p - q\eps + \frac{q \eps}{2}})$, so of size~$\Oh(n^{p - \delta})$ for~$\delta = \frac{q \eps}{2} > 0$. But this implies \containment via the results of Dell and van Melkebeek~\cite{dell2014Satisfiability}.
\end{proof}

Note that while the statement of \Cref{lem:kernel:lowerbound} requires~$p \geq 3$, for any rational number~$\frac{p}{q} \geq 1$ we can write it as a ratio~$\frac{3p}{3q}$ of numbers that are at least three. We now prove the corresponding upper bounds. The manner in which we use the Kruskal-Katona theorem to bound the size of the instance after preprocessing is very similar to the proof of \Cref{lem:or-dp:upper}, but additional effort is needed to make the proof algorithmic and avoid the assumption of having multipartite instances.

\begin{lemma} \label{lem:kernel:upperbound}
For each fixed~$p,q \in \mathbb{N}$ with~$p \geq q$, there is a polynomial-time kernelization algorithm that transforms any $n$-variable instance~$(V,C)$ of $\csp(\SATDP_{p,q})$ into an equivalent instance~$(V',C')$ with~$|C'| \leq \Oh_{p,q}(n^{\frac{p}{q}})$. Hence the kernelization exponent of~$\CSP(\SATDP_{p,q})$ is at most~$\frac{p}{q}$.
\end{lemma}
\begin{proof}
Given an instance~$(V,C)$ of~$\csp(\SATDP_{p,q})$ on~$n := |V|$ variables, the kernelization consists of two phases aimed at reducing the number of constraints over~$\CUT$ and over~$\ORDP_{p,q}$, respectively. For the first phase, it is easy to reduce the number of constraints over~$\CUT$ to~$\Oh(n)$. Consider the graph~$G$ on vertex set~$V$ which has an edge for each constraint over~$\CUT$. If this graph is not bipartite, then since each edge corresponds to a~$\CUT$ constraint that requires the variables represented by its endpoints to have distinct Boolean values, the instance is unsatisfiable. It then suffices to output a constant-size unsatisfiable instance as the result of the kernelization. If~$G$ is bipartite, then we compute a spanning forest~$F$ of~$G$, which contains a spanning tree of each connected component. Remove all~$\CUT$ constraints that do not correspond to an edge in~$F$. We argue that this does not change the satisfiability of the instance, as follows. For each removed constraint~$\CUT(x, x')$, there is a path~$P$ between~$x$ and~$x'$ in the spanning forest~$F$. The path~$P$ has an odd number of edges, as otherwise it would make an odd cycle together with the edge~$\{x,x'\}$ which would violate the assumption that~$G$ is bipartite. As each edge on~$P$ corresponds to ~$\CUT$ constraint that we preserve in the instance, any satisfying assignment to the remaining constraints must alternate the assigned Boolean true/false value for each variable on path~$P$. As the number of edges on~$P$ is odd, to satisfy all cut constraints from~$P$ the value assigned to~$x$ must differ to that of~$x'$. Hence any omitted constraint~$\CUT(x,x')$ must be satisfied by any assignment that satisfies the remaining clauses. Since a spanning forest for a graph on~$|V|$ vertices has less than~$|V|$ edges and can be computed efficiently, the first phase of the kernelization can be executed in polynomial time and reduces the number of~$\CUT$ constraints to~$\Oh(n)$.

For the second phase of the algorithm, which is the most difficult, the goal is to reduce the number of~$\ORDP_{p,q}$ constraints. As an opening step for phase two, we partition the variables of the instance into types: \emph{Boolean} variables (that occur in a constraint over~$\CUT$, or that occur in one of the first~$p$ positions of a constraint over~$\ORDP_{p,q}$) and \emph{padding} variables (that occur in a constraint over~$\ORDP_{p,q}$ at a position after~$p$). If these two sets are not disjoint, then the choice of~$\ORDP_{p,q}$ ensures that even the subinstance consisting of two clauses that have this conflicting use of the variable, is unsatisfiable. Hence we may output such a subinstance as the result of the kernelization and terminate. In the remainder, we may assume the variables are partitioned into Boolean and padding variables.

To reduce the number of constraints of type~$\ORDP_{p,q}$ we apply two simple reduction rules based on the semantics of the~$\ORDP_{p,q}$ predicate, exploiting the fact that the values of padding variables in such constraints uniquely determine the values of the Boolean variables on the positions they govern. We exhaustively apply the following two reduction rules.
\begin{enumerate}
    \item Suppose there are indices~$i, i' \in [p^q]$ for which the instance contains two (not necessarily distinct) constraints
    \begin{align*}
    y =& \ORDP_{p,q}(x_1, \ldots, x_p, \tilde{x}_1, \ldots, \tilde{x}_{p^q}),\\
    y' =& \ORDP_{p,q}(x'_1, \ldots, x'_p, \tilde{x}'_1, \ldots, \tilde{x}'_{p^q}),
    \end{align*}
    such that~$\tilde{x}_i = \tilde{x}'_{i'}$ but for the associated tuples~$(t_1, \ldots, t_q) = \idx_{p,q}(i)$ and~$(t'_1, \ldots, t'_q) = \idx_{p,q}(i')$ we have~$(x_{t_1}, \ldots, x_{t_j}) \neq (x'_{t'_1}, \ldots, x'_{t'_j})$. In other words, there are two constraints that use the same padding variable~$\tilde{x}_i = \tilde{x}'_{i'}$, but the tuples of Boolean variables they control differ. Let~$j \in [q]$ such that~$x_{t_j} \neq x'_{t'_j}$. We now simplify the instance as follows: replace all occurrences of variable~$x'_{t'_j}$ in the CSP by~$x_{t_j}$ and remove variable~$x'_{t'_j}$. It is clear that any satisfying assignment of the reduced instance yields a satisfying assignment to the original, by simply assigning the eliminated variable~$x'_{t'_j}$ the same value as~$x_{t_j}$. Due to the semantics of~$\ORDP_{p,q}$ (Property~\ref{constraint:dp} in \Cref{def:Rpq}) the other direction holds as well: any satisfying assignment to the original instance must assign~$x_{t_j}$ and~$x'_{t'_j}$ the same value, since the value of both variables must match the $j$-th bit in the bitstring assigned as the value for~$\tilde{x}_i = \tilde{x}'_{i'}$. Hence this replacement preserves the satisfiability of the instance while strictly reducing the number of variables.

    Intuitively, this reduction rule ensures that each $q$-tuple of Boolean variables that occurs in some constraint has a \emph{unique} padding variable associated to it, which we will exploit in the analysis below.
    \item Suppose there is an index~$i \in [p^q]$ for which the instance contains two constraints 
    \begin{align*}
    y =& \ORDP_{p,q}(x_1, \ldots, x_p, \tilde{x}_1, \ldots, \tilde{x}_{p^q}),\\
    y' =& \ORDP_{p,q}(x'_1, \ldots, x'_p, \tilde{x}'_1, \ldots, \tilde{x}'_{p^q}),
    \end{align*}
    in which the padding variables~$\tilde{x}_i, \tilde{x}'_i$ used on the $i$-th padding position are distinct, but in which the tuples of Boolean variables appearing at positions~$(t_1, \ldots, t_q) := \idx_{p,q}(i)$ are identical, so that~$\tilde{x}_i \neq \tilde{x}'_i$ while~$(x_{t_1}, \ldots, x_{t_q}) = (x'_{t_1}, \ldots, x'_{t_q})$. Then we do the following: replace all occurrences of padding variable~$\tilde{x}'_i$ by padding variable~$\tilde{x}_i$ and then remove~$\tilde{x}'$ from the instance. Due to the semantics of~$\ORDP_{p,q}$ (Property~\ref{constraint:dp} in \Cref{def:Rpq}), any satisfying assignment to the original instance must give~$\tilde{x}_i$ and~$\tilde{x}'_i$ the same value since the value of both variables is the bitstring of values assigned to the~$q$ Boolean variables~$(x_{t_1}, \ldots, x_{t_q}) = (x'_{t_1}, \ldots, x'_{t_q})$. Hence this replacement preserves the satisfiability of the instance. 
    
    Exhaustive application of this reduction rule ensures that for each $p$-tuple of Boolean variables, there is at most one constraint over~$\ORDP_{p,q}$ in which that exact $p$-tuple appears on the first~$p$ positions.
\end{enumerate}
The instance~$(V',C')$ that is obtained from exhaustively applying these reduction rules is returned as the output of the kernelization algorithm. Since all preprocessing steps preserve the satisfiability of the instance, it is equisatisfiable to~$(V,C)$. As the reduction rules can easily be implemented to run in polynomial time, it remains to argue that the number of constraints~$|C'|$ in the reduced instance is bounded by~$\Oh_{p,q}(n^{\frac{p}{q}})$. Since the number of constraints over~$\CUT$ has been reduced to~$\Oh(n)$ by the first phase, it suffices to bound the constraints~$C^* \subseteq C'$ over~$\ORDP_{p,q}$. Note that~$|V'| \leq |V| = n$ since we never introduce new variables. Now we set up an application of the Kruskal-Katona theorem to bound the number of constraints~$|C^*|$ in terms of the number of variables~$|V'|$. 

The set family~$\mathcal{A}$ to which we apply the Kruskal-Katona theorem is defined over the universe~$V' \times [p]$. For each constraint~$y = \ORDP_{p,q}(x_1, \ldots, x_p, \tilde{x}_1, \ldots, \tilde{x}_{p^q})$ in~$C^*$, we add a set~$A_y = \{ (x_1, 1), (x_2, 2), \ldots, (x_p, p)\}$ to~$\mathcal{A}$. Since all pairs in the set have different second coordinates, the resulting set family~$\mathcal{A}$ is $p$-uniform. Let~$\mathcal{B}$ be the $q$-shadow of~$\mathcal{A}$, i.e., it contains each size-$q$ subset~$B \subseteq V' \times [p]$ for which~$B$ is contained in some set~$A_y$ of~$\mathcal{A}$. Next, we relate~$|\mathcal{A}|$ to~$|C^*|$ and~$|\mathcal{B}|$ to~$|V'|$.

Observe first that~$|C^*| = |\mathcal{A}|$. This follows from the fact that each constraint~$y \in C^*$ contributes a set~$A_y$ to~$\mathcal{A}$, while no two distinct constraints can contribute the same set: that would mean there are two distinct constraints over~$\ORDP_{p,q}$ in which the \emph{same} $p$-tuple of variables appears on the first~$p$ positions, which triggers the second reduction rule.

Next, we will argue that~$|\mathcal{B}| \leq |V'| \cdot p^q$ by relating sets in the $q$-shadow to padding variables. To establish this we provide an injection~$f \colon \mathcal{B} \to V' \times [p^q]$, where the Cartesian product with~$[p^q]$ plays the role of the reduction to a multipartite instance in \Cref{lem:or-dp:upper}. The injection~$f$ is defined as follows. Consider a $q$-set~$B = \{(x_{t_1}, t_1), \ldots, (x_{t_q}, t_q)\} \in \mathcal{B}$ and label its elements such that~$t_1 < \ldots < t_q$. Let~$S_y \in \mathcal{A}$ be a set that contains~$B$ and consider the constraint~$y = \ORDP_{p,q}(x'_1, \ldots, x'_p, \tilde{x}'_1, \ldots, \tilde{x}'_{p^q}) \in C^*$ for which~$S_y$ was added to~$\mathcal{A}$. Since~$B$ is a subset of~$S_y$, while the pairs over~$V' \times [p]$ encode the positions at which Boolean variables appear in constraints of~$C^*$, we have~$x_{t_j} = x'_{t_j}$ for each~$j \in [q]$: for each pair~$(x,j) \in B$ consisting of a Boolean variable~$x$ and a position~$j \in [p]$, the $j$-th Boolean variable in constraint~$y$ is~$x$. Let~$i := \idx_{p,q}^{-1}(t_1, \ldots, t_q) \in [p^q]$ be the index of the $q$-tuple of positions represented by~$B$. Then we define~$f(B) := (\tilde{x}_{i}, i)$, that is, we map set~$B$ to the pair consisting of~$i$ and the padding variable that occurs in constraint~$y$ at the position~$i$ that controls the $q$-tuple of variables in~$B$.

Next, we argue that~$f$ is indeed an injection. Suppose for a contradiction that two distinct sets~$B,B' \in \mathcal{B}$ map to the same pair~$(\tilde{x}, i)$; let these be $B = \{(x_{t_1}, t_1), \ldots, (x_{t_q}, t_q)\}$ and~$B' = \{(x'_{t'_1}, t'_1), \ldots, (x'_{t'_q}, t'_q)\}$ labeled with~$t_1 < \ldots < t_q$ and~$t'_1 < \ldots < t'_q$. Since~$f(B) = f(B') = (\tilde{x}, i)$, they agree on their second coordinate which implies that~$i = \idx_{p,q}(t_1, \ldots, t_q) = \idx_{p,q}(t'_1, \ldots, t'_q)$. Hence the $q$-sets~$B, B'$ are obtained from $p$-sets~$A_y, A_{y'}$ by selecting the Boolean variables from the \emph{same} positions~$(t_1, \ldots, t_q) = (t'_1, \ldots, t'_q)$ of the respective constraints~$y, y'$. Since the sets~$B, B'$ are distinct while their projection onto the second element of the pairs yields the same set~$\{t_1, \ldots, t_q\} = \{t'_1, \ldots, t'_q\}$ of indices, there is an index~$j \in [q]$ for which~$(x_{t_j}, t_j) \neq (x'_{t'_j}, t'_j)$ with~$t_j = t'_j$. Consequently,~$x_{t_j} \neq x'_{t'_j}$ are \emph{distinct} Boolean variables in clauses~$y, y'$, that appear on a position~$t_j$ controlled by a padding variable with index~$i$, and the two constraints~$y,y'$ use the \emph{same} padding variable~$\tilde{x}$ at position~$i$. But this contradicts the assumption that the first reduction rule is exhaustively applied. We conclude that~$f$ is indeed an injection.

We now use these bounds to finish the proof. As in the proof of \Cref{lem:or-dp:upper}, let~$c_{p,q} \geq 1$ be any constant such that~$\binom{m}{p} \leq c_{p,q} \cdot \binom{m}{q}^{\frac{p}{q}}$. Since~$f$ is an injection from~$\mathcal{B}$ to~$V' \times [p^q]$ we infer~$|\mathcal{B}| \leq |V'| \cdot p^q$. Let~$m \geq p$ be the largest integer such that~$|\mathcal{A}| \geq \binom{m}{p}$, so that~$|\mathcal{A}| < \binom{m+1}{p} \leq 2p \binom{m}{p}$. By the Kruskal-Katona theorem (\Cref{kruskal-katona}) we have~$|V'| \cdot p^q \geq |\mathcal{B}| \geq \binom{m}{q}$. Hence we can now derive:
\begin{align*}
    |C^*| = |\mathcal{A}| \leq 2p \cdot \binom{m}{p} \leq 2p \cdot c_{p,q} \cdot \binom{m}{q}^{\frac{p}{q}} \leq 2p \cdot c_{p,q} \cdot (|V'| \cdot p^q)^{\frac{p}{q}} \leq \Oh_{p,q}(n^{\frac{p}{q}}).
\end{align*}
As this directly bounds the kernelization exponent of~$\CSP(\SATDP_{p,q})$, it concludes the proof of \cref{lem:kernel:upperbound}.
\end{proof}

By combining the upper- and lower bounds, we now derive the main result of this section.

\begin{theorem}
    Let~$p \geq q \in \N$ such that~$p \geq 3$. Assuming \ncontainment, the kernelization exponent of~$\CSP(\SATDP_{p,q})$ is~$\frac{p}{q}$. Hence every rational number of value at least one is the kernelization exponent of some CSP.
\end{theorem}
\begin{proof}
    For~$p \geq q \in \N$ with~$p \geq 3$, \Cref{lem:kernel:lowerbound} shows that under the assumption~\ncontainment the kernelization exponent of~$\CSP(\SATDP_{p,q})$ must be at least~$\frac{p}{q}$, while~\Cref{lem:kernel:upperbound} gives a matching upper bound. Hence the kernelization exponent is~$\frac{p}{q}$. As every fraction~$\frac{p'}{q'} \geq 1$ of integers can be written as~$\frac{3p'}{3q'}$ satisfying this premise, every rational at least one is the kernelization exponent of some CSP.
\end{proof}

\subsection{A Construction of Lower Arity}\label{app:frac-alt}

In this portion of the appendix, we prove Theorem~\ref{thm:frac-nrd-alt}. To help with understanding the main proof technique, we first prove a special case.

\subsubsection{Warm-up Predicate}\label{app:warmup}

As a warmup, we consider $p=3$, $q=2$ and $\cF = \{\{1,2\},\{2,3\},\{3,1\}\}$. It is clear that $\cF$ is $2/3$-regular. In that case, we have that the domain $D := E_{3,2}$ is $\{00,01,10,11\}$ and the relation $\DP_{\cF}$ being
\begin{align*}
    \DP_{\cF} := \{&(00,00,00),(00,01,10),(01,10,00),(01,11,10),\\
    &(10,00,01),(10,01,11),(11,10,01),(11,11,11)\}.
\end{align*}
With $\ORDP_{\cF}$ being the same relation except we omit $(00,00,00)$.

\begin{remark}
An fgpp-equivalent formulation of $\ORDP_{\cF}$ is $D := \{0,1,2,3\}$ and \[R := \{012, 120, 201, 321, 213, 132, 333\}.\]
\end{remark}

Using an argument similar to that of \Cref{lem:nrd:or-dp-star:linear}, we can prove that $\NRD(\DP_{\cF}, n) = \Oh(n)$. As such, it suffices understand the asymptotics of $\NRD(\ORDP_{\cF} \mid \DP_{\cF}, n)$. We begin with a lower bound.

\begin{lemma}\label{lem:lb}
$\NRD(\ORDP_{\cF} \mid \DP_{\cF}, n) \ge \Omega(n^{1.5})$.
\end{lemma}

\begin{proof}
Assume that $n = 3t^2$. Let $X := [3] \times [t]^2$. Define $Y \subseteq X^3$ as follows
\[
    Y := \{e_{a,b,c} := ((1,a,b),(2,b,c),(3,c,a)) \in X^3 : a,b,c \in [t]\}.
\]
Obviously, $|Y| = t^3$, so it suffices to prove that $(X,Y)$ is non-redundant. Fix $(a,b,c) \in [t]^3$ and define $\sigma_{a,b,c} : X \to D$ as follows
\begin{align*}
    \sigma_{a,b,c}(1,a',b') := \begin{cases}
    00 & a=a' \wedge b=b'\\
    01 & a=a' \wedge b\neq b'\\
    10 & a\neq a' \wedge b=b'\\
    11 & a\neq a' \wedge b\neq b'.
    \end{cases}\\
    \sigma_{a,b,c}(2,b',c') := \begin{cases}
    00 & b=b' \wedge c=c'\\
    01 & b=b' \wedge c\neq c'\\
    10 & b\neq b' \wedge c=c'\\
    11 & b\neq b' \wedge c\neq c'.
    \end{cases}\\
    \sigma_{a,b,c}(3,c',a') := \begin{cases}
    00 & c=c' \wedge a=a'\\
    01 & c=c' \wedge a\neq a'\\
    10 & c\neq c' \wedge a=a'\\
    11 & c\neq c' \wedge a\neq a'.
    \end{cases}
\end{align*}
One can check that for every $(a',b',c') \in [t]^3$ we have that $\sigma_{a,b,c}(e_{a',b',c'}) \in \DP_{\cF}$ and $\sigma_{a,b,c}(e_{a',b',c'}) = (00,00,00)$ if and only if $(a,b,c) = (a',b',c')$.
\end{proof}

Now we prove an essentially matching upper bound.

\begin{lemma}\label{lem:ub}
$\NRD(\ORDP_{\cF} \mid \DP_{\cF}, n) = \Oh(n^{1.5})$.
\end{lemma}

\begin{proof}
Let $(X, Y)$ be a non-redundant instance. By \Cref{lemma:reduction:to:multipartite}, we may assume that the instance is tripartite. That is, $X := X_1 \sqcup X_2 \sqcup X_3$ with $Y \subseteq X_1 \times X_2 \times X_3$. For each $y \in Y$, let $\sigma_y : X \to D$ be an assignment such that $\sigma_y(y) = (00,00,00)$ and $\sigma_y(y') \in \ORDP_{\cF}$ for all $y' \in Y \setminus \{y\}$. Since no $t \in \ORDP_{\cF}$ has two coordinates equal to $00$, we must have that $Y$ is linear (i.e., any two clauses have at most one variable in common).

For distinct $i,j \in [3]$, let $Y_{ij} := \{(y_i, y_j) : y \in Y\} \subseteq X_i \times X_j.$  Since $Y$ is linear, we have that each $Y_{ij}$ has the same cardinality as $Y$. We say that $x, x' \in X_i$ are $j$-equivalent, denoted by $x \sim_j x'$ if $x$ and $x'$ are in the same connected component of $Y_{ij}$. A crucial observation is as follows.
\begin{claim}\label{claim:merge}
Let $x, x' \in X_i$ be such that $x \sim_{i+1} x'$ and $x \sim_{i+2} x'$, then for every $\DP_{\cF}$-satisfying assignment $\sigma : X \to D$ for $Y$, we have that $\sigma(x) = \sigma(x')$.
\end{claim}
\begin{proof}
Note that for any $i \in [3]$, we have that $\DP_{\cF}|_{i,i+1} = \{(00,00),(00,01),(10,00),(10,01)\} \cup$ $\{(01,10),(01,11),(11,10),(11,11)\}$, where indices are taken modulo $3$.
Thus, if $x, x' \in X_i$ satisfy $x \sim_{i+1} x'$, then $(\sigma(x), \sigma(x')) \in \{00,10\}^2 \cup \{01,11\}^2$ for any $\DP_{\cF}$-satisfying assignment $\sigma$.

Likewise, since $\DP_{\cF}|_{i,i+2} := \{(00,00),(00,10),(01,00),(01,10)\} \cup$ $ \{(10,01),(10,11), (11,01),$ $(11,11)\}$, so any $x, x' \in X_i$ satisfying $x \sim_{i+2} x'$ must have $(\sigma(x), \sigma(x')) \in \{00,01\}^2 \cup \{10,11\}^2$. Therefore, if $x \sim_{i+1} x'$ and $x \sim_{i+2} x'$, then
\[
    (\sigma(x), \sigma(x')) \in (\{00,10\}^2 \cup \{01,11\}^2) \cap (\{00,01\}^2 \cup \{10,11\}^2) = \{(00,00),(01,01),(10,10),(11,11)\},
\]
as desired.
\end{proof}

As such, if for some $i \in [3]$ and distinct $x, x' \in X_i$, we have that $x \sim_{i+1} x'$ and $x \sim_{i+2} x'$, then we can replace $X$ with $X \setminus \{x'\}$ and replace every $y \in Y$ with $y_i = x'$ with a corresponding $y'$ with $y'_i = x$ while preserving non-redundancy. Thus, assume without loss of generality that for every distinct $x, x' \in X_i$, we have that either $x \not\sim_{i+1} x'$ or $x \not\sim_{i+2} x'$.

For each distinct $i,j \in [3]$, let $\ell_{ij}$ be the number of connected components of $Y_{ij}$. Let $\alpha_{ij} : X_i \sqcup X_j \to [\ell_{ij}]$ be an indicator map of these connected components. By our WLOG assumption, we have that each $x \in X_i$ is uniquely determined by $(\alpha_{i,i+1}(x), \alpha_{i,i+2}(x))$. In particular, this means that each edge $y \in Y$ is uniquely determined by
\[
(\alpha_{1,2}(y_1) = \alpha_{1,2}(y_2), \alpha_{1,3}(y_1) = \alpha_{1,3}(y_3), \alpha_{2,3}(y_2) = \alpha_{2,3}(y_3)).
\]

In other words, we can identify $Y$ with $\widetilde{Y} \subseteq [\ell_{12}] \times [\ell_{13}] \times [\ell_{23}]$ and each $X_i$ with $\widetilde{X}_i \subseteq [\ell_{i,i+1}] \times [\ell_{i,i+2}]$ corresponding to the projection of $\widetilde{Y}$ onto $[\ell_{i,i+1}] \times [\ell_{i,i+2}]$. By Shearer's inequality\footnote{This precise version of Shearer's inequality appeared as Problem 5 at the 1992 International Mathematical Olympiad~\cite{djukic2006imo}.} \cite{chung1986intersection}, we then have that
\[
    |Y| = |\widetilde{Y}| \le \sqrt{|\widetilde{X}_1| \cdot |\widetilde{X}_2| \cdot |\widetilde{X}_3|} = \sqrt{|X_1|\cdot |X_2|\cdot |X_3|} \le n^{1.5}. \qedhere
\]
\end{proof}

\begin{remark}
This proof in fact shows that the construction in Lemma~\ref{lem:lb} is the optimal-sized tripartite instance for infinitely many choice of $n$.
\end{remark}

As an immediate corollary of Lemma~\ref{lem:lb} and Lemma~\ref{lem:ub}.

\begin{theorem}
$\NRD(\ORDP_{\cF}, n) = \Theta(n^{1.5})$.
\end{theorem}

\subsubsection{Every Rational Number}

We now work toward proving Theorem~\ref{thm:frac-nrd-alt} in general. By adapting the proof of \Cref{lem:nrd:or-dp-star:linear}, we have that $\NRD(\DP_{\cF}, n) = \Oh(n)$, so it suffices to bound $\NRD(\ORDP_{\cF} \mid \DP_{\cF}, n)$. We begin with a lower bound.

\begin{lemma}\label{lem:or-lb}
$\NRD(\ORDP_{\cF} \mid \DP_{\cF}, n) = \Omega_{\cF}(n^{p/q})$.
\end{lemma}

\begin{remark}
Lemma~\ref{lem:or-lb} does not need $\cF$ to be regular. In fact, the only assumption needed is that every $i \in [p]$ is contained in some $S \in \cF$.
\end{remark}

\begin{proof}
Assume that $n = |\cF|t^q$. Let $X := \cF \times [t]^q$, and define $Y \subseteq X^{\cF}$ as
\[
    Y := \{e_{z} := ((S, z|_{S}) : S \in \cF) \in X^{\cF}  : z \in [t]^p\}
\]
Now fix, $z \in [t]^p$, we define an assignment $\sigma_z : X \to D$ as follows. For each $S \in \cF$, $w \in [t]^S$ and $i \in S$, we have that
\[
    \sigma_z((S,w))_i := \one[w_i \neq z_i].
\]
First, note that for all $z' \in [t]^p$, we have that $\sigma_z(e_{z'}) \in \DP_{\cF}$. Further, $\sigma_z(e_{z'}) = \proj_{\cF}(0^p)$ if and only if $z = z'$. Thus, $\{\sigma_z : z \in [t]^q\}$ witness that $(X,Y)$ is a non-redundant instance with $O_{\cF}(t^q)$ variables and $t^p$ clauses.
\end{proof}

We now prove a corresponding upper bound for $\NRD(\ORDP_{\cF} \mid \DP_{\cF}, n)$.

\begin{lemma}\label{lem:or-ub}
If $\cF$ is $p/q$-regular then $\NRD(\ORDP_{\cF} \mid \DP_{\cF}, n) = \Oh(n^{p/q})$.
\end{lemma}

\begin{proof}
By \Cref{lemma:reduction:to:multipartite}, we may assume without loss of generality that our non-redundant instance $(X,Y \subseteq X^{\cF}, \{\sigma_y : X \to D \mid y \in Y\})$ is $\cF$-partite. That is, $X$ is the disjoint union of variables sets $X_{S}$ for each $S \in \cF$.

For each $i \in [p]$ define $\cF_{i} := \{S \in \cF : i \in S\}$ and $X_{\cF_i} := \bigcup_{S \in \cF_i} X_S$. Further define $Y_i := \{y|_{\cF_{i}} : y \in Y\} \subseteq X_{\cF_i}^{\cF_i}$. It may be the case that $|Y_i| < |Y|$. Let $\cC_i$ be the set of connected components of $Y_i$ (viewed as a hypergraph on vertex set $X_{\cF_i}$). A key observation is that for every satisfying assignment $\sigma : X \to D$ to $\CSP(\DP_{\cF})$, and any $y, y' \in Y$ for which $y|_{\cF_i}$ and $y'|_{\cF_{i}}$ are in the same connected component of $Y_i$, we have that 
\[\proj_{\cF}^{-1}(\sigma(y))_i = \proj_{\cF}^{-1}(\sigma(y'))_{i}.\]
The reason why is that for any $y \in Y$ and any $S \in \cF_i$, $\proj_{\cF}^{-1}(\sigma(y))_i = \sigma(y_S)_i$, so $\proj^{-1}_{\cF}\circ \sigma$ must be constant on the hyperedges of a connected component of $Y_i$. Thus, for any $y \in Y$ that the list of connected components it resides in $\cC_1, \hdots, \cC_p$ uniquely identifies $y$. 

Let $c_i : X_{\cF_i} \to \cC_i$ be the connected component identification maps. Now fix $S \in \cF$. Consider any $x, x' \in X_S$ such that for every $i \in S$, $x$ and $x'$ are in the same connected component of $Y_i$. By the previous logic, we must have for every satisfying assignment $\sigma : X \to D$ to $\CSP(\DP_{\cF})$, we have that $\sigma(x) = \sigma(x')$. Thus, without loss of generality, we may replace every instance of $x'$ in some $y \in Y$ with $x$ without changing the non-redundancy of the instance. In other words, we may assume that $x \mapsto (c_i(x) : x \in S)$ is an injection from $X_S$ to $\prod_{i \in S}\cC_i$. 

Now, let $Z$ be the random variable on $\cC_1 \times \cdots \times \cC_p$ corresponding to the list of connected components of a uniformly random sample of $y \in Y$. Let $\cP$ be the $q/p$-regular probability
distribution over $\cF$. By Shearer's inequality, we have that
\[
    \log_2 |Y| = H(Z) \le \frac{p}{q} \cdot \underset{S \sim \cP}{\E}[H(Z|_{S})] \le \frac{p}{q} \cdot \underset{S \sim \cP}{\E}[\log_2 |X_S|] \le \frac{p}{q}\log_2 |X|.
\]
Thus, $|Y| \le |X|^{p/q}$, as desired.
\end{proof}

Combining the lemmas of this section with the triangle inequality, we have proved \Cref{thm:frac-nrd-alt}.

\end{document}